\documentclass[10pt,letterpaper]{article}
\usepackage[top=0.85in,left=2.75in,footskip=0.75in]{geometry}

\usepackage{amsmath,amssymb}

\usepackage{changepage}

\usepackage{textcomp,marvosym}

\usepackage{cite}

\usepackage{nameref,hyperref}

\usepackage[right]{lineno}

\usepackage[nopatch=eqnum]{microtype}
\DisableLigatures[f]{encoding = *, family = * }

\usepackage[table]{xcolor}

\usepackage{array}

\newcolumntype{+}{!{\vrule width 2pt}}

\newlength\savedwidth

\usepackage[aboveskip=1pt,labelfont=bf,labelsep=period,singlelinecheck=off]{caption}

\makeatletter
\renewcommand{\@biblabel}[1]{\quad#1.}
\makeatother

\usepackage{lastpage,fancyhdr,graphicx}
\usepackage{epstopdf}
\fancyheadoffset[L]{2.25in}
\fancyfootoffset[L]{2.25in}
\usepackage{amsthm}
\usepackage{enumitem}
\usepackage{comment}
\usepackage{todonotes}
\reversemarginpar
\usepackage{tikz}
\usepackage{thm-restate}

\usetikzlibrary{calc}
\usetikzlibrary{decorations.markings}
\usetikzlibrary{arrows.meta, shapes.misc, positioning}
\usetikzlibrary{patterns,patterns.meta}

\usepackage{siunitx}
\usepackage{subcaption}
\usepackage{booktabs}

\newcommand{\pr}{\mathbb{P}}
\renewcommand{\Pr}{\mathbb{P}}
\newcommand{\E}{\mathbb{E}}
\newcommand{\var}{\textnormal{Var}}
\newcommand{\cov}{\textnormal{Cov}}

\newcommand{\N}{\mathbb{N}}

\newcommand{\Vol}{\ensuremath{\text{Vol}}}

\newtheorem{theorem}{Theorem}[section]
\newtheorem{lemma}[theorem]{Lemma}

\newtheorem{proposition}[theorem]{Proposition}

\newtheorem{definition}[theorem]{Definition}
\theoremstyle{definition}

\begin{document}
\vspace*{0.2in}

\begin{flushleft}
{\Large
\textbf\newline{Assortativity in geometric and scale-free networks} 
}
\newline
\\
Marc Kaufmann\textsuperscript{1}\Yinyang,
Ulysse Schaller\textsuperscript{1}\Yinyang,
Thomas Bläsius\textsuperscript{2}\ddag,
Johannes Lengler\textsuperscript{2,*}\ddag
\\
\bigskip
\textbf{1} Institut für Theoretische Informatik, ETH Zürich, Zürich, Switzerland
\\
\textbf{2} Institut für Theoretische Informatik, Karlsruhe Institute of Technology, Karlsruhe, Germany
\\
\bigskip

\Yinyang These authors contributed equally to this work.

\ddag Shared senior authorship.

* johannes.lengler@inf.ethz.ch

\end{flushleft}

\section*{Abstract}
The assortative behavior of a network is the tendency of similar (or dissimilar) nodes to connect to each other. This tendency can have an influence on various properties of the network, such as its robustness or the dynamics of spreading processes. In this paper, we study degree assortativity both in real-world networks and in several generative models for networks with heavy-tailed degree distribution based on latent spaces. In particular, we study Chung-Lu Graphs and Geometric Inhomogeneous Random Graphs (GIRGs). 

Previous research on assortativity has primarily focused on measuring the degree assortativity in real-world networks using the Pearson assortativity coefficient, despite reservations against this coefficient. We rigorously confirm these reservations by mathematically proving that the Pearson assortativity coefficient does not measure assortativity in any network with sufficiently heavy-tailed degree distributions, which is typical for real-world networks. Moreover, we find that other single-valued assortativity coefficients also do not sufficiently capture the wiring preferences of nodes, which often vary greatly by node degree. We therefore take a more fine-grained approach, analyzing a wide range of conditional and joint weight and degree distributions of connected nodes, both numerically in real-world networks and mathematically in the generative graph models. We provide several methods of visualizing the results.

We show that the generative models are assortativity-neutral, while many real-world networks are not. Therefore, we also propose an extension of the GIRG model which retains the manifold desirable properties induced by the degree distribution and the latent space, but also exhibits tunable assortativity. We analyze the resulting model mathematically, and give a fine-grained quantification of its assortativity.

\section*{Author summary}

The degree of a node in a network, that is, the number of other nodes it is connected to, is a simple measure of ``importance'' within the network. Whether nodes connect predominantly to similarly important ``peers'', or whether there exist hierarchies where less important nodes connect to much more important nodes, has far-reaching consequences for processes that occur within networks as well as for the underlying network structure. This property is often called ``assortativity'' and typically measured by computing a single numerical value for a network. A lot of information is lost in this process and, to make matters worse, the most widespread way of computing this value has severe statistical flaws. We provide new evidence of these flaws and instead propose a local approach to measuring assortativity, which studies how the degree distribution of a node's neighbors changes with this node's degree. We further propose a ``tunable'' model, which allows to adjust the wiring preferences of nodes based on the degrees of potential neighbors, while at the same time capturing many established structural properties of real networks.
We evaluate our new assortativity measure both on real-world networks and theoretical network models including our new tunable models. 

\section{Introduction}\label{sec:intro}

The study of network topologies has been a focal point of network science since its inception. A central feature of social and other networks is the \emph{assortativity} or \emph{homophily} of its nodes.\footnote{In keeping with the common naming conventions in graph theory, in the more technical sections we often use the terms \textit{node} and \textit{vertex} interchangeably.} This is the tendency of nodes to form connections with other nodes that have similar properties. For example, a member of a social network will have many connections with people who share profession, hobbies, or kinship with her. 

While this feature is interesting in its own right and has been the subject of intensive studies~\cite{newman2002assortative, newman2003mixing,dagostino2012robustness,zhou2012robustness,duh2019assortativity,bialonski2013assortative, newman2003structure,chang2007assortativity}, it can also be utilized in order to construct artificial networks whose structure mirrors the structure of a large class of real-world networks. The idea is to use a latent space in which the nodes are located. The purpose of this latent space is to model hidden features of the nodes, like interests or profession for social networks, but also geographical distance. Distances in the latent space correspond to similarity of the nodes. Each node also randomly draws an (approximate) degree, which we will refer to as \emph{weight}. This is drawn from a heavy-tailed distribution as they are observed in social networks \cite{barabasi1999emergence}, typically from a power law. Then, nodes form connections, typically referred to as \textit{edges}, based on their distances and their weights. Sometimes, the models have a \emph{temperature} parameter which also allows them to form \emph{weak ties} \cite{granovetter1973strength}, i.e., nodes can randomly form edges even if not mandated by their distance and weights. This paradigm for generating artificial complex networks has been introduced and re-discovered by many different research communities, under names like \emph{Geographical Threshold Models}~\cite{masuda2005geographical}, \emph{Hyperbolic Random Graphs}~\cite{krioukov2010hyperbolic}, \emph{Scale-Free Percolation}~\cite{deijfen2013scale}, and \emph{Geometric Inhomogeneous Random Graphs} (GIRGs)~\cite{bringmann2019geometric, bringmann2024average} and its variations like MCD-GIRGs\cite{lengler2017existence}. 

The resulting complex network models have been found to be very expressive. They do not only faithfully reproduce the structural features that are explicitly built in, like sparsity and a power-law degree distribution. They also show a large number of other emerging structural properties that have also been observed in real-world networks. Those include the component structure consisting of one giant connected component and many very small connected components~\cite{bringmann2024average}, ultra-small graph distances~\cite{bringmann2024average} (the \emph{small world property}~\cite{bringmann2024average, watts1998collective}), high clustering coefficients~\cite{bringmann2024average}, closeness and
betweenness centrality~\cite{dayan2024expressivity}, small vertex and edge separators\cite{lengler2017existence, kaufmann2024sublinear}, small entropy per edge~\cite{bringmann2019geometric}, high compressibility~\cite{bringmann2019geometric}, and a very rich structure of overlapping communities of all sizes~\cite{bringmann2019geometric}.

The network models could also be used to explain why some algorithms perform much better on real-world networks than in worst-case instances, including algorithms for computing graph distances, diameters, vertex covers, maximal cliques, and Louvain's algorithm for community detection~\cite{bläsius2024external}. We highlight in particular the surprising fact that bidirectional breadth-first search can find the graph distance (hop distance) between two nodes in time that is not only sublinear in the graph size, but that variations of the algorithm need to explore fewer nodes than the largest degree on a shortest connecting path~\cite{cerf2024balanced}. Another remarkable feature of these network models is that greedy routing schemes that rely only on local information work extremely well and manage to send packages along (almost) shortest paths between the source and the target node~\cite{boguna2010sustaining,bringmann2017greedy,blasius2018hyperbolic}. This mirrors the outcome of Stanley Milgram's small-world experiment, which became famous under the name \textit{six degrees of separation}
~\cite{milgram1967small}. The same complex network models have also been argued to provide a better fit for epidemiological spread of diseases than other models\cite{komjathy2024polynomial, komjathy2023four}, and have been used to evaluate potential intervention strategies~\cite{jorritsma2020not,koch2016bootstrap}. 
A particularly welcome feature is that the number of weak ties can be explicitly controlled in the models and influences many of the aforementioned results~\cite{jorritsma2020not,bläsius2024external,komjathy2024polynomial,komjathy2023four,dayan2024expressivity}, which reflects the central role of weak ties in many empirical studies~\cite{granovetter1973strength}.

All those similarities with real-world networks emerge without specifying which properties of the nodes are modeled by the various dimensions of the latent space. In fact, such a specification is often difficult, or plainly impossible. While it is sometimes possible to identify a few dimensions in social networks which affect edge formation, it is hardly possible to identify \emph{all} such dimensions, as would be necessary for completely fitting a model to a real network\footnote{In some cases reconstruction of the underlying topology has been possible, for example for the autonomous systems of the internet network~\cite{boguna2010sustaining}}
, but in general this is a very challenging task, and it has been estimated that the dimensionality of the unknown latent space is between 5 and 10 in many cases~\cite{friedrich2024real}. Even if some important dimensions are abstractly known, often the according data is not available for the individual nodes. Such information is often even absent for ``hard'' data like geographic location, and it is very rare for ``soft'' information like interests and hobbies. Fortunately, many structural insights about complex networks and processes on them can be obtained without such a tight fit of individual nodes~\cite{komjathy2023four, komjathy2024polynomial, koch2016bootstrap, bläsius2024external,cerf2024balanced}.

While the use of a latent geometric space to model assortative tendencies of the nodes has generally been very successful, the aforementioned network models explicitly exclude a specific type of assortativity, namely \emph{degree assortativity}. Node degree, defined as the number of incident edges, is a prominent feature. Degree assortativity measures whether nodes of low degree tend to have a disproportionately large fraction of their connections to other low-degree nodes, and conversely for high-degree nodes. In fact, most empirical studies on assortativity in network science focus on degree assortativity. This has the simple reason that degrees can be directly inferred from the network structure and are thus always available, while other properties of nodes are usually not available. For the rest of the paper, we will focus on degree assortativity.

Contrarily to most other node properties, positive assortativity (homophily) \emph{of degrees} is not universally found in complex networks. Some such networks also show disassortativity (heterophily) of degrees. A common rule of thumb is that social networks are often positively assortative, whereas biological and technological networks tend to show negative assortativity~\cite{newman2002assortative}.\footnote{Note that we can \emph{not} confirm that social networks are positively assortative with respect to the degrees. Rather, we find two opposing trends, a positive and a negative one, and either of the two can dominate. We discuss those findings in Section~\ref{sec:joint-distr-real}.} Describing the assortative tendencies of degrees in a network is a fundamental goal of network science in its own right, but they have also been shown to have far-reaching consequences on community structure, network robustness and a multitude of network dynamics:
\begin{itemize}
    \item Disassortative networks exhibit a higher epidemiological threshold, whereas assortative networks allow for longer periods of intervention before an epidemic spreads in the BTW sandpile model of distress propagation \cite{dagostino2012robustness}.
    \item Assortativity decreases network robustness \cite{zhou2012robustness}.
    \item Assortativity in multilayer networks - connecting hubs of one network with the hubs in another network - can enhance cooperation under adverse conditions \cite{duh2019assortativity}.
    \item Assortative degree correlations can strongly improve the sensitivity for weak stimuli in the brain \cite{schmeltzer2015degree}.
    \item Conversely, dynamical processes can also influence the assortativity of an underlying network \cite{bialonski2013assortative}.
    \item Degree-degree correlations may cause the failure of heuristic algorithms \cite{vazquez2003complexity}.
\end{itemize}

\subsection{Our Contributions: Informal Discussion}

This paper makes several key contributions of varying nature. They range from rigorous mathematical theorems to experimental results. Thus, their level of technicality varies substantially - and as a consequence they may appeal to different audiences. To make our paper accessible, we first discuss all our main results informally, avoiding technical details as far as possible and only providing intuition. Afterwards, we go through the results again. In this second pass, we are still brief and only present the key results, but we formulate them more rigorously and with more technical details, so that readers can opt to skip the technical parts if desired.

\paragraph{Negative result on Pearson assortativity coefficient.} As motivated above, we need good and robust ways to describe and understand assortativity of degrees. However, assortativity is a complex concept and is not easy to measure. The most common approach is to compute a single real-valued coefficient, whose sign then indicates whether the network is assortative (positive sign) or disassortative (negative sign). The most commonly used coefficient is the Pearson assortativity coefficient. However, it has been challenged whether this coefficient is a good measure for assortativity~\cite{litvak2013uncovering}. In this paper we give a negative answer in the strongest possible sense, which is our first main result. We show that for every network with a sufficiently heavy-tailed degree distribution, the Pearson assortativity coefficient is negative, regardless of the assortative behavior of the network. This confirms mathematically that this coefficient does not measure assortativity of degrees if the degree distribution is heavily skewed.

\paragraph{Other coefficients and a new lens on assortativity.} There are other coefficients than the Pearson assortativity coefficient which are better suited, in particular Spearman's and Kendall's rank correlation coefficients~\cite{litvak2013uncovering}. Another contribution of this paper is that we investigate those coefficients for a selection of networks from the KONECT database~\cite{kunegis2013konect}. However, we argue that, in many of these cases, reducing assortativity to a single number is too simplistic. Instead, we propose an alternative way to study assortativity. In a nutshell, we suggest to consider a random edge $e=\{u,v\}$, and to study how the degree distribution of $v$ changes when we condition on the degree of $u$. We explore several ways in which this information can be graphically illustrated. We will discuss all alternatives and the full technical details in Section~\ref{sec:experiments}, but already show an example here for illustration in Fig~\ref{fig:example-intro}. Further examples both for artificially generated networks and for real networks from the KONECT database can also be found in Section~\ref{sec:experiments}, as well as alternative ways for graphical illustration.

\begin{figure}[ht!]
  \centering
  \includegraphics{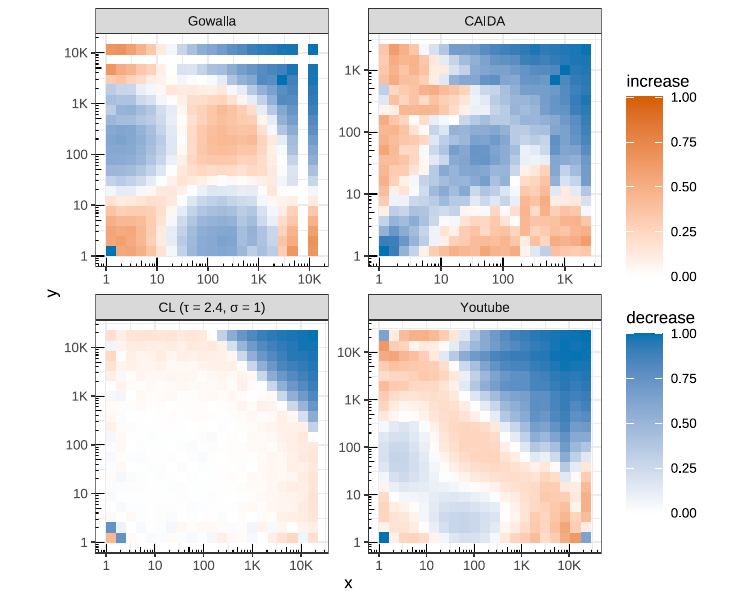}
  \caption{Four heatmaps displaying the effect on a random edge $e=\{u,v\}$ when conditioning on the degree of one endpoint. A red color at coordinate $(x,y)$ indicates that the information $\deg(u) = x$ makes it more likely that the other endpoint of $e$ has degree $\deg(v) = y$. A blue color shows that this becomes less likely. All networks follow a power-law degree distribution. Details on the networks are provided in Section~\ref{sec:experiments}. The \texttt{Gowalla} social network (top-left) shows mostly positive assortativity: the diagonal is red, so nodes of the same degree tend to connect to each other, while most regions elsewhere are blue. The \texttt{CAIDA} Autonomous Systems Network (top-right) shows strong negative assortativity.  Two networks (bottom row) with similar degree distributions. Both would be classified as assortativity-neutral by Spearman's and Kendall's assortativity coefficients because both coefficients are close to zero. However, while the artificially generated network (bottom-left) is indeed neutral (the blue triangle is inevitable in power-law networks, see Section~\ref{sec:intro-theoretical-results}), the \texttt{Youtube} Network (bottom-right) has both positive and negative assortative tendencies which cancel each other out. The heatmap is thus more informative than a single coefficient.}
  \label{fig:example-intro}
\end{figure}

\paragraph{Analyzing assortativity in generative models of complex networks.} As mentioned before, many complex network models with an underlying latent space do not include degrees in this latent space, but handle them separately.\footnote{Some models, like Hyperbolic Random Graphs, do include the weights into the latent hyperbolic space, but in a separable way. Our conclusion still applies to these models: they are mostly assortativity-neutral, and the spurious positive assortativity comes from the same source as for the GIRG model.} This allows to prescribe the degree distribution, for example as a power law. Thus, while the resulting networks are assortative with respect to the features modeled by the latent space, it has been unclear whether their degrees are also assortative. As our third main contribution, we answer this question for the GIRG model, one of the most prominent complex network models with an underlying geometry and an inhomogeneous degree distribution. We show that it is mostly, but not totally, \emph{assortativity-neutral}. Moreover, we can pinpoint exactly where the small (positive) assortativity comes from. It originates from random fluctuations in the number of nodes with similar properties. Consider that -- by chance -- there are more nodes than expected with some specific property, i.e., there are more nodes than expected in the same region $X$ of the latent space. Then the nodes in $X$ have higher degrees, but also their neighbors tend to be in $X$ and thus also have higher degree. This gives a small positive contribution to assortativity of degrees. The effect is relevant for nodes of very small degree, but is negligible for nodes of medium or large degree. If we remove this effect by removing the latent space, we show that the resulting \emph{Chung-Lu model}~\cite{chung2002average} is perfectly assortativity-neutral. Here we measure neutrality not just by a single real number, but in the strongest possible sense: we show that for a random edge $e=\{u,v\}$ the whole probability distribution of $\deg(v)$ is invariant under information about $\deg(u)$.\footnote{Except for a ceiling effect for extremely large degrees that we quantify precisely.} This can also be visually seen in Fig~\ref{fig:example-intro} (c), which displays such a Chung-Lu random graph. 

We provide a formal statement that pinpoints in which sense this neutrality is retained in GIRGs. More precisely, we show that the GIRG model is still perfectly assortativity-neutral with respect to \emph{weights}, but not to degrees. Weights are internal parameters of the model that are tightly, but not perfectly, coupled to the degrees. We can show that the only non-neutral contribution, a positive one, to assortativity of degrees comes from this slight mismatch between weights and degrees. If we isolate this effect by removing the degree inhomogeneity (since the effect is negligible for nodes of large degrees) and retaining the latent space, then we can exactly quantify the assortativity of the resulting \emph{Random Geometric Graph} model. Thus we can split the construction of latent space models like GIRGs into two steps, one of which is perfectly assortativity-neutral, while the other one gives a mild positive assortativity that we can quantify.

\paragraph{Tuning the assortativity of complex network models.} As our analysis shows, established models of complex networks such as GIRGs are mostly assortativity-neutral for degrees. This raises the question of how we can build models with a heavy-tailed degree distribution and a latent space (in order to maintain assortativity for \emph{other} properties than degrees), albeit with tunable assortativity of degrees. Our final contribution is to present exactly such a model. We show that we can slightly alter the connection probability of the GIRG model in such a way that assortativity of degrees can be tuned to be either positive or negative, while keeping the basic properties of the model intact. These properties include the tight coupling between weights and degrees, which maintains the prescribed degree distribution, but also a high clustering coefficient, small-world properties, and so on. As for the GIRG model, we can split the degree assortativity into two parts: a minor part coming from mismatches of weights and degrees of low-degree nodes, as before; and a major part that can be described via assortativity of the weights. We precisely quantify this latter part and show that the (positive or negative) assortativity of the model can indeed be tuned by a parameter.

\subsection{Our Contributions: Technical Summary}

In this section we elaborate on our contributions from a technical perspective. In the following, assortativity always refers to the degrees of the network's nodes. We denote a neighborhood relation between two nodes $u$ and $v$ by $u\sim v$.

\subsubsection{Networks and Network Models} 
We consider two types of networks, real-world networks and artificially generated networks. The real-world networks are taken from the KONECT database~\cite{kunegis2013konect}, where we restrict to undirected simple networks whose degrees follow a power law according to~\cite{voitalov2019scale}.

We focus on two generative models for complex networks. The first one is called \emph{Geometric Inhomogeneous Random Graphs (GIRGs)}. In this model, $n$ nodes are placed randomly in the latent space $[0,1]^d$ for some fixed dimension $d$.\footnote{Often, the torus topology is used, i.e., each dimension wraps around, since this removes boundary effects and makes the model more symmetric.} Weights are drawn from a power-law distribution with law $\Pr(W = w) = \Theta(w^{-\tau})$ for $w\ge 1$, where the power-law exponent $\tau$ is usually in the range $\tau\in (2,3)$. Afterwards, each pair of nodes $u,v$ connects independently with probability $\Theta(\min\{1,\tfrac{w_uw_v}{n\|x_u-x_v\|^d}\}^\alpha)$, where $w_u,w_v$ are the weights and $x_u,x_v$ the positions of $u$ and $v$ respectively, and the parameter $\alpha > 1$, sometimes called the inverse temperature of the model, modulates the proportion of weak ties in the network. The $\Theta(\cdot)$ allows for constant factor deviations. This formula turns out to give exactly the right scaling, as we explain in detail in Section~\ref{sec:basic-properties}.

The other complex network model we look at is the \emph{Chung-Lu random graph model}. It is similar to the GIRG model, except there is no underlying latent space. So, every node draws a weight as for GIRGs, but then any pair of nodes $u,v$ connects with probability $\Theta(\min\{1,\tfrac{w_uw_v}{n}\})$, with no geometric information. The two models are closely related, as the marginal connection probability is the same when randomizing over the geometric positions in GIRGs. This means that for two nodes $u,v$ with fixed weights but \emph{random} positions, the marginal connection probability in a GIRG is $\Pr[u\sim v \mid w_u,w_v] = \Theta(\min\{1,\tfrac{w_uw_v}{n}\})$, just as for Chung-Lu graphs. As a consequence, many properties transfer from Chung-Lu graphs to GIRGs. In particular, in both models the expected degree of a node coincides with its weight up to a constant proportionality factor. This leads to a degree distribution that follows the same power-law as the weights, ensuring that the neighborhood distribution\footnote{The neighborhood distribution is obtained by selecting a uniformly random edge $e$, choosing a random endpoint $v$ of $e$, and then returning the degree of $v$.} and the small-world property are conserved and remain equivalent for both models. However, due to the latent space, GIRGs also have strong clustering and communities, which Chung-Lu graphs do not have.

We also introduce a variant of GIRGs for which the assortativity coefficient can be tuned to be either positive or negative. We call the resulting model Tunable GIRG or TGIRG. The construction steps are the same as for GIRGs, except that we use a different connection probability: we connect each pair of nodes $u,v$ with probability $\Theta(\min\{1,\tfrac{(w_u \wedge w_v)^\sigma (w_u \vee w_v)}{n\|x_u-x_v\|^d}\}^\alpha)$, where $\sigma\in[0,\tau-1)$ is a parameter that determines how assortative or disassortative the model is. Here, $\sigma=1$ yields the GIRG model and thus (almost) neutral assortativity, while smaller values of $\sigma$ give disassortative networks and larger values give assortative networks.

Finally, we also give a brief analysis of assortativity in \emph{Random Geometric Graphs} (RGG). Those graphs are similar to GIRGs, though the weights do not follow a power law, but are all set to one. Moreover, as we only use those graphs to demonstrate a point for GIRGs, we restrict ourselves to the case of $\alpha = \infty$ (zero temperature), and to the case that the hidden constant factor in the $\Theta(\cdot)$-term is one. See the discussion of theoretical results on GIRGs below for why we discuss those graphs.

\subsubsection{Experimental Results} 
We systematically explore degree assortativity of the real networks and of artificially generated Chung-Lu and GIRG networks. We find, in accordance with our theoretical result, that the Pearson assortativity coefficient is not a good measure for assortativity and provides almost no information. Spearman's and Kendall's correlation coefficients are generally in good agreement with each other. However, we find cases in which those coefficients are close to zero even though there are strong positive and negative assortative effects in the network which happen to cancel out.

We argue that single real coefficients lose too much information, and that a more detailed approach is useful. We offer three visualizations of assortativity. We depict an example of each in Fig~\ref{fig:caida-intro}, but defer their detailed discussion, the systematic exploration of artificial and real networks, and the interpretation of the findings, to Section~\ref{sec:experiments}. For additional examples of networks where intricate wiring patterns are concealed by coefficient values close to zero, we refer the reader to Fig~\ref{fig:wordnet} and Fig~\ref{fig:youtube} in the Supporting Information. The three visualizations are:
\begin{itemize}
    \item A comparison of line plots showing the conditional probability $\Pr[\deg(v) \ge x \mid v \sim u]$ on a log-log scale, where $u$ is a random node of prescribed degree $d$. By inspecting how the curves change for different values of $d$, one can infer what the typical neighbors of nodes of different degrees look like. Of particular interest is the ordering of the curves and whether such curves cross each other for different $d$. For artificial networks, our theoretical results show that each such curve can be decomposed into two sub-parts, both of which are linear. The slopes of the two linear parts and the transition point between them are central and should be used as baselines for the curves for real networks.

    To obtain a proper baseline for such curves it is important to keep in mind the difference between sampling a random node and sampling a random edge. In the former case, one obtains the degree distribution of the network. In the latter case, one obtains the \emph{neighborhood distribution}\footnote{Equivalently, the neighborhood distribution may be obtained by choosing a uniformly random node $u$ of degree $\ge 1$, go to a uniformly random neighbor $v$ of $u$, and return the degree of $v$.}, which is shifted towards high-degree nodes. For heavy-tailed degree distributions, those two distributions are very different from each other. It is crucial to use the latter distribution as baseline, not the former one. This is closely related to the \emph{friendship paradox} that a random neighbor of a random node has larger average degree than a random node. 
    \item Heatmaps of joint degree frequencies, i.e., heatmaps where the coordinate $(x,y)$ displays the proportion of edges $e=\{u,v\}$ with $\deg(u) = x$ and $\deg(v)=y$. Comparing different rows and columns of those heatmaps reveals preferences of nodes of different degrees.
    \item The heatmaps get easier to interpret if one uses normalization. Thus, as the third option we display the normalized values $\frac{\Pr[\deg(v)=y \mid \deg(u)=x]}{\Pr[\deg(v)=y]}$ for a random edge $e=\{u,v\}$. The exact definition is slightly more complicated and involves a case distinction for whether the fraction is larger or smaller than one. We defer the details to Section~\ref{sec:experiments}.
\end{itemize}


\begin{figure}[ht]
  \centering
   \includegraphics[scale=0.6]{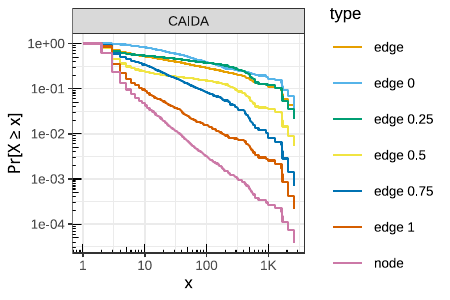}%
  \hfill%
  \includegraphics[scale=0.6]{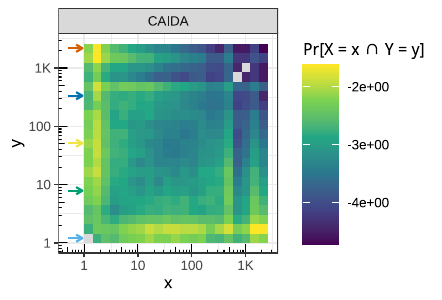}%
  \hfill%
  \includegraphics[scale=0.6]{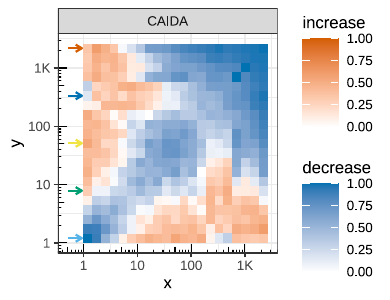}
  \caption{Conditional and Joint Degree Distributions of the \texttt{CAIDA} Autonomous Systems Network.}
  \label{fig:caida-intro}
\end{figure}

A detailed discussion of our findings for social and technical networks can be found in Section~\ref{sec:experiments}. The most important points are:
\begin{itemize}
    \item All social networks that we consider show an increased likelihood for edges between pairs of low-degree nodes. This is very similar to the effect of adding a latent geometric space to the artificial models, see below. This indicates that such a geometric space may be beneficial for modeling social networks.
    \item Beyond those pairs of low degrees, there are very different assortativity patterns in social networks. We find at least two types of different patterns. Those can contribute positively or negatively to an assortativity coefficient. If they contribute negatively, then we have opposing contributions to the coefficient, and either contribution may dominate. In general, the Spearman and Kendall assortativity coefficients of social networks range from strongly negative to strongly positive.
    \item The assortativity patterns for technical networks of autonomous systems strongly differ from those of social networks. They are usually negatively assortative. Other than social networks, those networks do not generally show an increased likelihood in connections between pairs of low-degree nodes.
    \item We can match the different assortativity patterns in real networks to different parameter regimes in the tunable models of artificial networks that we propose (see below). 
\end{itemize}

\subsubsection{Theoretical Results}\label{sec:intro-theoretical-results}

\paragraph{Negative result for the Pearson assortativity coefficient.} Our first theoretical result is showing that the Pearson assortativity coefficient is not a good measure for assortativity. This has already been argued before~\cite{voitalov2019scale, litvak2013uncovering,hofstad2014degree}, but we give a strong theoretical result:

\begin{restatable}{theorem}{Negativepearson}\label{thm:neg_pearson}
  Let $2<\tau<7/3$ and let $\mathcal{G} = (\mathcal{V}, \mathcal{E})$ be a graph on $n$ vertices with $\Delta := \max_{v\in\mathcal{V}}\deg(v) = \Theta(n^{1/(\tau-1)})$ such that for a vertex $v$ chosen uniformly at random in $\mathcal{V}$ and for all $0<k\le\Delta$ we have $\pr(\deg(v)\ge k) = \Theta(k^{1-\tau})$.  Then its Pearson assortativity coefficient $r(\mathcal{G})$ is negative, more precisely $r(\mathcal{G})=\Theta(-n^{-(\tau-2)/(\tau-1)})$.
\end{restatable}

\noindent This result shows that the Pearson assortativity coefficient does not measure assortativity for power-law networks with a sufficiently heavy tail. It always gives the same result, regardless of how the nodes in the network are wired.

The reason for this result is visible in Fig~\ref{fig:example-intro}. Note that in all four subplots there is a strongly blue triangle in the upper right corner. This is not an artifact of the networks, but is inevitable in all networks which follow a power-law degree distribution. Such networks simply have too few vertices of large weights, so that the required number of neighbors of large weight for neutrality in this region is impossible to reach. All assortativity coefficients compute a weighted average of the different regions depicted in Fig~\ref{fig:example-intro}, but the Pearson coefficient puts particular weight on this upper right region, which forces it to be negative if the power-law degree exponent $\tau$ is smaller than $7/3$.

\paragraph{Weight assortativity in Chung-Lu graphs and GIRGs.}
We investigate the assortativity of weights in Chung-Lu graphs and GIRGs. Note that the weights are closely coupled to the degrees, but they are not identical, and we discuss the difference further below. We show that both Chung-Lu graphs and GIRGs are assortativity-neutral for weights, except for a cut-off for extremely large weights. More precisely, consider a vertex $u$ with weight $w_u \le n$, and consider a random neighbor $v$ of $u$. Then the conditional weight distribution of $v$ is given by
\begin{align}
    \Pr[w_v \ge w\mid w_u, v\sim u] =
    \begin{cases}
    \Theta(w^{1-\tau}) & \text{if $w_u \le  \frac{n}{w}$,} \\
    \Theta(\tfrac{w^{-\tau} n}{w_u}) & \text{if  $\frac{n}{w} < w_u \le n$}.
    \end{cases}
    \label{eq:conditional-CL-intro}
\end{align}
The full formulation of the statement can be found in Proposition~\ref{prop:shifted-conditional-weight-distribution-tuned} (by setting $\sigma=1$ in the statement). 
Note that the second case is very exceptional. Most nodes have small (constant) weights, and the maximal weight in the graph is $\Theta(n^{1/(\tau-1)}) = o(n)$. So for most nodes $u$, the second case is void. Only if the weight $w_u$ is very large, the second case starts occurring for large values for $w$. For example, if $u$ has an exceptionally large weight of $w_u = \Theta(\sqrt{n})$, then the tail distribution of $w_v$ changes at $w = \Theta(\sqrt{n})$, which again is a very small part of the tail. 

The key insight about~\eqref{eq:conditional-CL-intro} is that the formula in the first case is independent of $w_u$. In other words, the distribution of $w_v$ in the neighborhood of $u$ does not depend on the weight of $u$, except possibly for a cut-off at very large weights.\footnote{The reason for the cut-off is that above the cut-off the number of vertices with so high weight in the graph is simply too small to satisfy the formula in the first case, the same effect as discussed in the previous section.} This means that, up to the exception at the cut-off, the models are perfectly assortativity-neutral with respect to the weights.

\paragraph{Degree assortativity in GIRGs and RGGs.} 
We have established that \emph{weights} are assortativity-neutral in GIRGs. If the weight is super-constant, then it is so tightly coupled to the degree that the degrees are necessarily also assortativity-neutral. However, for constant weights the coupling is not so tight, and in this case there is a positive assortative contribution from random fluctuations of the vertex locations. To describe this effect, we strip the GIRG model of nodes of larger weights, which yields the RGG model. Moreover, in order to keep the exposition simple we only consider the case of zero temperature. In Proposition \ref{prop:rggassortativity}, we derive an exact formula for the assortativity in this resulting model. Essentially, for a vertex $u$ with given degree, the degree of a random neighbor $v$ of $u$ is the sum of two independent Poisson random variables. One of them is independent of $\deg(u)$, while the expectation of the other is proportional to $\deg(u)$, which yields a positive constant assortativity. 

We believe, without formulating a mathematical theorem, that assortativity in GIRGs is best regarded as the result of two consecutive processes: for a random vertex $u$ the \emph{weight} of a random neighbor $v$ is not coupled to $w_u$. However, the translation of weight into degree gives positive assortativity, modulated by exactly the same process as for RGG. This modulation is only relevant for constant weights and degrees. This is also strongly confirmed by experiments, which show positive assortativity for constant degrees, and neutral assortativity otherwise.

\subsubsection{TGIRG: A Complex Network Model With Tunable Assortativity}

As our final contribution, we propose a complex network model that builds on the GIRG model and inherits its wealth of structural properties, and at the same time can be tuned to have either positive or negative assortativity. The idea is to use a slightly different probability kernel. In the GIRG model, the connection probability increases with the product of the weights $w_u \cdot w_v$. This product can equivalently be written as $(w_u\wedge w_v)\cdot (w_u \vee w_v)$. The idea is now to scale up or down the smallest of those two factors with an exponent $\sigma \in [0,\tau-1)$.\footnote{This range is chosen so that the expected degree of a vertex is of the same order as its weight as in the original GIRG model, see Lemma~\ref{lem:exp_deg_sigma}.} 
This follows an idea of~\cite{jorritsma2023cluster}, and in fact this model has been studied before~\cite{luchtrath2022percolation,gracar2022recurrence}, but not with respect to assortativity. 

In this paper we show two things: firstly, we analyze basic properties of the resulting TGIRG model, including its marginal connection probability (Proposition~\ref{prop:joint-weight-density-tuned}). We also show that it shares properties of the GIRG and the Chung-Lu model, in particular that it has the same neighborhood distribution as both (Proposition~\ref{prop:shifted-weight-density-tuned}) and that it has a large clustering coefficient of $\Omega(1)$ as GIRGs (Proposition~\ref{prop:clustering-tgirg}). And secondly we show that assortativity of the resulting model is indeed tunable. More precisely, let us fix as before a vertex $u$ of fixed weight $w_u$, and let $v$ be a random neighbor of $u$. Then we show the following formula, where for simplicity we ignore cases that only hold for exceptionally large weights (the full statement can be found in Proposition~\ref{prop:shifted-conditional-weight-distribution-tuned}):
\begin{align}\label{eq:conditional-tuned-intro}
    \Pr[w_v \ge w \mid w_u, v\sim u]=
    \Theta(w^{1-\tau}(w\wedge w_u)^{\sigma-1}). 
\end{align}
In order to understand the formula, let us first consider $\sigma >1$ and use the formula $\Theta(w^{1-\tau})$ of the assortativity-neutral GIRG model as baseline. Compared to the baseline, the probability is boosted by a factor $(w\wedge w_u)^{\sigma-1}$. For $w\le w_u$, this factor simplifies to $w^{\sigma-1}$, and it is capped at $w_u^{\sigma -1}$ for $w> w_u$. The first thing to note is that this factor is increasing in $w$, so larger weights $w$ get boosted more. Moreover, for larger $w_u$ the boost is stronger in two ways: firstly there is a prolonged interval $[1,w_u]$ in which the boost factor is increasing; and secondly the maximum boost factor $w_u^{\sigma-1}$ is higher. Both effects imply that for larger weight $w_u$, the distribution puts more emphasis on the tail of the distribution. Thus the resulting model is assortative.

On the other hand, if $\sigma<1$ then the exponent $\sigma-1$ is negative. This reverses the effect and makes the network disassortative.\footnote{Note that this is assortativity of \emph{weights}. The same effect as for GIRGs will add a small positive assortative effect for \emph{degrees} that is only relevant for small degrees.} We confirm these findings by experiments. Fig~\ref{fig:embedding_tgirg} illustrates the effects of increasing $\sigma$ on a graph with the same vertex set and approximately the same average degree.

\subsection{Related Work}
In this subsection, we give an overview of related work, including known results on TGIRGs.

\paragraph{Origins of the line of research and early results.} The study of assortativity through the lens of network science was initiated by Newman and various collaborators in a series of papers~\cite{newman2001structure, newman2002assortative, newman2003mixing}. In these early works, they introduced the Pearson correlation coefficient as a measure for network assortativity and observed experimentally that social networks tend to be positively assortative, while biological and technological networks exhibit negative assortativity~\cite{newman2002assortative}. Newman also computed the (vanishing) limits of the Pearson correlation coefficient for Erdös-Rényi random graphs and preferential attachment graphs as well as for the grown graph model of Callaway. In the same work, he also proposed a positively assortative network model, for which he runs simulations according to which the phase transition for the existence of a giant component occurs earlier with increasing assortativity. He further observed numerically that contrarily to other networks, the attack strategy of removing a small number of high-degree nodes in order to disconnect the giant component is relatively inefficient in assortative networks because high-degree nodes cluster together~\cite{newman2002assortative}. Degree-degree correlations, as well as non-degree assortativity were studied in~\cite{newman2003mixing}, including mixing by ethnicity or age in heterosexual partnerships, indicating strong positive correlations. The paper also proposes Monte-Carlo algorithms for generating networks of a desired assortativity level. Simulation results again point to an earlier emergence of the giant component with increasing assortativity. At the same time, the size of the giant component is smaller for assortative graphs than for neutral and disassortative instances. This is credited to the fact that the well-connected core of high-degree nodes already leads to a giant component at lower densities, but these lower densities also make it less likely that the giant component extends to regions in the graph outside this core. (Non-degree) assortativity was subsequently investigated as a plausible cause of emerging communities in networks~\cite{Newman2003community}, accompanied by an algorithm based on the betweenness centrality measure, which was later applied successfully also to non-human social networks~\cite{lusseau2005identifying}. Newman and Park demonstrated that the reverse implication connecting communities and assortativity sometimes also holds: community structure can lead to degree-degree correlations~\cite{newman2003why}. These early research articles already illustrate the difficulty of finding models which are both tunable and mathematically tractable.

\paragraph{Assortativity tuning.} Numerous approaches have been tried in order to generate networks with varying levels of assortativity. The most general procedures are the ones introduced by Newman~\cite{newman2003mixing} and Boguñá and Pastor-Satorras~\cite{boguna2003class} which devise two different schemes to construct general correlated networks with prescribed correlations. Simpler approaches focus on imposing merely the intuitive requirement that ``nodes with similar degrees connect preferably'', instead of ``hard-coding'' the desired correlations~\cite{XulviBrunet2004ReshufflingSN}. Another classical procedure relies on an idea from the so-called configuration model\cite{molloy1995critical}, which allows to generate networks with a desired degree sequence by assigning to each node the number of ``stubs'' that correspond to its desired degree and then randomly connecting stubs. This idea can be extended to a rewiring procedure where pairs of edges $uv$ and $rs$ are selected randomly replaced by $ur$ and $vs$ if and only if the assortativity of the network is increased (or decreased) through this rewiring \cite{newman2003mixing,XulviBrunet2004ReshufflingSN,maslov2002specificity}. Finally, there exist various modifications of the preferential attachment model, which was originally introduced by Barabási and Albert~\cite{barabasi1999emergence}. In this model of a \textit{growing network}, nodes are added one-by-one, connecting to existing nodes with a probability that is proportional to their respective degrees. I.e., the ``rich get richer'' and newly arriving (and therefore low-degree) nodes connect preferably to nodes of high degree. By modifying the attachment rule, one can directly change the wiring preferences of nodes~\cite{newman2002assortative,vazquez2003growing}. The iterative generation of such networks introduces a lot of dependencies which make the rigorous mathematical analysis of networks difficult. Although we focus here on degree-assortativity, we mention that for attribute-assortativity, besides the latent space approach that underlies GIRGs, there exist several mechanisms which allow to adjust the assortativity levels of a network, such as the stochastic block model and its variants\cite{holland1983stochastic,karrer2011stochastic,airoldi2008mixed}.

\paragraph{Structural and process implications of (dis)assortativity.} A key motivation for studying assortativity of networks is the belief that the level of assortativity can have far-reaching structural consequences for the network in question as well as the evolution of dynamical processes on it. Structurally, degree assortativity is said to lead to tightly-knit cores of influential nodes, often referred to as a ``rich-club''\cite{colizza2006detecting}. We remark here that such a well-connected core consisting of high-degree nodes can also be produced by other mechanisms, as exhibited by GIRGs~\cite{bringmann2024average}. Another assertion echoed in many papers is that network robustness is greatly affected by the (dis)assortativity of a network~\cite{holme2002attack,schneider2011mitigation}. Dynamical processes on networks have also been studied intensively through the lens of assortativity\cite{newman2002spread}. This has led to the observation that degree correlations impact the epidemic threshold~\cite{Bogu2002AbsenceOE} and the development of models for spreading which take assortativity into account\cite{noldus2015assortativity}. Synchronization phenomena in complex networks have also been investigated~\cite{gomez2007paths,kelly2011on}.

\paragraph{Challenges of research on assortativity.} Since the proposed models are often tailored to explain a single structural aspect of real-world networks or study a specific dynamical process, this has resulted in a wealth of otherwise sparsely studied models which are difficult to compare. From a practitioner's perspective, any insight relating degree-degree correlations to structural properties such as network robustness, which is crucial in critical infrastructures, also requires that the examined model captures more than just one aspect of real-world networks in order to be valid. As outlined in the introduction, GIRGs jointly exhibit many such desirable properties stemming from the latent geometric space. As we discuss in the next subsection, a lot is already known about our proposed extension of GIRGs as they fit well into existing frameworks. 
For the description of further properties of TGIRGs, we refer the reader to Section~\ref{sec:basic-properties}; most notably they have a large local clustering coefficient of $\Theta(1)$.

\subsubsection{Related Work On TGIRGs}
A variety of properties of TGIRGs have been shown for different ranges of $\sigma$. Some of these properties have been derived using a different parametrization of the model, known as the weight-dependent random connection models with the interpolation kernel (see for example \cite{luchtrath2022percolation}). Our choice of parametrization has the advantage that weights are intuitive proxies for expected degrees and that $\tau$ directly controls the decay of their power-law distribution, whereas the alternative parametrization only provide this indirectly. There also exist parametrizations which assign different exponents $\sigma_1$ and $\sigma_2$ to the smaller and larger weight in the connection probability. These can however always be reparametrized (as long as $\sigma_1>0$) to give $\sigma_1=1$ by changing the power-law exponent $\tau$ accordingly~\cite{jorritsma2023cluster}. 

We first give the conversion formulas between our parametrization (Definition \ref{def:simple-agirg}) and the alternative parametrization, which uses marks in the interval $[0,1]$ instead of power-law distributed weights. Apart from the dimension $d$ of the underlying space (for which the mark parametrization does not differ from ours), three parameters are used, $\gamma_M, \delta_M, \alpha_M$, where the subscript $_M$ serves to distinguish the parameters of the mark parametrization from the parameters of our parametrization (in particular since both parametrization use the symbol $\alpha$, but for different purposes). The conversion formulas between the two parametrizations are as follows: 

\begin{align*}
    \tau = 1+\tfrac{1}{\gamma_M} &\Longleftrightarrow \gamma_M = \tfrac{1}{\tau-1}, \\
    \alpha &= \delta_M, \\
    \sigma = \tfrac{\alpha_M}{\gamma_M} &\Longleftrightarrow \alpha_M = \tfrac{\sigma}{\tau-1}.
\end{align*}
The existence and uniqueness of a giant component is well understood \cite[Proposition 2.4 and Corollary 2.14]{luchtrath2022percolation}:
\begin{itemize}
    \item If $\sigma > \tau - 2$ or $\tau<2+\frac{1}{\alpha}$, whp\footnote{We say an event occurs \textit{with high probability (whp)} if it occurs with probability tending to $1$ as $n$ tends to  $\infty$.} the graph contains a unique giant component.
    \item If $d=1$, $\alpha>2$, $\tau>2+\frac{1}{\alpha-1}$, and $\sigma<\tau-2$, whp the graph does not contain a giant component.
    \item For other parameter regions, the presence of a giant component depends on the constants hidden in the $\Theta$-notation of the connection probability in \eqref{eq:agirg-connection}.
\end{itemize}
Additionally, the cluster size decay of TGIRGs is studied in \cite{jorritsma2023cluster, jorritsma2024large}. Regarding the average distance in the giant component, unlike with standard GIRGs for which this is always doubly logarithmic in the number of nodes, various scaling regimes can occur, depending on the model parameters. When $\sigma\in[0,\tau-2]$, TGIRGs satisfy Assumptions 1.1 and 1.2 in \cite{gracar2022chemical}, and hence Theorems 1.1 and 1.2 therein apply, which yields that the average distance is whp $\frac{4\pm o(1)}{\log 1/(\alpha(\tau-2)))}\log\log n$ if $\tau<2+\tfrac{1}{\alpha}$ and $\omega(\log\log n)$ if $\tau>2+\tfrac{1}{\alpha}$. When $\sigma \in [1,\tau-1)$, the proof in \cite{bringmann2024average} can be adapted (the proof is for GIRGs, i.e., for $\sigma=1$, and for larger $\sigma$ the connection probability between two given vertices is larger than for GIRGs), yielding a doubly logarithmic average distance in this case as well. On the other hand, if $\alpha>2$, $\tau>2+\frac{1}{\alpha-1}$, and $\sigma<\tau-2$, then the average distance is linear in the number of nodes \cite{luchtrath2024all}. In some parameter regimes, TGIRGs also exhibit an average distance that is neither linear nor ultra-small. For example, the connection probability in TGIRGs is lower bounded by the connection probability in long-range percolation, where the average distance is at most polylogarithmic when $\alpha<2$ \cite{biskup2004scaling}, yielding at most polylogarithmic distances for any TGIRG with $\alpha<2$. This holds for any choice of $\tau$, $\sigma$ and $d$, including choices where the average distance is $\omega(\log \log n)$.

\begin{figure}
    \centering
    \includegraphics[width=0.3\linewidth]{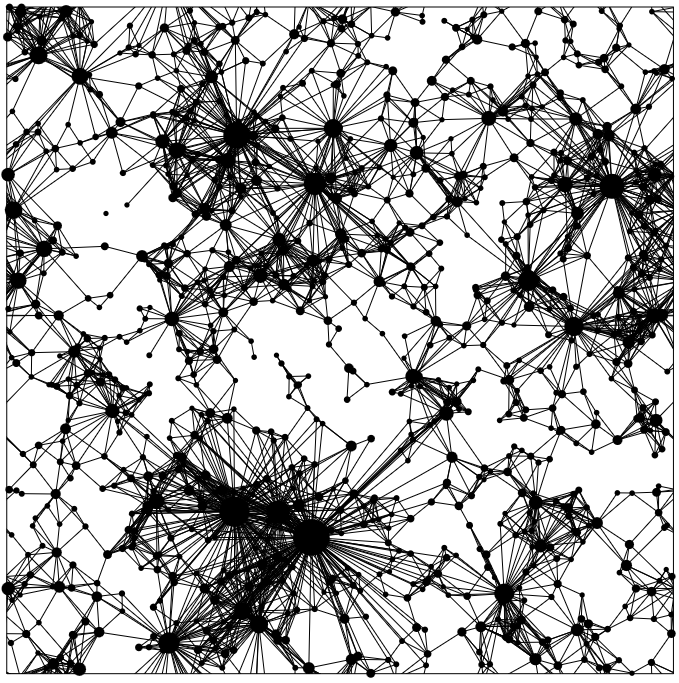}
    \includegraphics[width=0.3\linewidth]{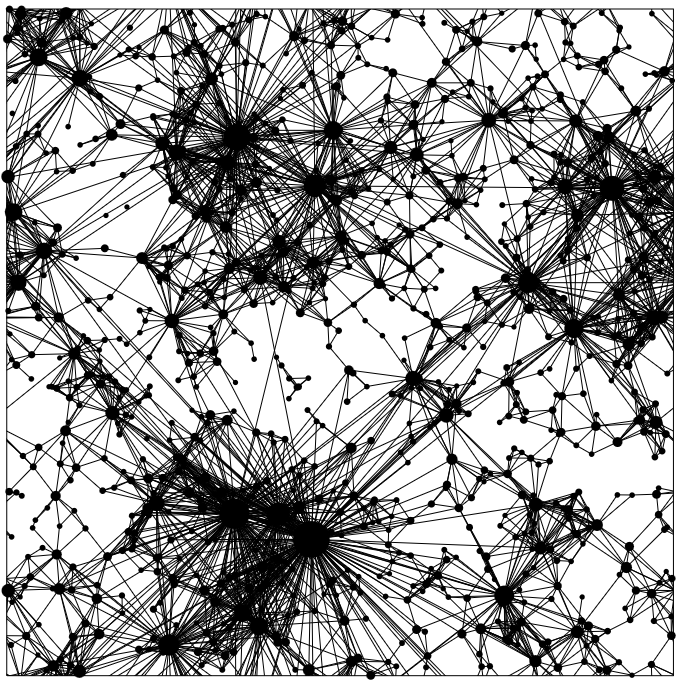}
    \includegraphics[width=0.3\linewidth]{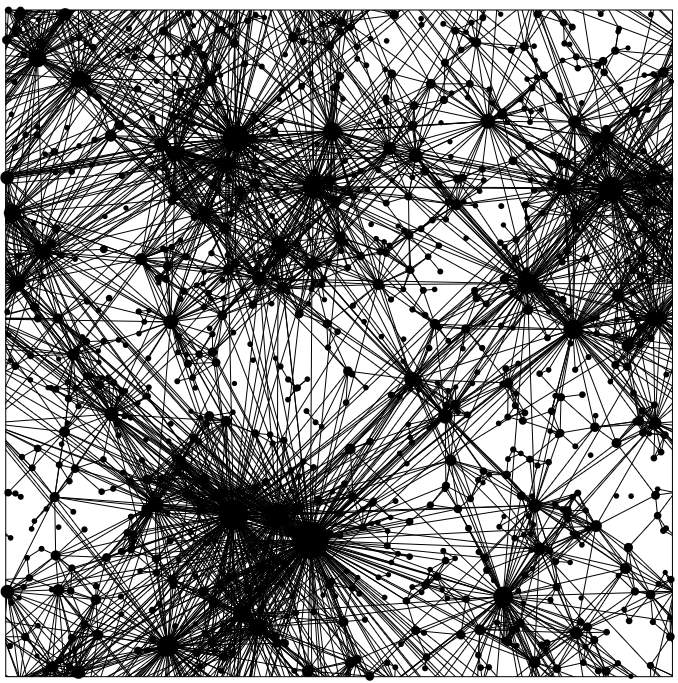}
    \caption{TGIRGs with $n=1000$, $\tau=2.8$, $\alpha=\infty$, and average degree 10 on the same vertex set for varying $\sigma$: 0.2 (\textit{left}), 1 (\textit{center}), 1.8 (\textit{right}).}
    \label{fig:embedding_tgirg}
\end{figure}

\paragraph*{Organization of the paper.} The remainder of the paper is organized as follows. In Section~\ref{sec:model}, we formally introduce the graph models we analyze and state some of their key properties in Section~\ref{sec:basic-properties}. In Section~\ref{sec:conditionaldistributions}, we derive a comprehensive set of conditional and joint weight distributions for vertices of Chung-Lu graphs and GIRGs. In Section~\ref{sec:positive}, we prove that degrees of neighboring vertices in Random Geometric Graphs are positively assortative and discuss degree-degree correlations in Chung-Lu graphs and GIRGs. Section~\ref{sec:experiments} contains experimental evaluations for a variety of real-world networks and for artificially generated networks from the network models. In Section~\ref{sec:pearson} we prove Theorem~\ref{thm:neg_pearson}, which shows that the Pearson correlation coefficient is negative for every networks whose degree distribution has a sufficiently heavy tail. We also discuss the use of alternative coefficients and report those coefficients for real-world networks.


\section{Graph Models}\label{sec:model}

We denote by $[n]$ the set $[n]:=\{1, \ldots, n\}$. We will often use the notations
\begin{align*}
    x \wedge y := \min\{x,y\} \qquad \text{and} \qquad x \vee y := \max\{x,y\}
\end{align*}
to denote minima and maxima. We consider simple undirected graphs with vertex set $\mathcal{V}=[n]$ and edge set denoted by $\mathcal{E}$. The set of neighbors of a vertex $v\in\mathcal{V}$ is denoted by $\Gamma(v)$. We use the short-hand notation uar for (events occurring) uniformly at random and iid when random variables are independent and identically distributed. In most graphs models we study, the degrees are distributed according to a power law, which is defined as follows.

\begin{definition}\label{def:power-law}
    Let $\tau>1$. A discrete random variable $X$ with values in $\N$ is said to follow a \emph{power-law with exponent} $\tau$ if $\pr(X = x) = \Theta(x^{-\tau})$ for $x\in \N$. A continuous random variable $X$ with values in $[1,\infty)$ is said to follow a \emph{power-law with exponent} $\tau$ if it has a density function $f_X$ satisfying $f_X(x) = \Theta(x^{-\tau})$.
\end{definition}

Our focus will be on scale-free graph models, i.e.\ graphs that have a power-law degree distribution with exponent $\tau\in(2,3)$. However, for the sake of generality, we give the definitions for any choice of $\tau>2$.
We start by defining Chung-Lu graphs, which are standard graphs yielding such a power-law degree distribution.

\begin{definition}[Chung-Lu graph~\cite{chung2002average}]\label{def:simple-chung-lu}
    Let $\tau>2$ and let $\mathcal{D}$ be a power-law distribution on $[1,\infty)$ with exponent $\tau$. A \emph{Chung-Lu graph} is obtained by the following two-step procedure:
    \begin{enumerate}[label=(\arabic*), leftmargin=1cm]
        \item Every vertex $v\in\mathcal{V}$ draws iid\ a \emph{weight} $w_v \sim \mathcal{D}$.

        \item For every two distinct vertices $u,v \in\mathcal{V}$, add an edge between $u$ and $v$ in $\mathcal{E}$ independently with probability
        \begin{align}\label{eq:chung-lu-connection}
        \Theta\Big(\min\Big\{\frac{w_uw_v}{n}, 1\Big\}\Big).
        \end{align}
    \end{enumerate}
\end{definition}

Classically, the $\Theta(\cdot)$ simply hides a factor 1, but this more general definition of the model also captures similar random graphs, like the Norros-Reittu model \cite{norros2006conditionally}, while important properties stay asymptotically invariant~\cite{janson2010asymptotic}.

Very frequently, real-world networks have an (implicit) underlying geometry. The Random Geometric Graph (RGG) model is the simplest model that includes geometry.
In RGGs, each vertex is assigned coordinates in an underlying ground space. Pairs of vertices are then connected independently of other pairs if their distance is below a global threshold\cite{penrose2003random}. We will take the $d$-dimensional unit hypercube $[0,1]^d$ equipped with the torus topology as the ground space. In particular, we define the distance between two points $x=(x_1, \ldots, x_d),x'=(x'_1, \ldots, x'_d)$ in $[0,1]^d$ as
\begin{align}\label{eq:torus-distance}
    \|x-x'\| := \max_{1 \le i \le d} \min\{|x_i-x_i'|, 1-|x_i-x_i'|\}.
\end{align}

We remark here that, while in the subsequent we fix the norm for our models to be the max-norm, all our results, except for the expected volume of the intersection of balls of influence in Proposition \ref{prop:rggassortativity}, also hold for any other choice of norm.
\begin{definition}[RGG]
\label{def:rgg}
   Fix a node set $\mathcal{V}$ of order $|\mathcal{V}|=n$ and a function $r = r(n) = \Theta(n^{-\frac{1}{d}})$. A threshold Random Geometric Graph is obtained by the following two-step procedure:
    \begin{enumerate}
        \item  Every node $v\in\mathcal{V}$ draws independently and uar\ a position $x_v$ in the hypercube $[0,1]^d$.
        \item  Connect each pair of distinct vertices $u,v \in\mathcal{V}$ by an edge iff $\|x_u-x_v\|\le r$
    \end{enumerate}
\end{definition}
Note that the condition on $r$ yields a sparse graph, that is, an expected number of edges which is linear in the number of vertices. We refer to the geometric region of points with distance at most $r$ from a node as its \emph{ball of influence} or \emph{box of influence}.

Geometric Inhomogeneous Random Graphs (GIRGs) combine the degree inhomogeneity of Chung-Lu graphs with the geometric component of RGGs. The vertices are assigned both a weight and a position in a given ground space. We will again take the $d$-dimensional unit hypercube $[0,1]^d$ equipped with the torus topology as ground space. 

\begin{definition}[GIRG~\cite{bringmann2019geometric}]\label{def:simple-girg}
    Let $\tau>2$, $\alpha>1$ and $d\in\mathbb{N}$ and let $\mathcal{D}$ be a power-law distribution on $[1,\infty)$ with exponent $\tau$. A \emph{Geometric Inhomogeneous Random Graph (GIRG)} is obtained by the following three-step procedure:
    \begin{enumerate}[label=(\arabic*), leftmargin=1cm]
        \item Every vertex $v\in\mathcal{V}$ draws iid\ a \emph{weight} $w_v \sim \mathcal{D}$.

        \item Every vertex $v\in\mathcal{V}$ draws independently and uar\ a position $x_v$ in the hypercube $[0,1]^d$.

        \item For every two distinct vertices $u,v \in\mathcal{V}$, add an edge between $u$ and $v$ in $\mathcal{E}$ independently with probability
        \begin{align}\label{eq:girg-connection}
            p_{uv} = \Theta\Big(\min\Big\{\frac{w_uw_v}{n \|x_u-x_v\|^d}, 1\Big\}^{\alpha}\Big).
        \end{align}
    \end{enumerate}
\end{definition}
We also allow $\alpha = \infty$ and in this case require that 
\begin{equation}\label{eq:puv2}
     p_{uv} = \begin{cases} \Theta(1), & \text{if } \|x_u - x_v\| \le O\big(\big(\tfrac{w_u w_v} n\big)^{1/d}\big), \\ 0, & \text{if } \|x_u - x_v\| \ge \Omega\big(\big(\tfrac{w_u w_v} n\big)^{1/d}\big), \end{cases} 
 \end{equation}
where the constants hidden by $O$ and $\Omega$ do not have to match, i.e., there can be an interval $[c_1 (\tfrac{w_u w_v}n)^{1/d}, c_2 (\tfrac{w_u w_v}n)^{1/d}]$ for $\|x_u - x_v\|$ where the behaviour of $p_{uv}$ is arbitrary.

We now introduce extensions of the Chung-Lu and the GIRG models that allow to modulate the assortativity of the graph. This is achieved by varying the influence of the smaller weight in the connection probability via a parameter $0 \le \sigma \le \tau-1$. We note here that the choice $\sigma = 1$ recovers the original models. We begin with the extension of the Chung-Lu model, which we call Tunable Chung-Lu graphs. 

\begin{definition}[Tunable Chung-Lu graph]\label{def:simple-assortative-chung-lu}
    Let $\tau>2$, $\sigma\in[0,\tau-1)$ and let $\mathcal{D}$ be a power-law distribution on $[1,\infty)$ with exponent $\tau$. A \emph{Tunable Chung-Lu graph} is obtained by the following two-step procedure:
    \begin{enumerate}[label=(\arabic*), leftmargin=1cm]
        \item Every vertex $v\in\mathcal{V}$ draws iid\ a \emph{weight} $w_v \sim \mathcal{D}$.

        \item For every two distinct vertices $u,v \in\mathcal{V}$, add an edge between $u$ and $v$ in $\mathcal{E}$ independently with probability
        \begin{align}\label{eq:assortative-chung-lu-connection}
            \Theta\Big(\min\Big\{\frac{(w_u \wedge w_v)^{\sigma}(w_u \vee w_v)}{n}, 1\Big\}\Big).
        \end{align}
    \end{enumerate}
\end{definition}

Next we define the analogous extension of GIRGs, which we call TGIRGs. This model has been studied in \cite{gracar2022finiteness,jorritsma2023cluster,jorritsma2024large,hofstad2023scaling} and by Lüchtrath in his PhD Thesis~\cite{luchtrath2022percolation}, but never in the context of assortativity. 

\begin{definition}[TGIRG]\label{def:simple-agirg}
    Let $\tau>2$, $\alpha>1$, $\sigma\in[0,\tau-1)$ and $d\in\N$ and let $\mathcal{D}$ be a power-law distribution on $[1,\infty)$ with exponent $\tau$. A \emph{Tunable Geometric Inhomogeneous Random Graph (TGIRG)} is obtained by the following three-step procedure:
    \begin{enumerate}[label=(\arabic*), leftmargin=1cm]
        \item Every vertex $v\in\mathcal{V}$ draws iid\ a \emph{weight} $w_v \sim \mathcal{D}$.

        \item Every vertex $v\in\mathcal{V}$ draws independently and uar\ a position $x_v$ in the hypercube $[0,1]^d$.

        \item For every two distinct vertices $u,v \in\mathcal{V}$, add an edge between $u$ and $v$ in $\mathcal{E}$ independently with probability
        \begin{align}\label{eq:agirg-connection}
             \Theta\Big(\min\Big\{\frac{(w_u \wedge w_v)^{\sigma}(w_u \vee w_v)}{n \|x_u-x_v\|^d}, 1\Big\}^{\alpha}\Big).
        \end{align}
    \end{enumerate}
\end{definition}

Analogously to GIRGs, also here we allow $\alpha = \infty$, requiring in  this case that 
\begin{equation}\label{eq:puv2_sigma}
     p_{uv} = \begin{cases} \Theta(1), & \text{if } \|x_u - x_v\| \le O\big(\big(\tfrac{(w_u \wedge w_v)^\sigma (w_u \vee w_v)} n\big)^{1/d}\big), \\ 0, & \text{if } \|x_u - x_v\| \ge \Omega\big(\big(\tfrac{(w_u \wedge w_v)^\sigma (w_u \vee w_v)} n\big)^{1/d}\big). \end{cases} 
 \end{equation}

We remark that for $\sigma=0$ we recover the (soft) Boolean model \cite{hall1985continuum, yukich2006ultra, gracar2022chemical}, which is basically a Random Geometric Graph where each vertex has an associated random radius drawn from a power-law distribution, while for $\sigma=\tau-2$ we recover the age-dependent random connection model introduced in \cite{gracar2022recurrence} as an approximation to the spatial preferential attachment model \cite{jacob2015spatial}. 
We also note that (T)GIRGs sometimes use a slightly different parametrization, replacing the long-range parameter $\alpha\in(1,\infty]$ by its inverse $T:=\tfrac{1}{\alpha}\in[0,1)$, called the \emph{temperature} of the model.

\section{General Properties Of The Graph Models}\label{sec:basic-properties}

In this section, we describe some general and important properties of Chung-Lu graphs and GIRGs. They hold in fact for a general class of graph models described in~\cite{bringmann2024average}. 
They also generalize to Tunable Chung-Lu graphs and TGIRGs (remember that we get the original Chung-Lu graph model and GIRG model by setting $\sigma=1$ in the definition of the corresponding tunable model), hence for the sake of conciseness we state each result directly for the more general tunable models.

The first result gives the marginal probability that an edge between two vertices with given weights is present in TGIRGs.

\begin{lemma}\label{lem:tgirg-marginal}
Let $\mathcal{G}=(\mathcal{V},\mathcal{E})$ be a TGIRG and $v \in\mathcal{V}$ be a vertex with fixed position $x_v \in [0,1]^d$. Then all edges $uv$ for $u \ne v$ are independently present with probability
\begin{align}\label{eq:tgirg-marginal-edge-prob-position}
    \Pr(uv\in\mathcal{E} \mid x_u, w_u, w_v) = \Theta\Big(\min\Big\{\frac{(w_u\wedge w_v)^\sigma (w_u \vee w_v)}{n}, 1\Big\}\Big). 
\end{align}
In particular, we can also remove the conditioning on $x_u$ (by integrating over $[0,1]^d$) and get
\begin{align}\label{eq:tgirg-marginal-edge-prob}
    \Pr(uv\in\mathcal{E} \mid w_u, w_v) = \Theta\Big(\min\Big\{\frac{(w_u\wedge w_v)^\sigma (w_u \vee w_v)}{n}, 1\Big\}\Big). 
\end{align}
\end{lemma}

\begin{proof}
    For classical GIRGs, this corresponds to Lemma~4.2 and Theorem~7.3 in~\cite{bringmann2024average}. We give the proof for TGIRGs in the supporting information.
\end{proof}

The next lemma says that the expected degree of a vertex is of the same order as its weight, thus allowing us to treat a given weight sequence $(w_v)_{v\in\mathcal{V}}$ as a sequence of expected degrees.

\begin{lemma}\label{lem:exp_deg_sigma}
Let $\mathcal{G}=(\mathcal{V},\mathcal{E})$ be a Tunable Chung-Lu graph or a TGIRG and $v \in\mathcal{V}$ be a vertex with fixed weight $w_v \in [1,\infty)$. Then we have $\E[\deg(v) \mid w_v]=\Theta(\min\{w_v,n\})$, or equivalently $\pr(uv\in\mathcal{E} \mid w_v) =\Theta(\min\{\tfrac{w_v}{n},1\})$ for a vertex $u$ chosen uniformly at random in $\mathcal{V}$.
\end{lemma}

\begin{proof}
    For classical GIRGs and Chung-Lu graphs, this corresponds to Lemma~4.3 in~\cite{bringmann2024average}. The general result was proved using a different parametrization in \cite{luchtrath2022percolation}. For convenience of the reader and in the spirit of self-containment, we provide a proof with our parametrization in the supporting information.
\end{proof}

By observing that $\E[w_v] = \Theta(1)$ for all vertices $v\in\mathcal{V}$, this then naturally yields the connection probability between two random vertices.

\begin{lemma}\label{lem:constant-average-degree}
    Let $\mathcal{G}=(\mathcal{V},\mathcal{E})$ be a Tunable Chung-Lu graph or a TGIRG and $u,v \in\mathcal{V}$ be vertices chosen uniformly at random. Then $\pr(uv\in\mathcal{E}) =\Theta(1/n)$.
\end{lemma}

For vertices of sufficiently large weight, their degree is concentrated around its expectation, as stated precisely in the next lemma.

\begin{lemma}\label{lem:degree-concentration}
    Let $\mathcal{G}=(\mathcal{V},\mathcal{E})$ be a Tunable Chung-Lu graph or a TGIRG. Then the following properties hold with probability $1-n^{-\omega(1)}$:
    \begin{enumerate}[label=(\roman*)]
    \item $\deg(v) = O(w_v + \log^2 n)$ for all $v \in\mathcal{V}$.
    \item $\deg(v)= (1+o(1))\E[\deg(v)]= \Theta(w_v)$ for all $v \in\mathcal{V}$ with $w_v = \omega(\log^2 n)$.
    \end{enumerate}
\end{lemma}

\begin{proof}
    For classical GIRGs and Chung-Lu graphs, this corresponds to Lemma~4.4 in~\cite{bringmann2024average}. The same proof can be used for the tunable models.
\end{proof}

The above results imply in particular that Chung-Lu graphs and GIRGs have a power-law degree distribution with exponent $\tau$. Next we will show that, just like classical GIRGs, TGIRGs exhibit clustering. Before stating the result formally, we give the required definition of the (average) clustering coefficient of a graph.

\begin{definition}\label{def:clustering-coefficient}
In a graph $\mathcal{G}=(\mathcal{V},\mathcal{E})$ the \emph{local clustering coefficient} of a vertex $v \in \mathcal{V}$ is defined as
\begin{align*}
    \mathrm{cc}(v) := \mathrm{cc}_\mathcal{G}(v) := \begin{cases}
       \frac{|\{\{u,u'\}\subseteq \Gamma(v), uu'\in\mathcal{E}\}|}{\binom{\deg(v)}{2}} & \text{if $\deg(v)\ge 2$,} \\
       0 & \text{otherwise,}
    \end{cases}        
\end{align*}
and the \emph{average clustering coefficient}, or simply \emph{clustering coefficient}, of $\mathcal{G}$ is given by
\begin{align*}
    \mathrm{cc}(\mathcal{G}) := \frac{1}{|\mathcal{V}|}\sum_{v \in \mathcal{V}} \mathrm{cc}(v).
\end{align*}
\end{definition}

Proposition \ref{prop:clustering-tgirg} below then states that the average clustering coefficient of TGIRGs is constant. 

Note that the statement places no further restrictions on $\sigma$. In particular, for $\sigma=1$, we recover the classical GIRG model for which the analogous result was proven originally in~\cite{bringmann2019geometric}. We remark that in Tunable Chung-Lu graphs the situation is completely different and due to the lack of geometry their clustering coefficient tends to 0 with increasing $n$.

\begin{proposition}\label{prop:clustering-tgirg}
    Let $\mathcal{G}=(\mathcal{V},\mathcal{E})$ be a TGIRG. Then whp its clustering coefficient satisfies $\mathrm{cc}(\mathcal{G})=\Theta(1)$.
\end{proposition}

\begin{proof}
Consider a vertex $v\in\mathcal{V}$ of weight $w_v\in[1,2]$, and define $U(v)$ as the subset of points which have geometric distance at most $n^{-1/d}$ from $v$. Given that all weights are $\ge 1$, the connection probability \eqref{eq:agirg-connection} for each pair of vertices in $U(v)$ is at least constant - and this holds in particular for all $\sigma \in [0, \tau-1)$. Therefore, the probability that $v$ has at least two neighbors $u,u'\in U(v)$ can be lower bounded by a constant $\rho$ that is independent of $w_v$. 

Define $E_2 = \Theta(1)$ to be the expected degree of a vertex of weight $2$, and note that $\E[\deg(v) \mid w_v] \le E_2$ since $w_v \le 2$. By Markov's inequality we deduce that $\pr(\deg(v)>2E_2/\rho) \le \rho/2$. A union bound then yields
\begin{align*}
    \pr(\deg(v) > 2E_2/\rho \text{ or } |\Gamma(v) \cap U(v)| \le 1 \mid w_v) \le 1-\rho+\rho/2 = 1-\rho/2.
\end{align*}
By taking the complement, we conclude that, with probability at least $\rho/2$, $v$ has at least two neighbors $u,u'\in U(v)$ and $\deg(v) \le 2E_2/\rho$. Therefore, remembering Definition \ref{def:clustering-coefficient} we get
\begin{align*}
    \E[\mathrm{cc}(\mathcal{G})] \ge \pr(w_v\in[1,2]) \cdot \rho/2 \cdot \pr(uu'\in \mathcal E \mid u,u'\in U(v)) \cdot \binom{2E_2/\rho}{2}^{-1}.
\end{align*}
Note that each term is of order $\Theta(1)$, and hence $\E[\mathrm{cc}(\mathcal{G})]=\Theta(1)$. Concentration can be shown using an Azuma-Hoeffding-type bound, see~\cite{bringmann2019geometric} for details.
\end{proof}

\section{Conditional Distributions Of Weights} \label{sec:conditionaldistributions}

In this section, we derive conditional and joint weight distributions for Tunable Chung-Lu graphs and TGIRGs. Note that all results directly apply to classical Chung-Lu graphs and GIRGs (often with a simplified statement) by taking $\sigma=1$. We begin with the weight distribution of a (vertex) endpoint of an edge drawn uniformly at random.

\begin{proposition} \label{prop:shifted-weight-density-tuned}
    Let $\tau>2$ $\sigma\in[0,\tau-1)$ and let $\mathcal{G} = (\mathcal{V}, \mathcal{E})$ be a Tunable Chung-Lu graph or a TGIRG whose weight distribution $\mathcal{D}$ is a power-law with exponent $\tau$. Denote by $f_{W_v}(w \mid uv\in\mathcal{E})$ the density of the weight of the endpoint of an edge $uv$ chosen uniformly at random in $\mathcal{E}$. Then 
\begin{align}\label{eq:shifted-weight-distribution-assortative}
    f_{W_v}(w \mid uv\in\mathcal{E})=
    \begin{cases}
    \Theta(w^{1-\tau}) & \text{if $w \le  n$,} \\
    \Theta(w^{-\tau}n) & \text{if  $w > n$.}
    \end{cases}
\end{align}
\end{proposition}

\begin{proof}
Let $f_{W_v}(w)$ denote the density function of the weight $w_v$ of a vertex $v\in\mathcal{V}$ chosen uniformly at random, so that $f_{W_v}(w)=\Theta(w^{-\tau})$ by assumption. Using Bayes' formula we have
\begin{align*}
    f_{W_v}(w \mid uv\in\mathcal{E}) = \frac{\pr(uv\in\mathcal{E} \mid W_v = w)f_{W_v}(w)}{\pr(uv\in\mathcal{E})}.    
\end{align*}
Note that $\pr(uv\in\mathcal{E} \mid W_v = w) = \Theta((w \wedge n)/n)$ by Lemma~\ref{lem:exp_deg_sigma} and $\pr(uv\in\mathcal{E}) = \Theta(1/n)$ by Lemma~\ref{lem:constant-average-degree}. Hence the above equation becomes
\begin{align*}
    f_{W_v}(w \mid uv\in\mathcal{E}) = \Theta\Big(\frac{(w \wedge n)/n \cdot w^{-\tau}}{1/n}\Big) = \Theta\big((w \wedge n)w^{-\tau}\big),    
\end{align*}
which concludes the proof.
\end{proof}

Note that the second case in Eq \eqref{eq:shifted-weight-distribution-assortative} basically never occurs, since whp\ the largest weight is of order $n^{1/(\tau-1)}=o(n)$. Proposition \ref{prop:shifted-weight-density-tuned} therefore tells us that the weight of the endpoint of a random edge follows a power-law with exponent $\tau-1$ (as opposed to the weight of random vertex, which also follows a power-law, but with exponent $\tau$).

The next proposition gives the joint weight density of two endpoints $u,v$ of a randomly chosen edge.

\begin{proposition}\label{prop:joint-weight-density-tuned}
    Let $\tau>2$ $\sigma\in[0,\tau-1)$ and let $\mathcal{G} = (\mathcal{V}, \mathcal{E})$ be a Tunable Chung-Lu graph or a TGIRG whose weight distribution $\mathcal{D}$ is a power-law with exponent $\tau$. Denote by $f_{(W_u,W_v)}(w_u, w_v \mid uv\in\mathcal{E})$ the joint density of the weights of the endpoints of an edge $uv$ chosen uniformly at random in $\mathcal{E}$. Then 
\begin{align*}
    f_{(W_u,W_v)}(w_u, w_v \mid uv\in\mathcal{E})=
    \begin{cases}
    (w_u\wedge w_v)^{\sigma-\tau} (w_u \vee w_v)^{1-\tau} & \mbox{if }  (w_u \wedge w_v)^{\sigma}(w_u \vee w_v) \le  n, \\
    n(w_uw_v)^{-\tau} & \mbox{if }  (w_u \wedge w_v)^{\sigma}(w_u \vee w_v) > n,
    \end{cases}
\end{align*}
i.e.\ $f_{(W_u,W_v)}(w_u, w_v \mid uv\in\mathcal{E}) =  \Theta\big(( (w_u \wedge w_v)^\sigma (w_u \vee w_v) \wedge n)\cdot(w_uw_v)^{-\tau}\big)$.
\end{proposition}

\begin{proof}
Let $f_{W_u,W_v}(w_u,w_v)$ denote the joint density function of the weights $W_u$ and $W_v$ of two vertices $u,v\in\mathcal{V}$ chosen uniformly at random, so that $f_{W_u,W_v}(w_u,w_v)=\Theta((w_uw_v)^{-\tau})$ since the weights are iid. Using Bayes' formula we have
\begin{align*}
    f_{W_v,W_v}(w_u, w_v \mid uv\in\mathcal{E}) = \frac{\pr(uv\in\mathcal{E} \mid W_u = w_u, W_v = w_v)f_{W_u,W_v}(w_u,w_v)}{\pr(uv\in\mathcal{E})}.  
\end{align*}
Now by Lemma~\ref{lem:tgirg-marginal} and Lemma~\ref{lem:constant-average-degree}, this becomes
\begin{align*}
    f_{W_u,W_v}(w_u, w_v \mid uv\in\mathcal{E}) = \Theta\Big(\frac{((w_u\wedge w_v)^\sigma (w_u \vee w_v) \wedge n)/n \cdot (w_uw_v)^{-\tau}}{1/n}\Big)  \\ =\begin{cases}
    (w_u\wedge w_v)^{\sigma-\tau} (w_u \vee w_v)^{1-\tau} & \mbox{if }  (w_u \wedge w_v)^{\sigma}(w_u \vee w_v) \le  n, \\
    n(w_uw_v)^{-\tau} & \mbox{if }  (w_u \wedge w_v)^{\sigma}(w_u \vee w_v) > n.
    \end{cases}
\end{align*}
\end{proof}

We remark here that the case distinction stems from the cutoff at 1 in the connection probability. There is a multigraph version of the model where we do not cap the expression at 1 but simply require that the expected number of edges connecting $u$ and $v$ is $\frac{(w_u \wedge w_v)^{\sigma}(w_u \vee w_v)}{n}$. For the vast majority of the vertices, this expectation will naturally be at most $1$, and for the remaining vertex pairs, one can for example randomly assign $\lfloor \frac{(w_u \wedge w_v)^{\sigma}(w_u \vee w_v)}{n} \rfloor$ or $\lceil\frac{(w_u \wedge w_v)^{\sigma}(w_u \vee w_v)}{n} \rceil$ edges, with the probabilities chosen to match the desired expectation. In this multigraph version, the first line in the display equation in Proposition \ref{prop:joint-weight-density-tuned} holds for all weights, not just for weights with $(w_u \wedge w_v)^{\sigma}(w_u \vee w_v) \le n$. In the simple graph version that is the focus of this paper, the first line of the display still applies for a large majority of vertex pairs. Note that for the classical Chung-Lu and GIRG models (obtained by choosing $\sigma=1$), this line simplifies to $f_{W_u,W_v}(w_u, w_v \mid uv\in\mathcal{E}) = w_u^{1-\tau} w_v^{1-\tau}$, and by Proposition \ref{prop:shifted-weight-density-tuned} this is the same as $f_{W_u}(w_u \mid uv\in\mathcal{E}) \cdot f_{W_v}(w_v \mid uv\in\mathcal{E})$. In other words, (under the cutoff) the two endpoint weights are independent, which already is a strong indication of the neutral assortativity of these classical models. The next proposition gives the weight density of a vertex $v$ as a function of the weight $w_u$ of its neighbor $u$, which gives more insight in the assortative behavior of Tunable Chung-Lu graphs and TGIRGs in general.

\begin{proposition}\label{prop:shifted-conditional-weight-distribution-tuned}
Let $\tau\in(2,3)$, $\sigma\in[0,\tau-1)$ and let $\mathcal{G} = (\mathcal{V}, \mathcal{E})$ be a Tunable Chung-Lu graph or a TGIRG whose weight distribution $\mathcal{D}$ is a power-law with exponent $\tau$. 
Consider an edge $uv \in \mathcal{E}$ connecting a vertex $u$ with weight $w_u$ to some vertex $v$. Then the conditional distribution of the weight $W_v$ of $v$ satisfies
\begin{align*}
    f_{W_v}(w \mid w_u, uv\in\mathcal{E})=
    \begin{cases}
    \Theta(w^{1-\tau}(w\wedge w_u)^{\sigma-1}) & \text{if $(w \wedge w_u)^{\sigma}(w \vee w_u) \le n$,} \\
    \Theta(\tfrac{w^{-\tau} n}{w_u}) & \text{if  $w_u \le n \le (w \wedge w_u)^{\sigma}(w \vee w_u)$,} \\
    \Theta(w^{-\tau}) & \text{if $w_u > n$.}
    \end{cases}
\end{align*}
\end{proposition}

\begin{proof}
From the definition of the conditional density, and using Propositions \ref{prop:shifted-weight-density-tuned} and \ref{prop:joint-weight-density-tuned}, we get
    \begin{align*}
    f_{W_v}(w \mid w_u, uv\in\mathcal{E})= \frac{f_{W_v,W_v}(w, w_u \mid uv\in\mathcal{E})}{f_{W_u}(w_u \mid uv\in\mathcal{E})} = \Theta\Big(\frac{((w \wedge w_u)^\sigma (w \vee w_u)\wedge n) w^{-\tau}}{w_u \wedge n}\Big),
\end{align*}
which concludes the proof.
\end{proof}

Again, the last case of the display equation in Proposition \ref{prop:shifted-conditional-weight-distribution-tuned} is basically void since whp\ all weight are much smaller than $n$. Moreover, the second case only applies to exceptionally large weights. For the classical choice $\sigma=1$, the first case gives $f_{W_v}(w \mid w_u, uv\in\mathcal{E})= \Theta(w^{1-\tau})$ and in particular this conditional distribution is actually independent of $w_u$, which shows the neutral assortativity of Chung-Lu graphs and GIRGs. Compared to this baseline, we see that by choosing $\sigma>1$ the distribution gets a boost of order $(w\wedge w_u)^{\sigma-1}$. In particular, this boosting factor is increased for larger $w$ and $w_u$, which translates to positive assortativity. On the other hand, choosing $\sigma<1$ makes the exponent $\sigma-1$ negative, and thus yields an opposite effect, which shows the disassortative behavior of the corresponding graph models.

\section{Degree Assortativity In RGGs, Chung-Lu Graphs And GIRGs}\label{sec:positive}

So far, we have analyzed assortativity with respect to the vertex weights of our random graph models. However, our ulterior motive is to understand assortativity with respect to vertex \emph{degrees}. 
In the graph models of interest here, the weight of a vertex corresponds to its expected degree, and hence Propositions \ref{prop:shifted-weight-density-tuned}-\ref{prop:shifted-conditional-weight-distribution-tuned} are also informative about the degree assortativity in typical instances of Chung-Lu graphs and GIRGs. Moreover, the degrees of vertices of sufficiently large (polylogarithmic) weights are whp\ of the same order as their weight (see e.g.\ Lemma \ref{lem:degree-concentration}), making these vertices completely assortativity-neutral in terms of degree as well.

In GIRGs however, there is an additional effect caused by random vertex-density fluctuations in certain regions, i.e., by chance, some regions in the ground space of the GIRG may contain more nodes. These effects are noticeable for low-degree nodes. In order to explain this effect, we consider the simpler model of Random Geometric Graphs. We note here that the same analysis applies to GIRGs but comes at the price of much higher technicality. We omit the details but note qualitatively that the larger the weights of the involved nodes, the smaller the effect.

\subsection{Degree Assortativity in Random Geometric Graphs}

In the following, we derive the conditional degree distribution for vertices given the degree of one of their neighbors in Random Geometric Graphs, demonstrating that degrees of neighboring vertices are positively correlated. In this section, we will denote by $B_r(x)$ the ball of radius $r$ centered at $x \in [0,1]^d$. By $\Vol{(A)}$ we denote the Euclidean volume of a set $A \subseteq [0,1]^d$ and by $Po(\mu)$, following the usual convention, a Poisson random variable with mean $\mu$. Recall from our definition of RGGs that the latent space we are using is equipped with the torus topology and distances are measured using the max-norm $\|\cdot\|_{\infty}$.
\begin{proposition}
\label{prop:rggassortativity}
Let $\mathcal{G}=(\mathcal{V}, \mathcal{E})$ be a Random Geometric Graph on $n$ vertices in the $d$-dimensional hypercube, let $uv\in \mathcal{E}$ be an edge selected uniformly at random from $\mathcal{E}$, and let the degree of $u$ be $\deg(u)=k$. Then $\deg(v)$ is distributed as
\[
\deg(v) \sim 1+ Po\Big((k-1)\frac{\mathrm{Vol}{(B_r(x_u)\cap B_r(x_v))}}{\mathrm{Vol}{(B_r(x_u))}}\Big)+ Po\Big((n-1-k) \frac{\mathrm{Vol}{(B_r(x_u) \setminus B_r(x_v))}}{\mathrm{Vol}{([0,1]^d)}}\Big),
\]
where 
\[
\mathbb{E}\Big[\frac{\mathrm{Vol}(B_r(x_u)\cap B_r(x_v))}{\mathrm{Vol}(B_r(x_u))}\Big]=\Big(\frac{3}{4}\Big)^d.
\]
In particular, $\deg(v)$ can be stochastically lower-bounded by a random variable $X$ distributed as
\[
X \sim 1+ Po(\Theta(k)).
\]
\end{proposition}

\begin{proof}
Take an edge $uv \in \mathcal{E}$ drawn uniformly at random, let $\deg(u)=k$ and consider $\deg(v)$. Let $Y$ be a random variable distributed as 
\begin{align*}
  Y \sim 1+ Po\Big((k-1)\frac{\mathrm{Vol}(B_r(x_u)\cap B_r(x_v))}{\mathrm{Vol}(B_r(x_u))}\Big)+ Po\Big((n-1-k) \frac{\mathrm{Vol}(B_r(x_u) \setminus B_r(x_v))}{\mathrm{Vol}([0,1]^d)}\Big).  
\end{align*}
We claim that $Y$ is indeed distributed as $\deg(v)$.
Observe that conditioning on $uv \in \mathcal{E}$ yields an additive $1$ term to $\deg(v)$, and $v$ might connect to the $n-2$ vertices in the set $\mathcal{V}\setminus\{u,v\}$. Consider now the intersection of the boxes of influence of $u$ and $v$ (see Fig~\ref{fig:geometric-assortativity} for reference). We know that $u$ has precisely $\deg(u)-1$ many neighbors besides $v$, and  $Po((k-1) \cdot \mathrm{Vol}(\frac{B_r(x_u)\cap B_r(x_v)}{B_r(x_u)})$ many of them lie in $B_r(x_u)\cap B_r(x_v)$ and are therefore also neighbors of $v$. 
Finally, there are exactly $n-1-k$ many non-neighbors of $u$, which lie outside of $u$'s box of influence $B_r(x_u)$. Among these vertices, all which additionally lie in $B_r(x_v)$ are neighbors of $v$, hence contributing a term of $Po((n-1-k)\cdot \mathrm{Vol}(B_r(x_u) \setminus B_r(x_v)))$. 
This concludes the proof of the first statement since $\mathrm{Vol}([0,1]^d)=1$.

We proceed to compute the expected volume of $B_r(x_u)\cap B_r(x_v)$, conditioned on $u$ and $v$ being neighbors. First note that
by definition
both vertices have boxes of influence of the same volume. Furthermore, $u$ and $v$ are connected if and only if they both lie in each other's box of influence, or equivalently, $x_u, x_v\in B_r(x_u)\cap B_r(x_v)$. Observe now, that, conditional on $uv \in \mathcal{E}$, the position $x_v$ of $v$ is uniformly random in $B_r(x_u)$. We first compute the volume of the expected intersection in dimension $d=1$ and then extend the computation to arbitrary $d$.

Let now $d=1$, and consider hence $B_r(x_u)$, an interval of length $r=\Theta(1/n^{1/d})=\Theta(1/n)$ in $[0,1]$. Then the volume (length) of the intersection of the boxes (intervals) of influence is given by
\begin{align*}
    \mathrm{Vol}(B_r(x_u) \cap B_r(x_v))=r-|x_u-x_v|,
\end{align*}
and note that, conditioned on $uv\in\mathcal{E}$, the distance $|x_u-x_v|$ is distributed uniformly in $[0,r/2]$. We can thus compute the expected volume as follows
\begin{align*}
   \E[\mathrm{Vol}(B_r(x_u) \cap B_r(x_v)) \mid uv\in\mathcal{E}]
   = r -\E[|x_u-x_v| \mid uv\in\mathcal{E}]
   = r-\frac{r}{4} = \frac{3r}{4}.
\end{align*}

We now generalize this to higher dimensions $d$. Denote the coordinates of $x_v$ by $(x_v^1, \ldots , x_v^d)$. Recall that, for $i\neq j\in \{1,\dots, d\}$ the random variables $x_v^i, x_v^j$ are independent. Observe furthermore that the total volume $\mathrm{Vol}(B_r(x_u) \cap B_r(x_v))$ is just the product of the length of the coordinate-wise intersections. Hence, due to the independence of the coordinates, we can compute the expected volume by taking the $d$-fold product of the expectation in the display above. More formally, denote by $B_{x_i}$ the intersection of the boxes of influence of $u$ and $v$ in the $i$-th coordinate, for $i=1,\dots, d$. Then we have
\[
\E[\mathrm{Vol}(B_r(x_u) \cap B_r(x_v)) \mid uv\in\mathcal{E}]=\E\Big[\prod_{i=1}^d B_{x_i} \mid uv\in\mathcal{E}\Big]= \prod_{i=1}^d \E[B_{x_i} \mid uv\in\mathcal{E}]= \frac{3^dr^d}{4^d}.
\]
Since $\mathrm{Vol}(B_r(x_u)) = r^d$, this concludes the proof of the second statement.

Finally, the last statement is a simple corollary of the first statement. It follows by noticing that
\begin{align*}
     Y &\sim  1+ Po\Big((k-1)\frac{\mathrm{Vol}(B_r(x_u)\cap B_r(x_v))}{\mathrm{Vol}(B_r(x_u))}\Big)+ Po\Big((n-1-k) \frac{\mathrm{Vol}(B_r(x_u) \setminus B_r(x_v))}{\mathrm{Vol}([0,1]^d)}\Big) \\
     &\ge 1+ Po\Big((k-1)\frac{\mathrm{Vol}(B_r(x_u)\cap B_r(x_v))}{\mathrm{Vol}(B_r(x_u))}\Big) = 1+ Po(\Theta(k)),
\end{align*}
where the last step holds because $\mathrm{Vol}(B_r(x_v) \cap B_r(x_u)) \ge 2^{-d}\mathrm{Vol}(B_r(x_u))$ deterministically.
\end{proof}

\begin{figure}
\centering

\begin{tikzpicture}[scale=0.2]

\fill[pattern={Lines[angle=-45, distance=10pt]}, pattern color=orange] (-14,-14) rectangle (24,14);
\filldraw[fill=white] (-10,-10) rectangle (10,10);

\filldraw[black] (0,0) circle (5pt) node[anchor=east]{$u$};
\filldraw[black] (7,2) circle (5pt) node[anchor=west]{$v$};
\draw (-10,-10) rectangle (10,10);
\draw (-3,-8) rectangle (17,12);
\draw[very thick] (0,0) -- (7,2);
\node[anchor=north] (edge) at (3.5,1){$e$};

\fill[pattern={Lines[angle=45, distance=10pt]}, pattern color=violet] (-10,-10) rectangle (10,10);

\begin{scope}
  \clip (-10,-10) rectangle (10,10);
  \draw[blue, very thick] (-3,-8) rectangle (17,12);
\end{scope}

\begin{scope}
  \clip (-3,-8) rectangle (17,12);
  \draw[blue, very thick] (-10,-10) rectangle (10,10);
\end{scope}

\end{tikzpicture}

\caption{Illustration of the intersection of the boxes of influence $B_r(x_u) \cap B_r(x_v)$ of two vertices $u,v$. If we know that $\deg(u) = k$, then the density of vertices inside $B_r(x_u)$ (area shaded in purple) is proportional to $k$, while the density of vertices outside $B_r(x_u)$ (shaded in orange) decreases (slightly) with $k$. Note that the intersection $B_r(x_u) \cap B_r(x_v)$ (delimited in blue) makes up a constant fraction of the area of $B_r(x_v)$, which leads to positive assortativity in RGGs.
}

\label{fig:geometric-assortativity}
\end{figure}
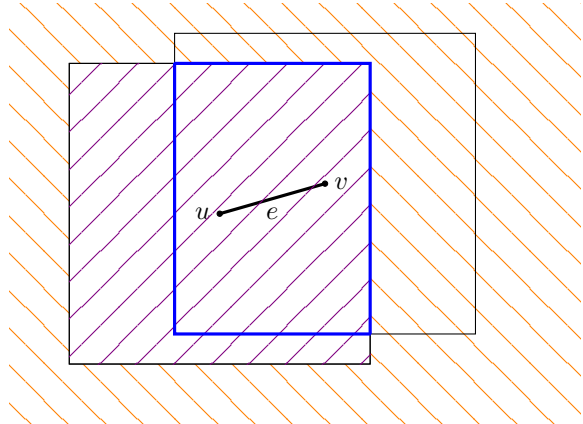

Proposition \ref{prop:rggassortativity} tells us that RGGs have positive degree assortativity in the following sense. The degree of a random neighbor $v$ of some vertex $u$ can be decomposed into two independent Poisson random variables. The expectation of the first one does not change with $\deg(u)$, while the expectation of the second one scales linearly with $\deg(u)$.

\subsection{Degree Assortativity in GIRGs} \label{subsec:degreeassortativity}

We are now in a position to summarize our findings. We can distinguish the following cases. Pick a vertex $u$ of weight $w_u$ and one of its neighbors $v$ of weight $w_v$. 
\begin{itemize}
    \item Whenever $w_u\cdot w_v \le n$, the conditional weight distribution of $w_v$ is independent of $w_u$ (Proposition~\ref{prop:shifted-conditional-weight-distribution-tuned})
    . Hence the only assortative effect on the degree of $v$ comes from the geometric (positive) assortativity: The degree of a vertex conditional on the degree of its neighbor depends positively on the intersection of their balls of influence - whereby here the ball of influence is the geometric region (which is a function of $w_v$) in which $v$ connects to every vertex with at least constant probability.\footnote{Note that whenever $w_v \gg w_u$, then the ball of influence of $u$ is typically contained in the ball of influence of $v$.} For constant-weight vertices, this effect of the geometry indeed plays a prominent role. This influence decays as $w_v$ grows. Since whp\ no weight is of order $\Omega(n)$, this scenario occurs in particular whenever $w_v=O(1)$.
    \item If $w_u\cdot w_v > n$, the conditional density resp.\ tail of the distribution of $w_v$ decays (proportionally) as $w_u$ increases. Since, as mentioned, degrees are concentrated for large weights, our theoretical results suggest that the negatively assortative effect of the weights exceed the positive (but vanishing) assortative effect of the geometry.
    \item If either $w_v$ or $w_u$ is larger than $n$, the weights have no assortative effect anymore (cf.\ Proposition~\ref{prop:shifted-conditional-weight-distribution-tuned}). 
    However, whp\ this case does not occur across the whole graph.
\end{itemize}

\section{Experiments}\label{sec:experiments}

Here we complement our theoretical considerations with an empirical evaluation.
Besides other interesting observations, our experiments in particular support our claims that (i) our tunable models indeed control the assortativity in the desired way and that (ii) using just a single number to describe how vertices of different degrees tend to connect does not do the rich underlying structure justice.

We start in Section~\ref{sec:eval-assort-coeff} by studying assortativity coefficients based on the correlation measures between the two endpoints of an edge.
As correlation measures, we consider Pearson's, Spearman's, and variants of Kendall's coefficients.
Our analysis confirms previous findings \cite{litvak2013uncovering} that Spearman's rank correlation coefficient (and similarly Kendall's) provides a more appropriate measure of assortativity than Pearson's correlation coefficient, although we will see in Section ~\ref{sec:eval-joint-distribution} that reducing the measure assortativity  to a single number nonetheless has serious limitations. Additionally, we show how the $\sigma$ parameter in our tunable models can be leveraged to regulate these assortativity coefficients.

In Section~\ref{sec:eval-joint-distribution}, we take a more fine-grained look at the joint distribution of vertex degrees.
This provides deeper insights into how vertices of different degrees tend to connect that go beyond what can be expressed with just a single number.

\paragraph{Overview Of Real-World Networks.}
Before presenting the experimental results, we provide brief background information on the real-world network dataset examined here; also see Table~\ref{tab:real-world-graphs}. More detailed descriptions of the individual networks can be found in the Supporting Information. In order to allow for meaningful results, we restrict our analysis to networks which are verifiably power-law. The rigorous detection of power laws in empirical degree distributions in real-world networks is a challenging task -- hence we rely on the evaluations of~\cite{voitalov2019scale}, who analyzed 35 undirected real-world networks without multi-edges or loops collected from the KONECT database~\cite{kunegis2013konect}. The methodology of~\cite{voitalov2019scale} ensures in particular statistically consistent and robust estimations of the tail exponent $\tau$ of the power law from the measured degree sequence.
We further required that the networks were non-bipartite, and had estimated power-law parameters within or close to the range $2 < \tau < 3$, for all three employed estimators: Hill, Moments and Kernel estimator.
For the experiments comparing assortativity coefficient values, we additionally include networks which are of strong power-law type but whose exponent $\tau$ is larger than $3$.

\begin{table}[ht]
  \centering
  \caption{The real-world networks we use in our experiments.
    We report the number of vertices, the number of edges, the average degree, the Spearman assortativity coefficient and the power-law exponent $\tau$ estimated by the Hill estimator~\cite{Simpl_Gener_Approac_Infer_About_Tail_Distr-Hill75} as reported in \cite{voitalov2019scale}.}
  \label{tab:real-world-graphs}
  \vspace{1ex}
  \begin{tabular}{lrrrrr}
    \toprule
    graph           & vertices      & edges           & avg degree  & $\tau$     & assortativity \\ 
    \midrule
    LiveJournal     & \num{5204176} & \num{49174464}  & \num{18.90}  & \num{3.86} & \num{0.4}     \\ 
    Orkut           & \num{3072441} & \num{117184899} & \num{76.28} & \num{3.58} & \num{0.36}    \\ 
    Gowalla         & \num{196591}  & \num{950327}    & \num{9.67}  & \num{2.80}  & \num{0.25}    \\ 
    Brightkite      & \num{58228}   & \num{214078}    & \num{7.35}  & \num{3.51} & \num{0.22}    \\ 
    Bible names     & \num{1773}    & \num{9131}      & \num{10.30}  & \num{3.09} & \num{0.08}    \\ 
    WordNet         & \num{146005}  & \num{656999}    & \num{9.00}     & \num{2.86} & \num{0.03}    \\ 
    Human PPI       & \num{3133}    & \num{6726}      & \num{4.29}  & \num{3.04} & \num{-0.03}   \\ 
    Youtube         & \num{1134890} & \num{2987624}   & \num{5.27}  & \num{2.48} & \num{-0.08}   \\ 
    Skitter         & \num{1696415} & \num{11095298}  & \num{13.08} & \num{2.38} & \num{-0.14}   \\ 
    Proteins        & \num{1870}    & \num{2277}      & \num{2.44}  & \num{3.09} & \num{-0.16}   \\ 
    Amazon          & \num{334863}  & \num{925872}    & \num{5.53}  & \num{3.99} & \num{-0.19}   \\ 
    Catster/Dogster & \num{623766}  & \num{15699276}  & \num{50.34} & \num{2.10}  & \num{-0.23}   \\ 
    Dogster         & \num{426820}  & \num{8546581}   & \num{40.05} & \num{2.15} & \num{-0.3}    \\ 
    Route views     & \num{6474}    & \num{13895}     & \num{4.29}  & \num{2.13} & \num{-0.5}    \\ 
    CAIDA           & \num{26475}   & \num{53381}     & \num{4.03}  & \num{2.10}  & \num{-0.53}   \\ 
    Hyves           & \num{1402673} & \num{2777419}   & \num{3.96}  & \num{2.98} & \num{-0.55}   \\ 
    Catster         & \num{149700}  & \num{5449275}   & \num{72.80}  & \num{2.09} & \num{-0.56}   \\ 
    \bottomrule
  \end{tabular}
\end{table}

\paragraph{Code And Data.}
The code used for these experiments is freely available on GitHub\footnote{\url{https://github.com/thobl/assortativity}} and Zenodo\footnote{\url{https://zenodo.org/records/16746826}}.
All used and produced data is available on Zenodo\footnote{\url{https://zenodo.org/records/16745980}}.

\subsection{Assortativity Coefficients}
\label{sec:eval-assort-coeff}

An assortativity coefficient measures the correlation between the degrees of the vertices of an edge.
To make this more precise, consider an undirected graph $\mathcal{G} = (\mathcal{V}, \mathcal{E})$.
For each edge $\{u, v\} \in \mathcal{E}$ we consider the pairs $(\deg(u), \deg(v))$ and $(\deg(v), \deg(u))$.
For this set of pairs, we consider three types of correlation measures: Pearson's coefficient, Spearman's coefficient, and Kendall's coefficient.

\paragraph{Definition of Different Correlation Measures.}

The \emph{Pearson} correlation is a normalized measure of covariance and essentially measures how linear the correlation between the two dimensions is; also see the definition in Section~\ref{sec:pearson}.
It yields values in $[-1, 1]$, where the extreme values of $-1$ and $1$ are obtained if all data points lie on a line with negative and positive slope, respectively.

The \emph{Spearman} correlation is defined as the Pearson correlation but on the ranks.
Thus a value of $1$ is attained if sorting the points by the first or second coordinate yields the same order, and $-1$ if it yields the reversed order.

While Spearman correlation also takes into account the magnitude of the deviation in ranks, Kendall's correlation only records the sign of the deviation.
To make this more precise, we first require the concepts of concordant and discordant pairs.
Consider two observations $(X_i, Y_i)$ and $(X_j, Y_j)$ (each representing the vertex degrees of an edge).
They are called \emph{concordant} if the line through the two points has positive slope, i.e., if either $X_i < X_j$ and $Y_i < Y_j$ or $X_i > X_j$ and $Y_i > Y_j$.
They are \emph{discordant} if the line has negative slope.
If they are neither concordant nor discordant, they are \emph{tied}.
The \emph{Kendall} correlation is the normalized difference between the number of concordant and discordant pairs, where there are different variants for the normalization depending on how to deal with ties. 
We simply use $\frac{C - D}{C + D}$ where $C$ and $D$ are the number of concordant and discordant pairs respectively.
In preliminary experiments, we also used other variants of normalization (in particular Kendall $\tau_a$ and $\tau_b$) but the differences are negligible.
We note that when determining $C$ and $D$, we only consider pairs of points that come from different edges.
The reason for this is that $(\deg(u), \deg(v))$ and $(\deg(v), \deg(u))$ would always count as discordant, unless $\deg(u) = \deg(v)$, in which case it is a tie.

\subsubsection{Assortativity Coefficients of Real-World Networks}
\label{sec:assort-coeff-real}

\begin{figure}[t]
  \centering
  \includegraphics{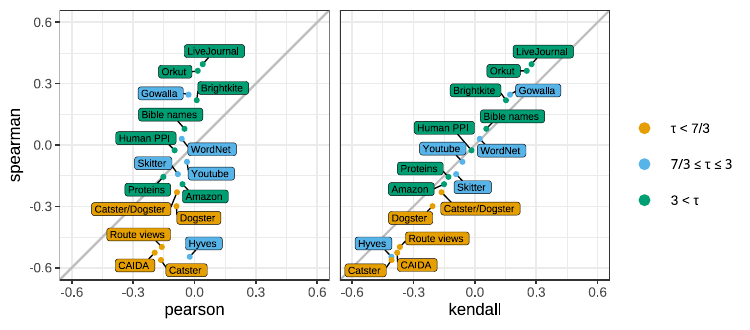}
  \caption{Comparison of assortativity coefficients (Pearson, Kendall, Spearman) for real networks.}
  \label{fig:coeffs_real}
\end{figure}

Fig~\ref{fig:coeffs_real} shows and compares the different assortativity coefficients for the surveyed networks.
Consistent with our theoretical result (Theorem~\ref{thm:neg_pearson}), the sign of the Pearson correlation coefficient is always negative for networks with $\tau<\frac{7}{3}$ and essentially non-positive for all surveyed networks including those with $\tau >3$.
For Spearman and Kendall, we obtain a much wider range of different assortativity values.
For all considered networks, the sign of the assortativity is the same for Spearman and Kendall -- and their magnitudes track each other closely.
This strongly supports the previous suggestion of Litvak and van der Hofstad~\cite{litvak2013uncovering} to use Spearman and not Pearson to measure assortativity.
In the following, unless explicitly mentioned otherwise, the term assortativity coefficient always refers to the Spearman correlation.

Overall we can observe that most networks have negative assortativity.
In particular for networks with power-law exponent $\tau <3$, the three networks with largest assortativity coefficients are Gowalla (\num{0.25}), WordNet (\num{0.03}) and Youtube (\num{-0.08}).

\subsubsection{Our Tunable Models}
\label{sec:our-tunable-models}

\begin{figure}[t]
  \centering
  \includegraphics{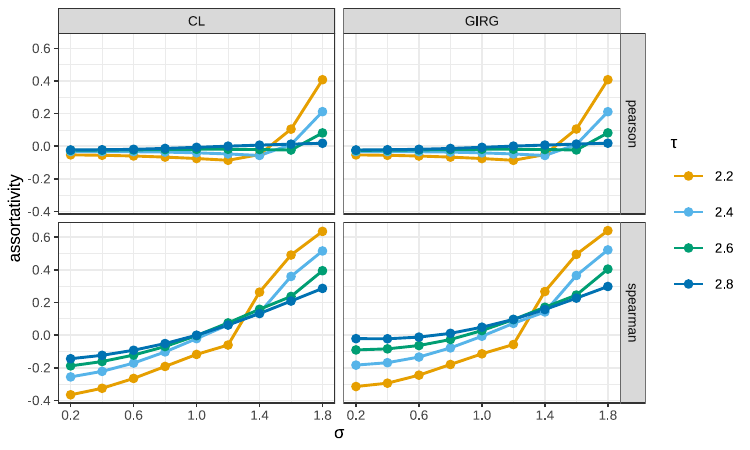}
  \caption{Assortativity coefficients of Tunable Chung-Lu graphs and TGIRGs as a function of $\sigma$ for different power-law exponents $\tau$.
    Each point is the mean of five generated networks.
    Each generated network has \num{200000} vertices and expected average degree \num{15}.
    For the TGIRGs we use dimension \num{2} and temperature \num{0}.}
  \label{fig:assort-scale-sigma}
\end{figure}

The goal with our tunable models was to have a parameter $\sigma$ that controls the assortativity coefficient.
Before analyzing the plots in detail, let us briefly recall the role of $\sigma$ in the connection probability, which we restate here for the case of Tunable Chung-Lu graphs -- the role of $\sigma$ in the connection probability of TGIRGs is analogous: 
\begin{align}
 \Theta\Big(\min\Big\{\frac{(w_u \wedge w_v)^{\sigma}(w_u \vee w_v)}{n}, 1\Big\}\Big).    
\end{align}
Fig~\ref{fig:assort-scale-sigma} shows the assortativity depending on $\sigma$.
Focusing on the Spearman coefficient (bottom row) for now, one can clearly see that changing $\sigma$ has the desired effect that the assortativity changes monotonically depending on $\sigma$.
Note that from our theoretical considerations, we know that we only get a power-law degree distribution for $\tau \ge \sigma + 1$, which corresponds to the change of behavior with sudden upticks for $\sigma \in \{1.2, 1.4, 1.6\}$ for the curves representing $\tau \in \{2.2, 2.4, 2.6\}$, respectively. We note that for $\tau =2.2$, in the range $\sigma \le \tau -1 $, also the value of Spearman, despite increasing as $\sigma$ increases, never becomes positive. Our theoretical results (Proposition~\ref{prop:joint-weight-density-tuned}), which hold for all $2<\tau<3$ indicate that this is another artifact of the use of coefficients and not a sign that there is no assortative wiring behavior occurring in the graph.

Comparing Chung-Lu graphs with GIRGs, we can see that the geometry facilitates a larger assortativity coefficient.
While the effect is not very strong for most parameter combinations, it is particularly pronounced for small values of $\sigma$ and large values of $\tau$. The latter is predicted by our theoretical results: we know that the positive effect on assortativity from the latent geometry is restricted to low-degree vertices (Section~\ref{sec:positive}), on those are more dominant for larger $\tau$.

We note that the assortativity coefficient is rather stable with respect to scaling the graph size, see Fig.\ref{fig:assort-synthetic-sizes}.
We also refer to Section~\ref{sec:eval-joint-distribution} for degree distributions showing that we indeed still get a power-law distribution if $\tau \ge \sigma + 1$.

The top row of Fig~\ref{fig:assort-scale-sigma} is only there to again show the shortcomings of Pearson correlation in the context of assortativity.
The coefficient only slightly changes and actually decreases for increased $\sigma$, except for the regimes where $\sigma$ is too large to actually yield a power-law distribution.

\subsection{Joint Distribution}
\label{sec:eval-joint-distribution}

Here we take a more detailed look at the joint distribution of edge degrees, i.e., when drawing a random edge $\{u, v\}$, $X = \deg(u)$ and $Y = \deg(v)$ are random variables and we are interested in their joint distribution.
We propose three different ways of looking at this joint distribution, resulting in three different types of plots, which we introduce in the following, using a Chung-Lu graph with different values of $\sigma$ as an example; see Fig~\ref{fig:plot_intro_cl}.

\begin{figure}[p]
  \centering
  \begin{subfigure}[t]{\linewidth}
    \centering
    \includegraphics{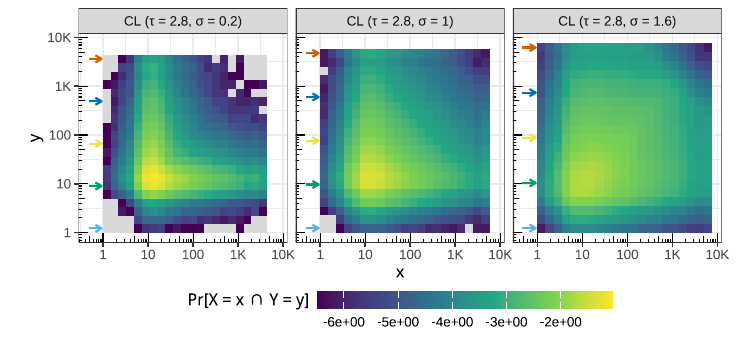}
    \vspace{-1ex}
    \caption{Joint distribution of vertex degrees of a random edge.
      The colors use a logarithmic scale.}
    \label{fig:plot_intro_cl_joint}
    \vspace{2ex}
  \end{subfigure}
  
  \begin{subfigure}[t]{\linewidth}
    \centering
    \includegraphics{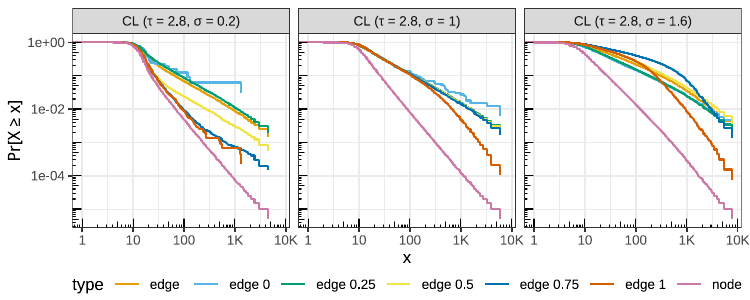}
    \vspace{-1ex}
    \caption{Complementary cumulative distribution of the degree of a random node (\texttt{node}) and of a random endpoint of a random edge (\texttt{edge}).
      The numbers indicate conditioning on the degree of the other endpoint.}
    \label{fig:plot_intro_cl_lines}
    \vspace{2ex}
  \end{subfigure}
  
  \begin{subfigure}[t]{\linewidth}
    \centering
    \includegraphics{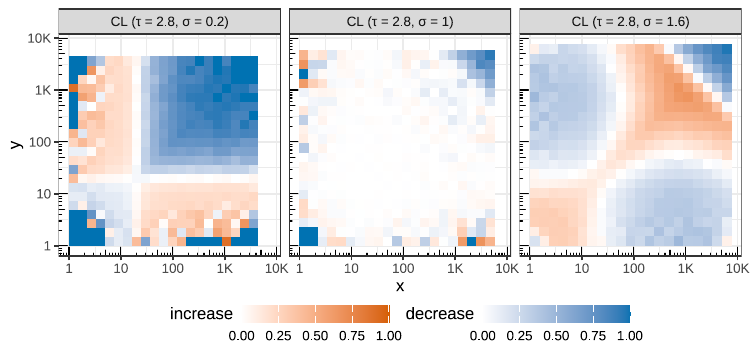}
    \vspace{-1ex}
    \caption{For a random edge with degrees $(X, Y)$, the color indicates how $\Pr[X = x]$ changes when conditioning on $Y = y$.}
    \label{fig:plot_intro_cl_change}
    \vspace{2ex}
  \end{subfigure}
  
  \caption{Three different types of plots illustrating the joint distribution.
    Each column shows a Chung-Lu graph with \num{200000} vertices, average degree \num{15}, power-law exponent $\tau = 2.8$, and varying $\sigma \in \{0.2, 1, 1.6\}$, controlling the assortativity.}
  \label{fig:plot_intro_cl}
\end{figure}

\paragraph{Joint Distribution Heatmap.}

Fig~\ref{fig:plot_intro_cl_joint} directly shows the joint distribution of $X$ and $Y$, i.e., the color at coordinate $(x, y)$ indicates the probability that $X = x$ and $Y = y$. Note that the plot uses logarithmic axes and that the degrees are grouped into buckets. To be more precise, we use \num{21} buckets and for $i \in [0, 20]$, the $i$th bucket $B_i$ is the range $B_i = [b^i, b^{i + 1})$. The base $b$ is chosen such that the maximum degree $d_{max}$ is in the last bucket, i.e., $b = \sqrt[21]{d_{max} + 1}$.

Comparing the three plots in Fig~\ref{fig:plot_intro_cl_joint}, we can clearly see the effect of $\sigma$ on the joint distribution, with $\sigma < 1$ (disassortative) facilitating edges between vertices of different degrees while $\sigma > 1$ leads to more edges with endpoints of similar degree.
However, we note that this trend becomes only apparent due to the comparison of the plots for different values of $\sigma$.
When just looking at the plot for $\sigma = 1$, it is not obvious that the corresponding graph has neutral assortativity.
Indeed, the slightly brighter \scalebox{1.8}{$\llcorner$}-shape might be perceived as negative assortativity, which would be a misinterpretation; also see Fig~\ref{fig:assort-scale-sigma}.
Thus, although this is certainly the simplest representation of the joint distribution, we recommend to use one of the other representations introduced below.

\paragraph{Complementary Cumulative Distributions.}

Fig~\ref{fig:plot_intro_cl_lines} shows the complementary cumulative distribution of the random variable $X = \deg(u)$ for two different random trials.
The line labeled \texttt{node} shows $\Pr[X \ge x]$ if $u$ is a random node.
The line labeled \texttt{edge} shows $\Pr[X \ge x]$ if $u$ is the random endpoint of a random edge $\{u, v\}$.
For the latter, we are also interested in conditioning on the degree $Y = \deg(v)$ of the other endpoint $v$, i.e., $\Pr[X \ge x \mid Y]$.
For this, we consider five different values of $Y$.
More precisely, we condition on $Y$ being in one of the five evenly distributed buckets $B_0$, $B_5$, $B_{10}$, $B_{15}$, or $B_{20}$, as defined above.
In Fig~\ref{fig:plot_intro_cl_lines}, this is labeled as \texttt{edge $\mathtt{c}$} where $\mathtt{c} \in \{0, 0.25, 0.5, 0.75, 1\}$ refers to the bucket $B_{20c}$.
Also note the colored arrows on the left side of the plots in Fig~\ref{fig:plot_intro_cl_joint}.
They indicate on which value of $Y$ the corresponding line in Fig~\ref{fig:plot_intro_cl_lines} conditions.

Before we discuss the line plots, we briefly comment on the scaling of the plots and the interpretation of the curves.
As is common for networks whose degree sequence tail follows a power law, we use a logarithmic scale on both axes.
In this scaling, conforming to a power law then means that the curve is a straight line and its gradient corresponds to $-\tau+1$, where $\tau$ is the exponent of the power law.

Considering the \texttt{node} line, one can see that changing $\sigma$ only slightly changes the curve, i.e., we get the desired power-law distribution consistent with our theoretical results that the (expected) degree of a node is of the same order as its weight and hence does not depend on $\sigma$ (Lemma~\ref{lem:exp_deg_sigma}).
Concerning assortativity considerations, note that complete independence of the two endpoint-degrees of an edge would mean that conditioning on the degree $Y$ should not change the distribution of $X$.
In terms of the curves, this would mean that all curves \texttt{edge $\mathtt{c}$} should coincide with \texttt{edge}.
Note that in the central plot of Fig~\ref{fig:plot_intro_cl_lines} with $\sigma = 1$ (neutral assortativity), this is indeed the case for $\mathtt{c} \in \{0.25, 0.5, 0.75\}$.
Moreover, for $c \in \{0, 1\}$, it is also the case for small values of $x$. For larger values of $x$, however, conditioning on $Y$ being very large ($\mathtt{c} = 1$) makes it less likely that $X \ge x$.
This makes sense as in a Chung-Lu graph, there are not enough high-degree vertices such that many high-degree vertices can connect to other high-degree vertices, see Proposition~\ref{prop:shifted-conditional-weight-distribution-tuned} and also the discussion after the re-statement of Theorem~\ref{thm:neg_pearson} in Section~\ref{subsec:Pearson-proof}. 
Conversely, the probability $X \ge x$ is slightly increased for large $x$ when conditioning on $Y$ being very small ($\mathtt{c} = 0$).

Varying $\sigma$ clearly changes the picture.
For $\sigma = 0.2$ (negative assortativity), one can see in the left plot of Fig~\ref{fig:plot_intro_cl_lines} that the lines corresponding to a large value of $Y$ ($c \in \{0.5, 0.75, 1\}$) lie significantly below the \texttt{edge} line, i.e., conditioning on $Y$ being large decreases the probability that $X \ge x$.
This matches the intuition of what negative assortativity is supposed to mean.
For $\sigma = 1.8$ (positive assortativity) in the right plot, one can see that for small values of $x$, the same lines ($c \in \{0.5, 0.75, 1\}$) lie above the \texttt{edge} line, which again makes sense for positive assortativity.
For larger values of $x$, however, the lines for $c = 1$ and later also for $c = 0.75$ fall below the \texttt{edge} line, which again comes from the fact that there are just not enough high-degree vertices to support enough edges between vertices of very high degree. This agrees with the derivation in Proposition~\ref{prop:shifted-conditional-weight-distribution-tuned}, where we prove that the tail of the curves is composed of two different linear pieces of different slopes. As predicted in Proposition~\ref{prop:shifted-conditional-weight-distribution-tuned}, the second slopes starts very late for small values of $c$ and is thus not visible, but it becomes observable for $c\in\{0.75,1\}$.

While these plots in Fig~\ref{fig:plot_intro_cl_lines} are somewhat difficult to read, we would argue that they are more informative than the heatmaps shown in Fig~\ref{fig:plot_intro_cl_joint}.  

\paragraph{Conditional Heatmaps.}

Finally, the plots in Fig~\ref{fig:plot_intro_cl_change} show how $\pr[X = x]$ changes when conditioning on $Y = y$ (again using the same buckets as before).
We normalize this change as follows.
If $\pr[X = x \mid Y = y] \ge \pr[X = x]$, we say that we have an \emph{increase} of $1 - \frac{\pr[X = x]}{\pr[X = x \mid Y = y]}$.
Otherwise, we have a \emph{decrease} of $1 - \frac{\pr[X = x \mid Y = y]}{\pr[X = x]}$.
Note that this normalizes the values to lie between $0$ and $1$ with values close to $1$ indicating a large change due to conditioning on $Y = y$.
Moreover, $0$ increase is equivalent to $0$ decrease, indicating independence.
We note that the plots are symmetric as $\frac{\pr[X = x]}{\pr[X = x \mid Y = y]} = \frac{\pr[Y = y]}{\pr[Y = y \mid X = x]}$.
For brevity, we call these heatmaps \emph{conditional heatmaps}.

We believe these conditional heatmaps to be easy to read and insightful.
In Fig~\ref{fig:plot_intro_cl_change}, one can nicely see that for neutral assortativity ($\sigma = 1$), the plot is mostly white, except for the blue triangle in the top-right corner (coming from the effect that there are not enough high-degree vertices described above) and some noise for low-degree vertices.
Moreover, the effect of varying $\sigma$ is very apparent.
We discuss these effects in more detail in the following section.

\subsubsection{Discussion -- Generated Networks}
\label{sec:joint-distr-gener}

\begin{figure}[ht!]
  \centering
  \includegraphics{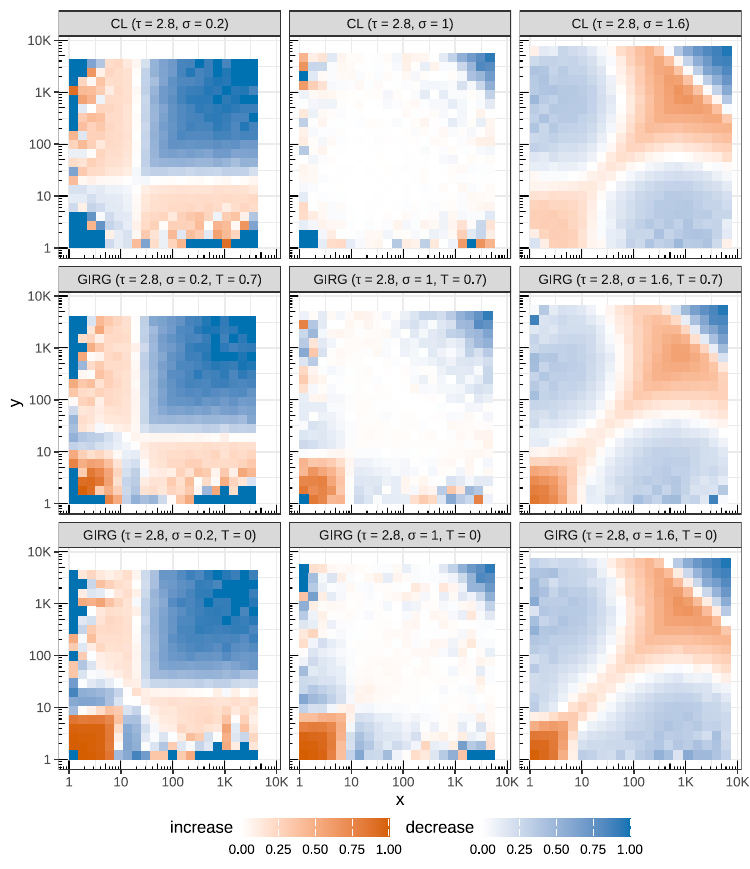}
  \caption{Conditional heatmaps for Tunable Chung-Lu (top), TGIRGs with temperature $T = 0.7$ (middle row), and $T = 0$ (bottom) for $\sigma = 0.2$ (left), $\sigma = 1$ (middle column), and $\sigma = 1.6$ (right).
    Each graph has \num{200000} vertices, average degree \num{15}, and power-law exponent $\tau = 2.8$.
    The dimension for the TGIRGs is \num{2}.}
  \label{fig:conditional-heatmaps-generated}
\end{figure}

The three rows of Fig~\ref{fig:conditional-heatmaps-generated} show the conditional heatmaps for tunable Chung-Lu as well as TGIRGs with temperature $T = 0$ and $T = 0.7$, corresponding to $\alpha=\infty$ and $\alpha\approx 1.4$ respectively.
The columns show different values from $\sigma \in \{0.2, 1, 1.6\}$.
Thus, considering the different columns aids our understanding of $\sigma$'s influence on the assortativity of the underlying graphs.
Moreover, the different rows let us study the impact of the geometry.

For the central column ($\sigma = 1$), we can see that the conditional heatmap is mostly white, indicating neutral assortativity.
The exceptions to this are worth discussing.
First, one can observe the blue triangle in the top-right corner.
This is not surprising, as we have discussed before that there are not enough high-degree vertices to get many edges where both vertices have high degree, as was made precise in Proposition~\ref{prop:shifted-conditional-weight-distribution-tuned}.
Secondly, notice that for Chung-Lu (first row, central column) there is some variation for the low-degree vertices (in particular for degrees $1$ and $2$).
This is most likely just noise: As there are in general only very few vertices of such low degree, a slight absolute change results in a big relative change.
The completely blue square in the bottom-left corner for example tells us that there are no edges where both endpoints have degree $1$ or $2$, which is not surprising if there are in general few vertices of such low degree.
For GIRGs the picture is similar, except that we see an increase in edges where both endpoints have low degree (red bottom-left corner).
This can be explained by random fluctuation in the density of nodes in the underlying geometry, as derived in Section~\ref{sec:positive}. 
One can also nicely see how TGIRGs with higher temperature ($T = 0.7$) lie between Chung-Lu and TGIRG with temperature zero.

For varying $\sigma$, we observe two extremal patterns.
As one extreme, we see an \scalebox{1.8}{$\llcorner$}-shape, with a vertical red stripe on the left hand of the heatmap and horizontal red stripe on the bottom of the heatmap ($\sigma = 0.2$).
This reflects a strongly {disassortative} wiring pattern where {high-degree vertices prefer to connect to low-degree vertices} and vice versa.
It indicates that low values of $\sigma$ successfully produce disassortative networks (while preserving key properties of classical Chung-Lu graphs and GIRGs respectively).
Again, for TGIRGs, the geometry makes connections between low-degree vertices more assortative, resulting in the red lower left corner of the conditional heatmap.

On the other extreme, in highly assortative networks one would expect heatmaps to display a distinguished red diagonal (connecting the lower left corner of the heatmap to its upper right corner).
Indeed, this would indicate that vertices in the corresponding network of degree, say $k$, preferably connect to vertices of degree $k$, whereas connecting to vertices of any other degree is less likely.
And indeed, such assortative behavior can be observed qualitatively for a large range of degrees for both Tunable Chung-Lu graphs and TGIRGs when $\sigma=1.6$; see the right column in Fig~\ref{fig:conditional-heatmaps-generated}.
Notably, this effect now also reaches low-degree nodes in Tunable Chung-Lu graphs.
Note that for very high-degree nodes, this idealized assortative wiring indicated by the red diagonal can not be maintained.
Instead, the conditional heatmaps exhibit a blue triangle in the upper right corner, which is always present in power-law networks with exponent $\tau>2$ -- in this parameter regime there are too few vertices in the entire graph with degrees in this high range.
As a consequence, the probability mass which can not complete the diagonal is spread out laterally to the nodes with highest degrees that are available.
We remark again that the Pearson correlation coefficient puts special emphasis on these very-high-degree vertices which are forced to connect to lower-degree vertices and thus contribute negatively to the overall picture, which is highly problematic for scale-free networks in general and renders it useless for small values of $\tau$. This is made rigorous by our results in Section~\ref{sec:pearson}. 

\subsubsection{Discussion -- Real-World Networks}
\label{sec:joint-distr-real}

\begin{figure}[ht!]
  \centering
  \includegraphics{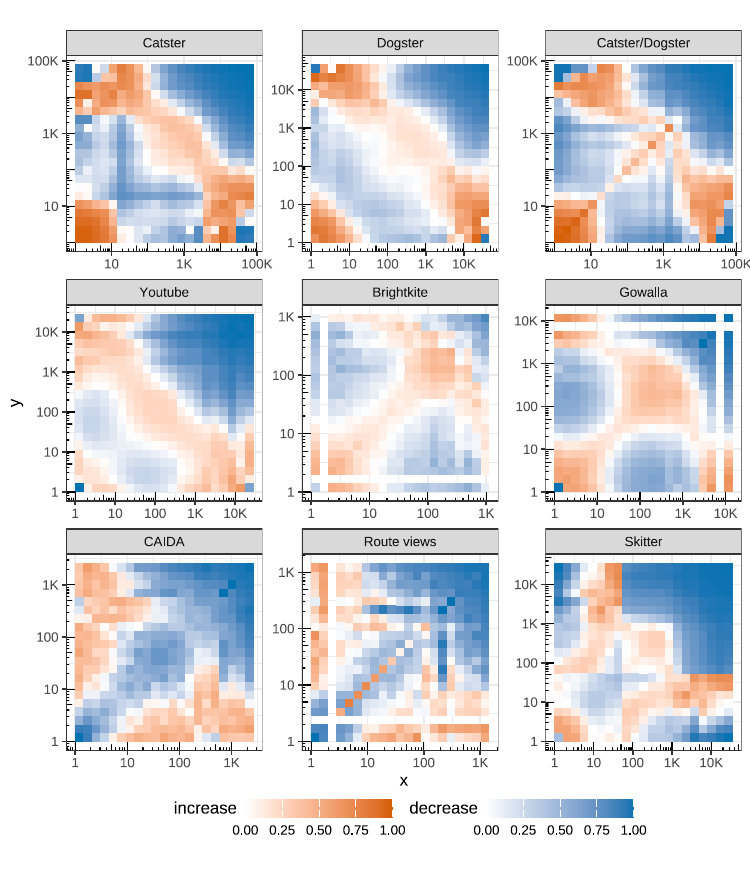}
  \caption{Conditional heatmaps of nine real-world networks.
    The first six are social networks, the bottom three are technical networks representing connections between autonomous systems.
    The networks are ordered by their Spearman assortativity coefficient: \num{-0.56} (\texttt{Catster}), \num{-0.30} (\texttt{Dogster}), \num{-0.23} (\texttt{Catster/Dogster}), \num{-0.08} (\texttt{Youtube}), \num{0.22} (\texttt{Brightkite}), \num{0.25} (\texttt{Gowalla}) for the social networks and \num{-0.53} (\texttt{CAIDA}), \num{-0.50} (\texttt{Route views}), \num{-0.14} (\texttt{Skitter}) for the autonomous systems.}
  \label{fig:conditional-heatmaps-real}
\end{figure}

Changing focus now to the real-world networks, the wiring patterns are much more nuanced.
Nonetheless, some large-scale phenomena can be made out in the conditional heatmaps shown in Fig~\ref{fig:conditional-heatmaps-real}.
Similar to the generated graphs, all plots show a blue triangle in the top-right corner.
As discussed before (Section~\ref{sec:intro-theoretical-results}), this is to be expected for power-law distributions, coming from the fact that there are not enough high-degree vertices to facilitate many edges where both endpoints have high degree. Due to normalization, this also means that the remaining regions in the same rows and columns appear more red. To compensate for this effect we need to compare with an assortative-neutral baseline, for which we can use Chung-Lu networks with matching power-law exponent $\tau$. Fig~\ref{fig:Chung-Lu-baseline} shows such Chung-Lu networks for various values of $\tau$. 
The size of the blue region is larger for smaller $\tau$. For large $\tau$ the resulting effect is negligible since then the blue region is too small to have much impact, but it is notable for small $\tau$. 

\begin{figure}[ht!]
  \centering
  \includegraphics{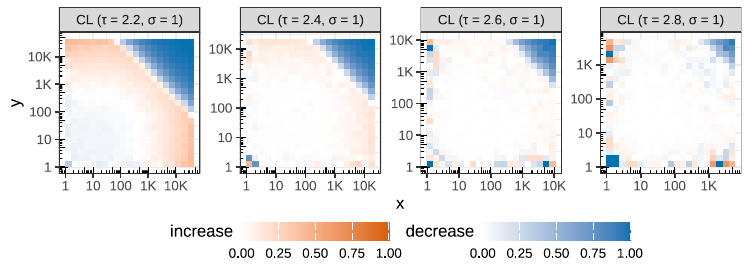}
  \caption{Conditional heatmaps for two Chung-Lu networks with different power law exponents $\tau$. We use those as assortative-neutral baselines.}
  \label{fig:Chung-Lu-baseline}
\end{figure}

\paragraph{Social Networks.}

Considering the social networks in the first two rows of Fig~\ref{fig:conditional-heatmaps-real}, we want to point out that their Spearman assortativity coefficients cover a large range.
In fact, from all networks in our data set with power-law exponent $\tau \in (2, 3]$, the \texttt{Catster} and \texttt{Gowalla} networks have the smallest (\num{-0.56}) and largest (\num{0.25}) assortativity coefficients, respectively; see Fig~\ref{fig:coeffs_real} and Table~\ref{tab:real-world-graphs}.
This contradicts the common belief that social networks tend to have positive assortativity~\cite{newman2002assortative}.
Nonetheless, the conditional heatmaps in Fig~\ref{fig:conditional-heatmaps-real} reveal a pattern that is common among the social networks in our data set: Low degree vertices in the range $[1, 10]$ preferably connect to vertices of similar degree and to vertices with much higher degree but not so much to vertices in the intermediate range. Note that the colors especially in the three petster networks (\texttt{Catster}, \texttt{Dogster}, \texttt{Catster/Dogster}) and in the \texttt{Gowalla} network are substantially more pronounced than in the baseline in Fig~\ref{fig:Chung-Lu-baseline}. Thus the pattern of preferences is not just a normalization artifact.

This pattern can at least partially explain the wide range in assortativity values.
Connections between pairs of low-degree vertices contribute to positive assortativity while edges with one low-degree and one high-degree endpoint yield negative assortativity.
These are conflicting effects and depending on which is stronger, one might obtain positive, negative, or neutral assortativity. 

Beyond these similarities between the networks in the wiring patterns of low-degree vertices, there are also some differences, in particular for the medium degree vertices.
In the \texttt{Catster} network, we can observe that vertices in the range $[10, 100]$ preferably connect to vertices of much higher degree.
In contrast to that, in the \texttt{Gowalla} network, we also get an increased number of connections between medium-degree vertices, resulting in a slightly red diagonal.
We note that this fits to the fact that \texttt{Catster} has negative and \texttt{Gowalla} has positive assortativity.
The \texttt{Youtube} network, shows a similar behavior as discussed before, with low-degree vertices (range $[1, 7]$) connecting to other low-degree vertices and to high-degree vertices (degree $> 400$).
In this network, vertices of intermediate degree in the range $[20, 400]$ predominantly connect to vertices in the same range.
However, we do not see a red diagonal as there is a gap where vertices in the range $[7, 20]$ are slightly less likely to connect with vertices in the same range.
We note that, while the assortativity coefficient of the \texttt{Youtube} network is close to $0$, we observe a highly interesting wiring pattern that is far from being neutral, i.e., there is a dependence between the vertex degrees of an edge that it is too intricate to be well described by a single number.

We want to end this section by discussing two patterns that we would view as artifacts in the data.
First, the \texttt{Youtube} and the \texttt{Gowalla} network both have a bright blue box in the bottom-left corner, coming from the fact that these networks do not contain any edges where both endpoints have degree~$1$.
Note that such edges would form their own connected components.
Thus, this blue box just tells us that the graph contains no connected component consisting of two vertices.
Secondly, the \texttt{Catster/Dogster} network (and also the \texttt{Route views} network we discuss in the next paragraph) show a sharp increase along the diagonal.
This comes from the fact that the data for these networks contain self loops.
Thus, we believe that these sharp diagonals are an artifact of the data set and do not reflect an interesting property of the network itself.

\paragraph{Autonomous Systems.}

The three networks in the bottom row of Fig~\ref{fig:conditional-heatmaps-real} show computer networks where each vertex represents an autonomous systems and edges indicate direct connections.
Interestingly, we see a very different wiring pattern compared to the social networks.
In the first two networks, \texttt{CAIDA} and \texttt{Route views}, we get a very similar picture.
The main difference is the sharp increase on the diagonal for \texttt{Route views}, which is an artifact from self-loops in the network data.
Moreover, \texttt{Route views} is more noisy than \texttt{CAIDA}, which is likely due to the fact that the network is substantially smaller (\num{6.5}{k} and \num{26}{k} vertices, respectively).
We thus focus on \texttt{CAIDA} in the following, although similar observations are true for \texttt{Route views}.

The overall pattern for \texttt{CAIDA} is that the diagonal is strikingly sparse and generally most connections are between vertices of low degree ($<10$) and vertices of higher degree ($> 10$).
Examining the wiring preferences more in detail, the only exception to this pattern are vertices with degree in the range $[200, 600]$ that also have an increased connection probability to vertices in the range $[10, 30]$, i.e., slightly above $10$.
This overall negative assortativity of the \texttt{CAIDA} network, which is also reflected by the Spearman assortativity coefficient of \num{-0.53}, is consistent with the general belief that technological networks tend to be disassortative.

For the \texttt{Skitter} network, the picture is more nuanced.
There is an increased number of edges where both endpoints have a similar degree for the ranges $[1, 7]$ and $[100, 600]$.
Beyond this assortative wiring, we mostly get disassortative patterns: Vertices with degree in the range $[1, 7]$ have an increased number of connections to vertices in the range $[50, 200]$.
Vertices of slightly higher degree in range $[7, 50]$ tend to connect to vertices of even higher degree ($> 100$).
Although less prominent than for \texttt{CAIDA}, this overall yields mostly disassortative wiring, which is also reflected in the assortativity coefficient of \num{-0.14}.

We believe that the differences between \texttt{CAIDA} and \texttt{Skitter} are quite interesting and we can offer some potential explanations (which should be viewed as educated guesses rather than definite truth).
These networks of autonomous systems form strong hierarchical structures with connections between different levels of the hierarchy.
Moreover, vertices higher up in the hierarchy tend to be central hubs with high degree while vertices lower down in the hierarchy have fewer connections.
It thus makes sense that we get disassortative patterns if most connections are between different levels of the hierarchy. 
The core difference between \texttt{CAIDA} and \texttt{Skitter} is their size, with \num{26}{k} vertices for \texttt{CAIDA} and \num{1.7}{M} vertices for \texttt{Skitter}.
In a smaller network like \texttt{CAIDA}, it is conceivable to have a flat hierarchy with vertices of very low degree directly connecting to vertices of very high degree.
However, this likely becomes infeasible in the larger \texttt{Skitter} network.
This is consistent with the rather extreme power-law exponents of $\tau = 2.1$ for \texttt{CAIDA} and $\tau = 2.13$ for \texttt{Route views}, while \texttt{Skitter} has $\tau = 2.38$.
With this, the patterns we observe for the \texttt{Skitter} network between vertices of different degrees (top-left and bottom-right region of the plot) could come from a hierarchy with multiple levels, where connections are predominantly between adjacent levels containing vertices with different (but not extremely different) degrees.
Moreover, the two red regions on the diagonal could come from connections on the same level of the hierarchy or from adjacent levels of the hierarchy containing vertices of similar degree (like, e.g., in a regular tree).

Although \texttt{Skitter} is less disassortative than the other two, we would say that our results are in agreement with the general belief that technological networks tend to be disassortative.
However, the detailed picture is again more nuanced than can easily be described with a single number.

\paragraph{Comparison With the Models.}

We have already seen in the previous paragraphs that real-world networks in Fig~\ref{fig:conditional-heatmaps-real} can exhibit a variety of assortativity patterns, and that it makes sense to distinguish between social networks (first two rows in Fig~\ref{fig:conditional-heatmaps-real}) and the technical networks stemming from autonomous systems (last row in Fig~\ref{fig:conditional-heatmaps-real}). 

For social networks, we universally observe that low-degree vertices show an increased preference to connect to other low-degree vertices, which contributes positively to assortativity. On the other hand, assortativity between other pairs of vertices may vary widely. In Section~\ref{subsec:degreeassortativity}, we showed for our generative models that the effect of the latent geometric space on assortativity is to increase assortativity between pairs of vertices of low degree, while not affecting assortativity between other pairs of vertices. The fact that all studied social networks show such an increased assortativity between low-degree vertices suggests that the changes induced by a latent space are compatible with the observed assortativity in social networks. 

In general the real-world networks show a more complex structure than the network models. When we compare individual social networks with generated networks of the same parameters, the fit is not tight. Fig~\ref{fig:comparison-social-networks} shows a comparison for several social networks. The \texttt{Gowalla} network shows a red {$\rotatebox[origin=c]{-45}{$\top$}$-shape} in its conditional heatmap in Fig~\ref{fig:conditional-heatmaps-real} which has some resemblance to the {$\rotatebox[origin=c]{-45}{$\top$}$-shape} of the TGIRG with the same power-law exponent $\tau = 2.8$ and with $\sigma = 1.6$ (left bottom panel). However, apart from the more nuanced structure in the real network we want to highlight one important difference: for the generated network the blue triangle is substantially smaller and the upper diagonal of the red {$\rotatebox[origin=c]{-45}{$\top$}$-shape} does not extend to the corners, while it does so for the \texttt{Gowalla} network. This means that vertices of low degree (which are most vertices) have an increased probability to form an edge with a vertex of very high degree in the \texttt{Gowalla} network, but not in the corresponding TGIRG. This effect might partially stem from an imperfect power-law in the \texttt{Gowalla} network, see Fig~\ref{fig:gowalla} in the Supporting Information. But even for generated models with a smaller value of $\tau$, the probability is not as much increased as in the \texttt{Gowalla} network, cf. Fig~\ref{fig:Chung-Lu-baseline}. This indicates that the wiring pattern of vertices of very high degree is not adequately  captured by the models. There is only a small number of such vertices, but due to their high degrees they contribute a substantial fraction of all edges. 

The situation for the \texttt{Youtube} network is similar to \texttt{Gowalla}. Fitting the parameters gives a loose fit (middle row, bottom panel), but it does not capture the finer structures and is not a good quantitative fit in all regions. Note that in the plot we chose $\tau=2.2$ because that gives a better fit better the size of the blue triangle in the upper right, but it does not correspond to the exponent $\tau=2.48$ of the \texttt{Youtube} network.

Among social networks, we get the worst fit for the petster networks. Since those have negative assortativity coefficient, this suggest using a tuning parameter $\sigma <1$, but this does not result in a similar conditional heatmap, see the right column in Fig~\ref{fig:comparison-social-networks}. Arguably, the gestalt shape for neutral $\sigma =1$ (Fig~\ref{fig:comparison-social-networks}, middle row bottom panel) or even $\sigma = 1.2$ (Fig~\ref{fig:models-matching-real-world} left) fits better, but even then the quantitative match is not good. In particular, it underestimates the number of connections between vertices of low and of very high degree. This suggests that the assortative structure of the petster networks is richer and can not be captured by the uniform application of a single tuning parameter as in the TGIRG model.

\begin{figure}[t!]
  \centering

  \includegraphics{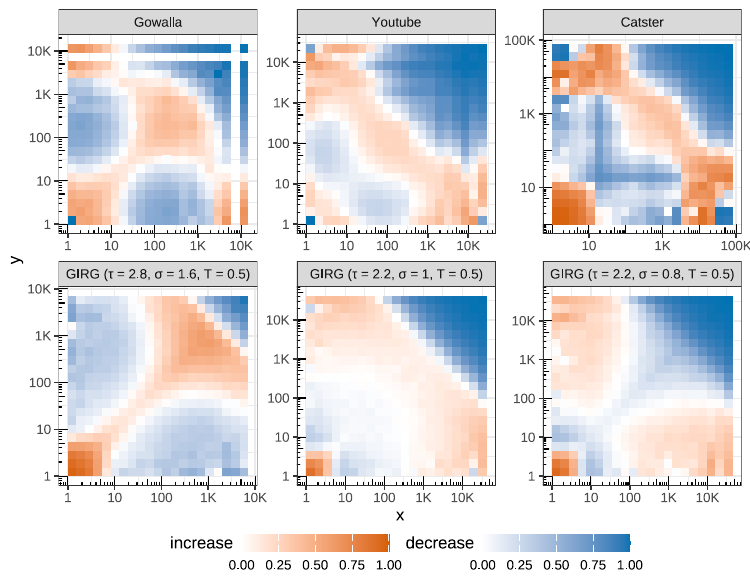}
  \caption{Direct comparison of social networks with generated networks of corresponding parameters.}
  \label{fig:comparison-social-networks}
\end{figure}

\begin{figure}[t!]
  \centering
  \includegraphics{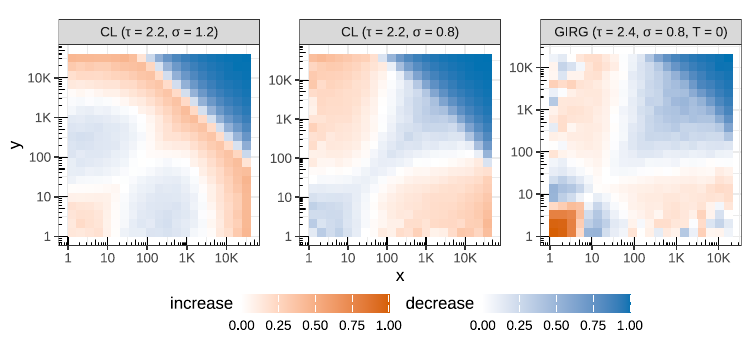}
  \caption{Conditional heatmaps for three models that, to some degree, show similar patterns as observed in the real-world networks.}
  \label{fig:models-matching-real-world}
\end{figure}

For the autonomous systems networks (bottom row in Fig~\ref{fig:conditional-heatmaps-real}), we get somewhat similar wiring patterns when using $\sigma = 0.8$ (Fig~\ref{fig:models-matching-real-world} middle and right).
Due to the fact that an underlying geometry increases the probability for edges between low-degree vertices, the Tunable Chung-Lu seems to be a better fit for \texttt{CAIDA} and \texttt{Route views}, while TGIRG seems to be a better fit for \texttt{Skitter}. A notable difference to the artificial networks is that \texttt{Skitter} has an increased probability for edges where both vertices are in the range $[100,600]$, which is not present in the models.

\begin{figure}[ht!]
    \centering
    \includegraphics[width=0.5\linewidth]{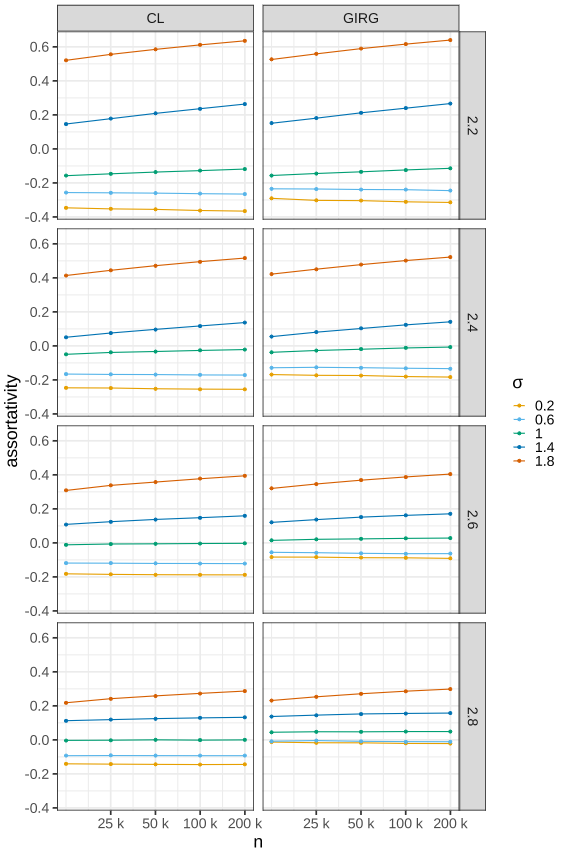}
    \caption{Spearman Assortativity coefficient values in Tunable Chung-Lu Graphs and TGIRGs for varying network sizes and levels of $\sigma$.}
    \label{fig:assort-synthetic-sizes}
\end{figure}


\section{Scale-Free Networks And Assortativity Coefficients}\label{sec:pearson}

This section is dedicated to the analysis of scale-free networks using assortativity coefficients. In the first subsection, we show that for a substantial range of $\tau$, all networks whose degree distribution follow a power-law with parameter $\tau$ have negative Pearson correlation coefficient. It has been argued previously that the Pearson coefficient is ill-suited to model degree-degree correlations in scale-free networks. Here, we provide a much stronger result which supports this claim. Intuitively, our main theorem of this section means that the Pearson correlation coefficient can never give meaningful results in any network with with power-law degree distribution and exponent close to 2, independent of any homophily or heterophily in the network. In the second subsection, we provide a discussion of experimental results as well as alternative coefficients - and why the analysis of conditional and joint distributions is better suited to understand assortativity than coefficients.

\subsection{Severe Shortcoming of Pearson Correlation Coefficient for the Analysis of Scale-Free Networks}\label{subsec:Pearson-proof}

We next provide the formal definition of the Pearson correlation coefficient, as given by Newman in~\cite{newman2002assortative}.

Consider a network $\mathcal{G}=(\mathcal{V}, \mathcal{E})$ and define a pair $(X_0,X_1)$ of random variables as follows.  We sample an edge $e=uv\in\mathcal{E}$ uniformly at random, and then set $(X_0,X_1) = (\deg(u)-1, \deg(v)-1)$ with probability $\tfrac{1}{2}$ and $(X_0,X_1) = (\deg(v)-1, \deg(u)-1)$ with probability $\tfrac{1}{2}$. These $\deg(\cdot)-1$ terms correspond to the \emph{remaining degree} of the vertex, i.e.\ the number of additional edges that are incident to the vertex, discounting the edge $e$ that led there in the first place. Note that $X_0$ and $X_1$ are identically distributed but not necessarily independent. One way to quantify the dependence between the endpoint degrees of a random edge is to compute the Pearson correlation coefficient of $X_0$ and $X_1$. The \emph{Pearson assortativity coefficient} of the graph $\mathcal{G}$ is defined as such:
\begin{align}\label{eq:pearson-assortativity-def}
    r(\mathcal{G}):= \frac{\cov(X_0,X_1)}{\sqrt{\var(X_0)}\sqrt{\var(X_1)}} = \frac{\E[X_0X_1]-\E[X_0]\E[X_1]}{\var(X_0)},
\end{align}
where the equality above holds because $X_0$ and $X_1$ have the same distribution. Note in particular that the denominator is always positive by definition. For convenience, we restate the main theorem of this section.

\Negativepearson*

Intuitively, the negativity of $r(\mathcal{G})$ comes from the contribution of high-degree vertices. More precisely, because of the power-law degree distribution, any vertex of degree $\omega(n^{1/\tau})$ has to connect to vertices of smaller degree, because there are $o(n^{1/\tau})$ vertices of degree $\omega(n^{1/\tau})$. This effect is even stronger for vertices of maximum degree $\Theta(n^{1/(\tau-1)})$, because almost all of their neighbors have much smaller degree. If the tail of the degree distribution is heavy enough (which translates to the $\tau<7/3$ condition), the contribution of these vertices dominates the numerator in \eqref{eq:pearson-assortativity-def}.

Before proving Theorem \ref{thm:neg_pearson}, we need to introduce the following lemma, which will be used to perform summation by parts.

\begin{lemma}[Abel's summation formula \cite{abel1826untersuchung, tenenbaum2015introduction}]\label{lem:abel}
    Let $(s_k)_{k\in\N}, (t_k)_{k\in\N}$ be two sequences of real numbers, and let $a<b$ be integers. Then we have
    \begin{align*}
        \sum_{k=a}^b (s_{k+1}-s_k)t_k = (s_{b+1}t_{b+1}-s_at_a) - \sum_{k=a}^b s_{k+1}(t_{k+1}-t_k).
    \end{align*}
\end{lemma}

\begin{proof}[Proof of Theorem \ref{thm:neg_pearson}]
We start by noting that $|\mathcal{E}|= \Theta(n)$ since, taking $v$ as a uniformly random vertex, the average degree can be computed as 
\begin{align*}
    \sum_{k=1}^{\Delta} \pr(\deg(v) \ge k) = \Theta\left(\sum_{k=1}^{\Delta} k^{-(\tau-1)}\right) = \Theta(1)
\end{align*}
since $\tau-1 > 1$.

The probability that a vertex chosen uniformly at random has degree $\ge k$ is of order $k^{1-\tau}$, however when we pick a uniformly random edge (and then the endpoint of such an edge), the probability to find a vertex of degree $l$ is multiplied by $l$. In other words, (recalling that the $X_i$ are defined as the \emph{remaining} degrees of $u,v$),
\begin{align}\label{eq:shifted-distribution}
    \pr(X_i \ge k-1)= \Theta(k^{2-\tau}).
\end{align}

By the assumption on the degree distribution in $\mathcal{G}$, there are $O(n^{2} \cdot k^{1-\tau} \cdot l^{1-\tau})$ pairs of vertices such that one of them has degree $\ge k$ and the other has degree $\ge l$. In particular, there are at most  $O(n^{2} \cdot k^{1-\tau} \cdot l^{1-\tau})$ edges connecting a vertex of degree $\ge k$ to a vertex of degree $\ge l$. 
Since $|\mathcal{E}|=\Theta(n)$, we deduce that 
\begin{align}\label{eq:joint_estimate}
    \pr(X_0 \ge k-1, X_1 \ge l-1) \le O(n \cdot k^{1-\tau} \cdot l^{1-\tau}).
\end{align}

For the sake of readability, we introduce the following notation:
\begin{align*}
    e_{kl} := \pr(X_0=k,X_1=l), \qquad q_k := \pr(X_0=k) = \pr(X_1=k).
\end{align*}
We start by computing the numerator in \eqref{eq:pearson-assortativity-def}, which determines the sign of $r(\mathcal{G})$ since the denominator is always positive. We expand the covariance expression from \eqref{eq:pearson-assortativity-def}:
\begin{align}
\begin{split}\label{eq:covariance-pearson}  
    \cov(X_0, X_1) = \E[X_0X_1]-\E[X_0]\E[X_1] = \sum_{k,l=0}^{\Delta}kle_{kl} - \sum_{k,l=0}^{\Delta}klq_kq_l.
\end{split}
\end{align}
We will show that, when $2<\tau<7/3$, the first sum is of order $O(n)$, while the second sum is of order $\omega(n)$, which yields negativity of $r(\mathcal{G})$.

Let $i\in\{0,1\}$ and $A\subset \N$ be a (non-empty) subset of the integers. Using Abel's summation formula (Lemma \ref{lem:abel}, with the sequence $(s_k)_{k\in \N}$ defined as either $s_k:=\pr(X_i \ge k)$ or $s_k:=\pr(X_i \ge k, X_{1-i}\in A)$ and the sequence $(t_k)_{k\in \N}$ defined as $t_k:=k$), we get the following identities, which will be used throughout the proof:
\begin{align}
    \sum_{k=0}^\Delta k\pr(X_i=k) &= \sum_{k=0}^\Delta \pr(X_i>k), \label{eq:abel-simple} \\
    \sum_{k=0}^\Delta k\pr(X_i=k, X_{1-i} \in A) &= \sum_{k=0}^\Delta \pr(X_i>k, X_{1-i}\in A), \label{eq:abel-double}
\end{align}
where in particular we used the fact that $X_i\ge \Delta+1$ has probability $0$ since $\Delta$ is the maximum degree in $\mathcal{G}$.

We start by bounding the first term in \eqref{eq:covariance-pearson}. By definition of $e_{kl}$, we have
\begin{align*}
    \sum_{k=0}^{\Delta} \sum_{l=0}^{\Delta} kl e_{kl}
    = \sum_{k=0}^{\Delta} \sum_{l=0}^{\Delta} kl \pr(X_0=k, X_1=l).
\end{align*}
Using identity \eqref{eq:abel-double} with $i=1$ and $A=\{k\}$ we get
\begin{align*}
    \sum_{k=0}^{\Delta} \sum_{l=0}^{\Delta} kl e_{kl}
    = \sum_{k=0}^{\Delta}k \sum_{l=0}^{\Delta} \pr(X_0=k, X_1>l).
\end{align*}
Using the same identity but now with with $i=0$ and $A=\{l+1, l+2, \ldots\}$ yields
\begin{align*}
    \sum_{k=0}^{\Delta} \sum_{l=0}^{\Delta} kl e_{kl}
    &= \sum_{l=0}^{\Delta} \sum_{k=0}^{\Delta} \pr(X_0>k, X_1>l).
\end{align*}
Now, using \eqref{eq:joint_estimate}, we finally get
\begin{align*}
    \sum_{k=0}^{\Delta} \sum_{l=0}^{\Delta} kl e_{kl}
    \le \sum_{l=1}^{\Delta+1}\sum_{k=1}^{\Delta+1} O(n \cdot k^{1-\tau} \cdot l^{1-\tau}) = O(n),
\end{align*}
where we used $\tau>2$.

It remains to bound the first term in \eqref{eq:covariance-pearson}, i.e.\ to show that
\begin{align*}
    \sum_{k,l=0}^{\Delta} kl q_kq_l = \omega(n).
\end{align*}
Note that 
\begin{align*}
     \sum_{k,l=0}^{\Delta} kl q_kq_l = \Big(\sum_{k=0}^{\Delta} k q_k\Big)^2 = \Big(\sum_{k=0}^{\Delta} k \pr(X_0=k)\Big)^2,
\end{align*}
where we used the definition of $q_k$ for the second equality. Using identity \eqref{eq:abel-simple} with $i=0$ yields
\begin{align*}
     \sum_{k,l=0}^{\Delta} kl q_kq_l = \Big(\sum_{k=0}^{\Delta} \pr(X_0>k)\Big)^2.
\end{align*}
Recall that $\pr(X_0>k)=\Theta(k^{2-\tau})$ for all $k\le\Delta$ by Eq. \eqref{eq:shifted-distribution}. Hence, by approximating the sum with the corresponding integral we get
\begin{align*}
     \sum_{k,l=0}^{\Delta} kl q_kq_l = \Big(\sum_{k=0}^{\Delta} \Theta(k^{2-\tau})\Big)^2 = \Theta((\Delta^{3-\tau})^2).
\end{align*}
We now use that $\Delta = \Theta(n^{1/(\tau-1)})$ to deduce
\begin{align*}
     \sum_{k,l=0}^{\Delta} kl q_kq_l = \Theta\big(n^{\tfrac{6-2\tau}{\tau-1}}\big) = \Theta\big(n^{\tfrac{4}{\tau-1}-2}\big).
\end{align*}
The exponent $\tfrac{4}{\tau-1}-2$ satisfies $\tfrac{4}{\tau-1}-2>1$ if and only if $\tau<7/3$. In this case, we have shown that $\cov(X_0, X_1) = \Theta(-\Delta^{6-2\tau}) = \Theta(-n^{\tfrac{4}{\tau-1}-2})$.

We now turn to the computation of the denominator $\var(X_0)$ in \eqref{eq:pearson-assortativity-def}. Note that $\var(X_0) = \E[X_0^2]-\E[X_0]^2$, and since $X_0$ and $X_1$ follow the same distribution $\E[X_0]^2 = \E[X_0]\E[X_1] = \Theta(\Delta^{6-2\tau})$ as we have just computed. Therefore we just need to compute $\E[X_0^2] = \sum_{k=0}^{\Delta}k^2\pr(X_0=k)$. Using Abel's summation formula (Lemma \ref{lem:abel}) with $s_k := \pr(X_0 \ge k)$ and $t_k:=k^2$, we get
\begin{align*}
    \E[X_0^2] &= \sum_{k=0}^{\Delta}k^2\pr(X_0=k)
    = \sum_{k=0}^\Delta ((k+1)^2-k^2)\pr(X_0>k) \\
    &= \sum_{k=0}^\Delta (2k+1)\Theta(k^{2-\tau})
    = \sum_{k=0}^\Delta \Theta(k^{3-\tau}) 
    = \Theta(\Delta^{4-\tau}).
\end{align*}
Since $\tau>2$, we have $4-\tau > 6-2\tau$, and hence 
\begin{align*}
    \var(X_0) = \E[X_0^2]-\E[X_0]^2 = \Theta(\Delta^{4-\tau})-\Theta(\Delta^{6-2\tau}) = \Theta(\Delta^{4-\tau}).    
\end{align*}
Therefore, using that $\Delta=\Theta(n^{1/(\tau-1)})$, we conclude that
\begin{align*}
    r(\mathcal{G}) = \Theta\left(\frac{-\Delta^{6-2\tau}}{\Delta^{4-\tau}}\right) = \Theta(-\Delta^{2-\tau}) = \Theta\big(-n^{-\tfrac{\tau-2}{\tau-1}}\big).
\end{align*}

\end{proof}

\section*{Supporting Information}

\subsection*{Proof of Lemmas \ref{lem:tgirg-marginal}-\ref{lem:exp_deg_sigma}}

\begin{proof}[Proof of Lemma \ref{lem:tgirg-marginal}]
Assume first that $\alpha <\infty$. Note that
\begin{align*}
    \Pr(uv\in\mathcal{E} \mid x_u, w_u, w_v) = \int_0^1 r^{d-1} \Theta\Big(\min\Big\{\frac{(w_u \wedge w_v)^{\sigma}(w_u \vee w_v)}{n r^d}, 1\Big\}^{\alpha}\Big) dr,
\end{align*}
and therefore it suffices to compute the above integral.
If $(w_u \wedge w_v)^{\sigma}(w_u \vee w_v) \ge n$, the minimum in the probability is attained by $1$ for all $r\in[0,1]$ and the integral becomes
\begin{align*}
    \int_0^1 r^{d-1} \Theta(1) dr = \Theta(1).
\end{align*}
Otherwise, we split the integral in two terms depending on where which expression attains the minimum, which yields
\begin{align*}
    \int_0^{(\frac{(w_u \wedge w_v)^{\sigma}(w_u \vee w_v) }{n})^{\frac{1}{d}}} r^{d-1} \cdot  \Theta(1) dr + \int_{(\frac{(w_u \wedge w_v)^{\sigma}(w_u \vee w_v) }{n})^{\frac{1}{d}}}^1 r^{d-1}\cdot \Big(\frac{(w_u \wedge w_v)^{\sigma}(w_u \vee w_v) }{nr^d} \Big)^{\alpha} dr \\
    =\Theta\Big(\frac{(w_u \wedge w_v)^{\sigma}(w_u \vee w_v) }{n}\Big) + \Theta\Big(\frac{(w_u \wedge w_v)^{\sigma}(w_u \vee w_v) }{n}\Big)^\alpha \cdot \Big[-r^{d(1-\alpha)}\Big]^1_{(\frac{(w_u \wedge w_v)^{\sigma}(w_u \vee w_v) }{n})^{\frac{1}{d}}},
\end{align*}
where the minus sign for the second integral comes because $\alpha >1$ and hence $d(1-\alpha) < 0$. Hence this second integral is also of order $\Theta(\frac{(w_u \wedge w_v)^{\sigma}(w_u \vee w_v) }{n})$, which concludes the proof for $\alpha <\infty$. If $\alpha=\infty$, we may proceed analogously, with the only difference that the only nonzero contribution to the integral stems from the term
\begin{align*}
    \int_0^{(\frac{(w_u \wedge w_v)^{\sigma}(w_u \vee w_v) }{n})^{\frac{1}{d}}} r^{d-1} \cdot  \Theta(1) dr = \Theta\Big(\frac{(w_u \wedge w_v)^{\sigma}(w_u \vee w_v) }{n}\Big).
\end{align*}
Taken together, this concludes the proof.
\end{proof}

\begin{proof}[Proof of Lemma \ref{lem:exp_deg_sigma}]
By Lemma \ref{lem:tgirg-marginal}, we have
\begin{align*}
    \mathbb{E}[\deg(v) \mid w_v] &= \sum_{u \in \mathcal{V} \setminus \{v\}} \Theta\Big(1 \wedge \frac{(w_u\wedge w_v)^\sigma (w_u \vee w_v)}{n}\Big) \\
    &=\Theta\Big(\int_1^\infty ((w_u\wedge w_v)^\sigma (w_u \vee w_v) \wedge n) w_u^{-\tau} dw_u\big) \\
    &=\Theta\Big(\int_1^{w_v} (w_v \cdot w_u^\sigma \wedge n) w_u^{-\tau} dw_u + \int_{w_v}^{\infty} (w_v^\sigma \cdot w_u \wedge n) w_u^{-\tau} dw_u\Big)
\end{align*}
We now make a case distinction depending on whether or not $w_v \ge n^{\frac{1}{\sigma +1}}$. Assume first that this holds. Then
\begin{align*}
    \mathbb{E}[\deg(v) \mid w_v] 
    &=\Theta\Big(w_v \int_1^{w_v} w_u^{\sigma-\tau} dw_u + w_v^\sigma \int_{w_v}^\frac{n}{w_v^\sigma} w_u^{1-\tau} dw_u+ n \int_{\frac{n}{w_v^\sigma}}^{\infty} w_u^{-\tau} dw_u\Big) \\ 
    &= \Theta\Big(w_v + w_v^{\sigma+2-\tau} + n^{2-\tau} w_v^{\sigma(\tau-1)}\Big),
\end{align*}
where for the first integral we used that $\sigma < \tau-1$ while for the second and third integral we used that $\tau>2$. Now observe that the last summand in the expression above can be upper-bounded by
\begin{align*}
     n^{2-\tau} w_v^{\sigma(\tau-1)} \le w_v^{(2-\tau)(\sigma +1)+\sigma(\tau-1)}=w_v^{2+\sigma -\tau}.
\end{align*}
This concludes the proof when $w_v \ge n^{\frac{1}{\sigma +1}}$ since $\sigma < \tau-1$ and hence $2+\sigma -\tau<1$. In the case $w_v < n^{\frac{1}{\sigma +1}}$, the expected degree becomes 
\begin{align*}
    \mathbb{E}[\deg(v) \mid w_v] =\Theta\Big(w_v \int_1^{(\frac{n}{w_v})^{\frac{1}{\sigma}}} w_u^{\sigma-\tau} dw_u + n \int_{(\frac{n}{w_v})^{\frac{1}{\sigma}}}^{\infty} w_u^{-\tau} dw_u\Big) = \Theta\Big(w_v +  n^{1 +\frac{1-\tau}{\sigma}} w_v^{\frac{\tau-1}{\sigma}}\Big),
\end{align*}
where again we used that $\sigma < \tau-1$ for the first integral and $\tau>2$ for the second integral.
Note that $1 +\frac{1-\tau}{\sigma}<0$, so we can upper bound the second summand as
\begin{align*}
n^{1 +\frac{1-\tau}{\sigma}} w_v^{\frac{\tau-1}{\sigma}} \le w_v^{(\sigma+1)(1+\frac{1-\tau}{\sigma})+\frac{\tau-1}{\sigma}}= w_v^{\sigma -\tau + 2 }.
\end{align*}
Since $2+\sigma -\tau<1$, this concludes the proof in this second case.
\end{proof}

\subsection*{Descriptions of Real-World Networks}\label{sec:network-descriptions}
In this subsection, we give some background information on the real-world networks analyzed. The information is gathered from the KONECT database itself~\cite{kunegis2013konect}, from~\cite{voitalov2019scale} and the Stanford Network Analysis Project (SNAP)~\cite{snapnets}. The node counts and average degree values are taken from the KONECT databaset. The estimates for the power-law exponents come from \cite{voitalov2019scale}.

\paragraph{CAIDA (as-caida20071105)}

This dataset models the network of autonomous systems of the internet, collected through the CAIDA project in 2007. Nodes represent autonomous systems, edges denote communication. 
The network consists of \num{26475} nodes with average degree \num{4.03}. The estimated range for the power-law exponent is $[2.10,2.11]$.

\paragraph{Skitter (as-skitter)} This is an internet topology graph, i.e. modeling connections of autonomous systems collected with the internet probing tool skitter, which used traceroute data from the IPv4 address space. It consists of \num{1696415} nodes, with maximum degree \num{35455} and average degree \num{13.08}. The estimated range for the power-law exponent is $[2.36,2.43]$.

\paragraph{Pet relationship datasets}  There are a variety of datasets derived from the - now defunct - social networking websites dogster.com and catster.com,  which were platforms for pet owners to create profiles for their pets, interact, and form connections. These datasets capture relationships such as friendships and family links between users representing pets on these sites.

\begin{itemize}
    \item Catster (petster-friendships-cat). 
    The network consists of \num{149700} nodes with average degree \num{72.80}. The estimated range for the power-law exponent is $[1.98,2.09]$.
    \item Catster/Dogster (petster-carnivore). The network consists of \num{623766} nodes with average degree \num{50.34}. The estimated range for the power-law exponent is is $[2.04,2.11]$.
    \item Dogster (petster-friendships-dog). The network consists of \num{426 816} nodes with average degree \num{40.05}. The estimated range for the power-law exponent is $[2.12,2.15]$.
\end{itemize}

\paragraph{Gowalla (loc-gowalla-edges)}
This undirected network contains user–user friendship relations from Gowalla, a location-based social network where users share their locations. Gowalla was founded in 2007 and gained popularity for its gamified check-in system. 
After acquisition by Facebook in 2011, the platform ceased operation on March 11, 2012. It was relaunched on March 10, 2023 with updated features focusing on social mapping and real-world interactions.

It consists of \num{196591} nodes with average degree \num{9.67}. The estimated range for the power-law exponent is $[2.80,2.86]$.
\paragraph{Hyves (hyves)} This dataset is taken from the Dutch online social network Hyves. Hyves was launched in 2004, primarily catering to users in the Netherlands. It became the country's largest social network during its peak, with over 10 million accounts by 2010. The nodes encode users, the edges encode friendships. It consists of \num{1402673} nodes and average degree \num{3.96}. The estimated range for the power-law exponent is  $[1.99, 2.98]$.

\paragraph{Autonomous systems (as20000102, Route Views)}
The as20000102 dataset represents the graph of Autonomous Systems (AS) in the Internet, as observed on January 2, 2000. Autonomous Systems are sub-networks of routers that exchange traffic flows via the Border Gateway Protocol (BGP). This dataset is part of a broader collection of AS graphs derived from the University of Oregon Route Views Project. It contains \num{6474} nodes with average degree \num{4.29}. The estimated range for the power-law exponent is $[2.13,2.16]$.

\paragraph{Wordnet (wordnet-words)}
This is the lexical network of words from the WordNet dataset. Nodes in the network are English words, and edges are relationships between them, such as synonymy, antonymy, meronymy, etc.  The network consists of \num{146005} nodes with average degree \num{9.00}. The estimated range for the power-law exponent is is $[2.61,2.86]$.

\paragraph{Youtube (com-youtube)} This is the social network of YouTube users and their friendship connections, consisting of 1134890 nodes and average degree 5.3. The estimated range for the power-law exponent is is $[2.17,2.58]$.

\paragraph{Orkut (livejournal-links)} Livejournal is a free membership-based online community 
which allows members to maintain journals, individual and group blogs, and it allows people to declare which other members are their friends. The dataset was assembled through a crawl of the website. It consists of \num{5204176 } nodes with average degree \num{18.90}. The estimated range for its power-law exponent is $[3.15, 4.04]$.

\paragraph{Orkut (orkut-links)} Orkut is a free online social network where users form friendships with each other. Orkut also allows users form a group which other members can then join. The dataset was gathered through a crawl of the website. It consists of \num{3072441} nodes with average degree \num{76.28}. The estimated range for its power-law exponent is $[2.65, 3.58]$.

\paragraph{Brightkite (loc-brightkite\_edges)} Brightkite was a location-based social networking service provider where users shared their locations by checking in. The friendship network was collected using their public API, and consists of \num{58228} nodes with average degree \num{7.35}.  The estimated range for its power-law exponent is $[2.96, 3.8]$.

\paragraph{Human proteins (maayan-vidal, Human PPI)} This network represens an initial version of a proteome-scale map of Human binary protein–protein interactions. It consists of \num{3133} nodes  and average degree \num{4.29}. The estimated range for its power-law exponent is $[3.09, 3.87]$.

\paragraph{Yeast (moreno\_propro, Proteins)} This undirected network contains protein interactions contained in yeast. 
A node represents a protein and an edge represents a metabolic interaction between two proteins.  The network consists of \num{1870} nodes with average degree \num{2.44}. The estimated range for its power-law exponent is $[3.09, 3.87]$.

\paragraph{Bible (moreno\_names, Bible names)} This undirected network contains nouns (places and names) of the King James Version of the Bible and information about their co-occurrences. A node represents one of the above noun types and an edge indicates that two nouns appeared together in the same Bible verse. It contains \num{1773} nodes with average degree \num{10.3}. The estimated range for its power-law exponent is $[2.88,3.09]$.

\paragraph{Amazon (com-amazon)} This is the co-purchase network of Amazon based on the "customers who bought this also bought" feature. Nodes are products and an undirected edge between two nodes shows that the corresponding products have been frequently bought together. It consists of \num{334863} nodes with average degree \num{5.53}. The estimated range for its power-law exponent is $[3.44, 3.99]$.

\subsection*{Additional Figures}
In this section we present some additional figures. In the first two figures, we show how the (conditional) degree distribution evolves as $\sigma$ varies for Tunables Chung-Lu graphs and TGIRGs (with all other parameters being fixed). Fig. \ref{fig:caida}--\ref{fig:youtube} display these degree distributions as well as the heatmap for the joint degree distribution for all real-world networks which are power-law with an exponent between $2$ and $3$ according to the evaluation of~\cite{voitalov2019scale}.

\begin{figure}[ht!]
    \centering
    \includegraphics[width=0.45\linewidth]{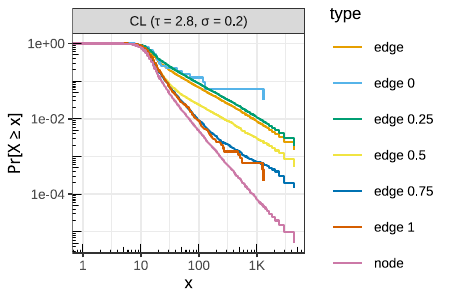}
    \includegraphics[width=0.45\linewidth]{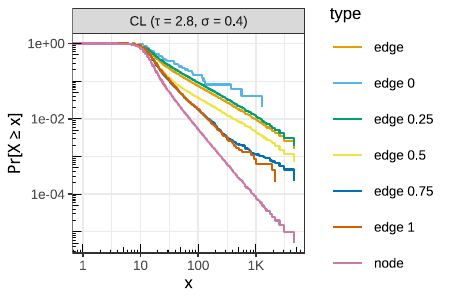}
    \\
    \includegraphics[width=0.45\linewidth]{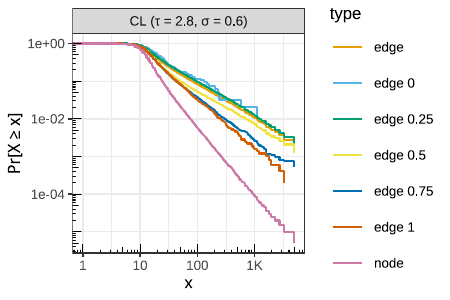}
     \includegraphics[width=0.45\linewidth]{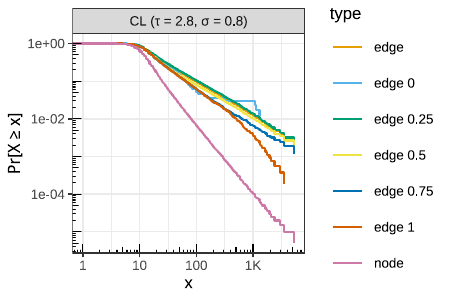}
     \\
    \includegraphics[width=0.45\linewidth]{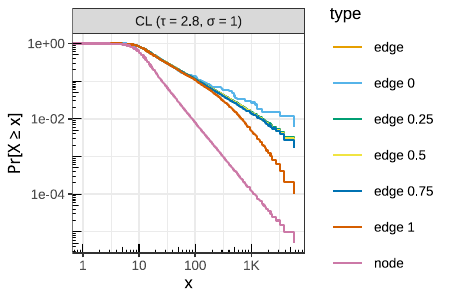}
    \includegraphics[width=0.45\linewidth]{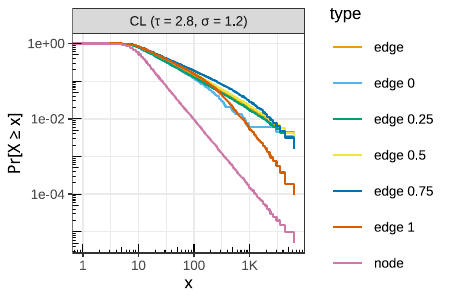}
    \\
     \includegraphics[width=0.45\linewidth]{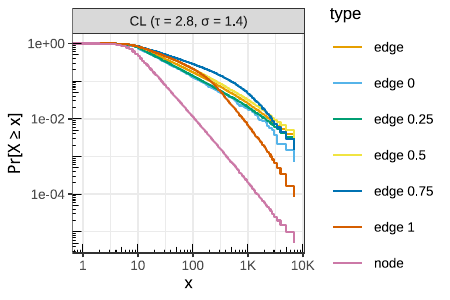}
    \includegraphics[width=0.45\linewidth]{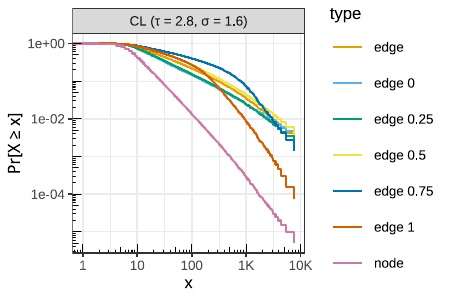}  
    
    \caption{Degree Distribution of Tunable Chung-Lu Graphs for $\sigma \in [0.2, 1.6]$. For all graphs, the number of nodes is \num{200000}, the average degree is \num{15} and the parameter $\tau$ is set to \num{2.8}.}
    \label{fig:degrdistr-cl}
\end{figure}

\begin{figure}[ht!]
    \centering
    \includegraphics[width=0.45\linewidth]{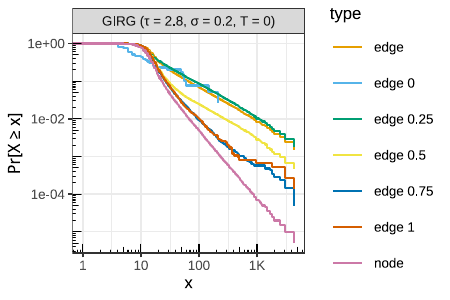}
    \includegraphics[width=0.45\linewidth]{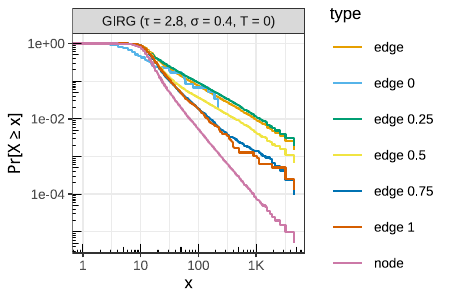}
    \\
    \includegraphics[width=0.45\linewidth]{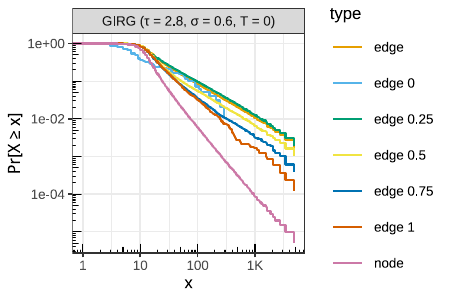}
     \includegraphics[width=0.45\linewidth]{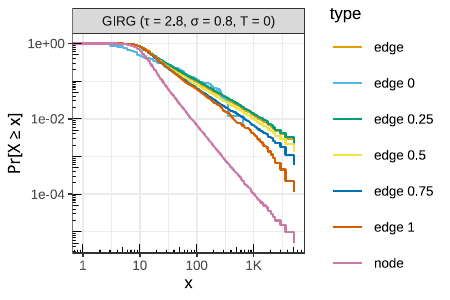}
     \\
    \includegraphics[width=0.45\linewidth]{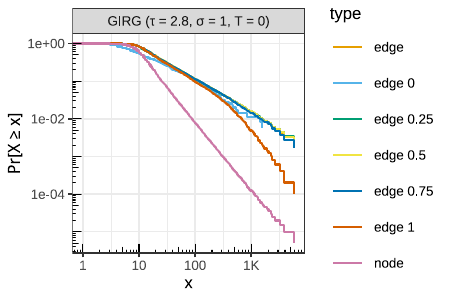}
    \includegraphics[width=0.45\linewidth]{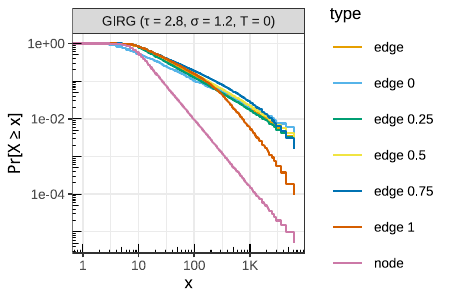}
    \\
     \includegraphics[width=0.45\linewidth]{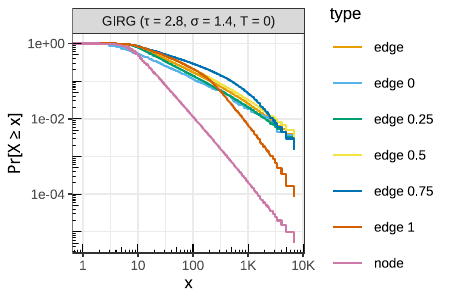}
    \includegraphics[width=0.45\linewidth]{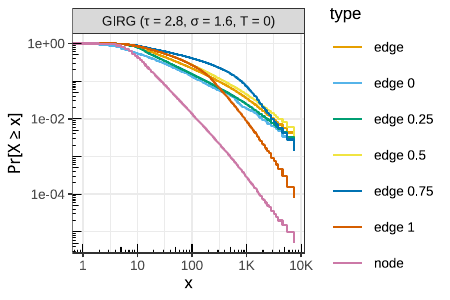}

    \caption{Degree Distribution of TGIRGs for $\sigma \in [0.2, 1.6]$. For all graphs, the number of nodes is \num{200000}, the average degree is \num{15} and the TGIRG parameters are set to $\tau=2.8$, $T=0$ and $d=2$.}
    \label{fig:degrdistr-girg}
\end{figure}

\begin{figure}[ht!]
    \centering
     \includegraphics[scale=0.6]{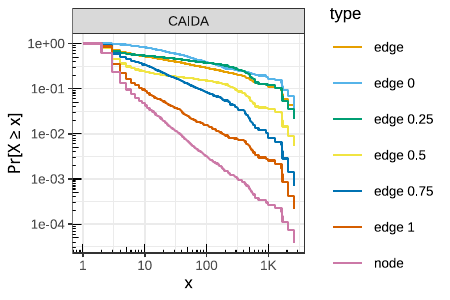}%
  \hfill%
  \includegraphics[scale=0.6]{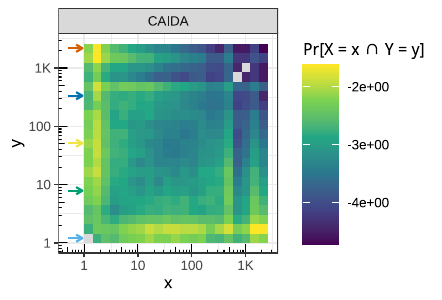}%
  \hfill%
  \includegraphics[scale=0.6]{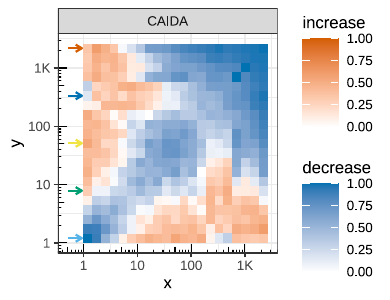}
    \caption{Conditional and Joint Degree Distributions of the \texttt{CAIDA} network.}
    \label{fig:caida}
\end{figure}

\begin{figure}[ht!]
    \centering
     \includegraphics[scale=0.6]{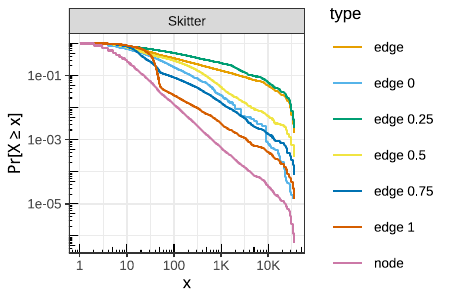}%
  \hfill%
  \includegraphics[scale=0.6]{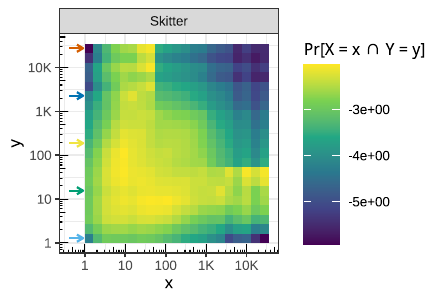}%
  \hfill%
  \includegraphics[scale=0.6]{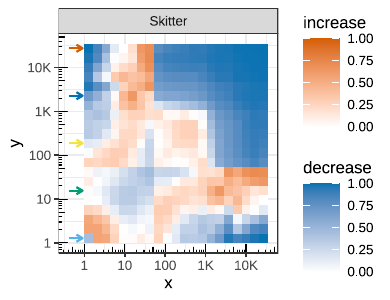}
    \caption{Conditional and Joint Degree Distributions of the \texttt{Skitter} network.}
    \label{fig:skitter}
\end{figure}

\begin{figure}[ht!]
    \centering
    \includegraphics[scale=0.6]{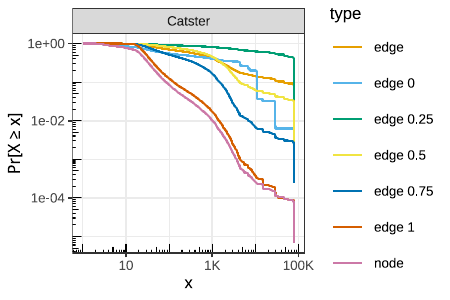}%
  \hfill%
  \includegraphics[scale=0.6]{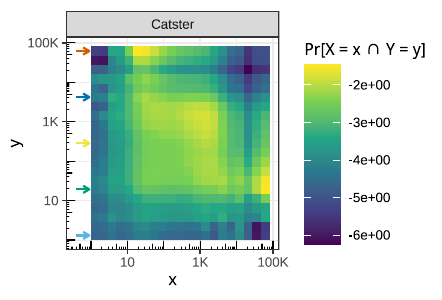}%
  \hfill%
  \includegraphics[scale=0.6]{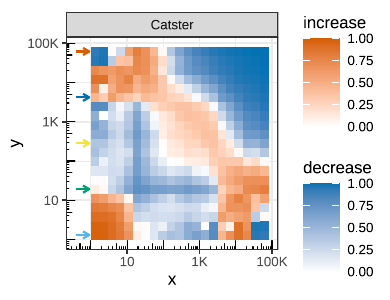}
    \caption{Conditional and Joint Degree Distributions of the \texttt{Catster} network.}
    \label{fig:cat}
\end{figure}

\begin{figure}[ht!]
    \centering
    \includegraphics[scale=0.6]{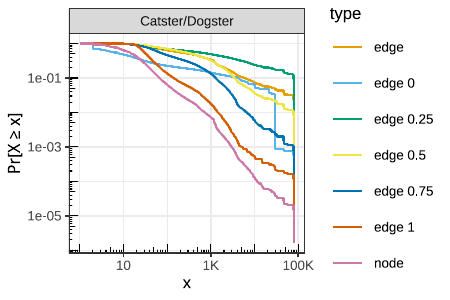}%
  \hfill%
  \includegraphics[scale=0.6]{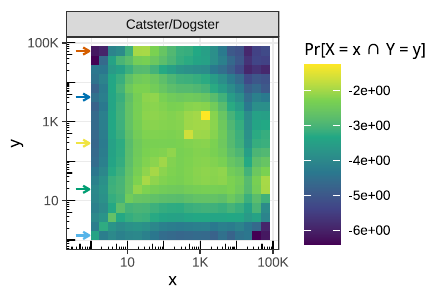}%
  \hfill%
  \includegraphics[scale=0.6]{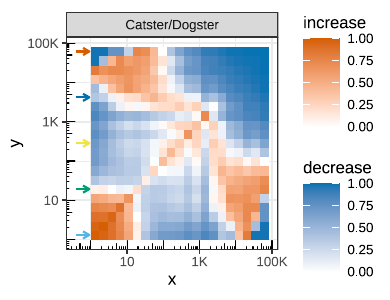}
    \caption{Conditional and Joint Degree Distributions of the \texttt{Catster/Dogster} network.}
    \label{fig:carnivore}
\end{figure}

\begin{figure}[ht!]
    \centering
     \includegraphics[scale=0.6]{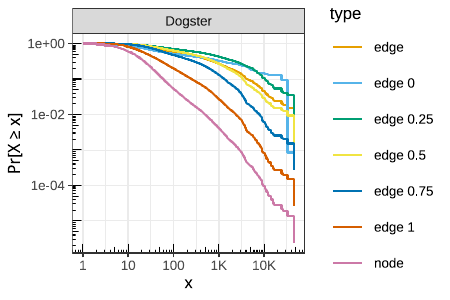}%
  \hfill%
  \includegraphics[scale=0.6]{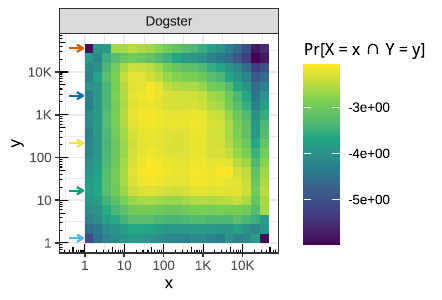}%
  \hfill%
  \includegraphics[scale=0.6]{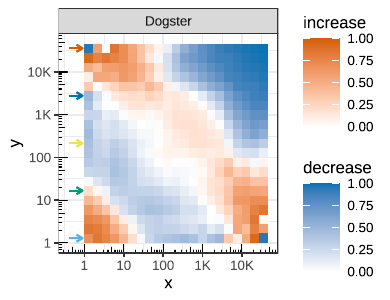}
    \caption{Conditional and Joint Degree Distributions of the \texttt{Dogster} network.}
    \label{fig:dog}
\end{figure}

\begin{figure}[ht!]
    \centering
    \includegraphics[scale=0.6]{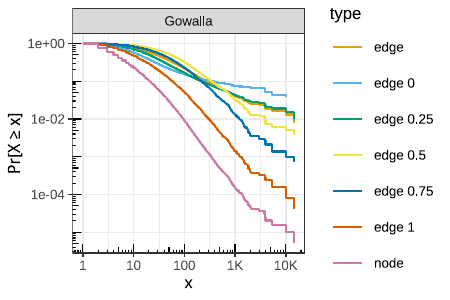}%
  \hfill%
  \includegraphics[scale=0.6]{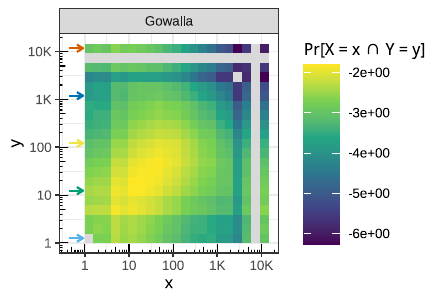}%
  \hfill%
  \includegraphics[scale=0.6]{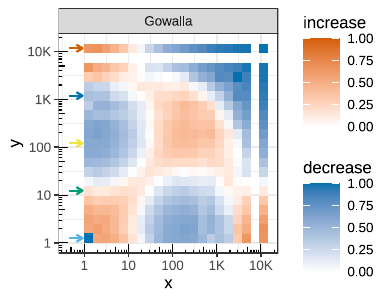}
    \caption{Conditional and Joint Degree Distributions of the \texttt{Gowalla} network.}
    \label{fig:gowalla}
\end{figure}

\begin{figure}[ht!]
    \centering
    \includegraphics[scale=0.6]{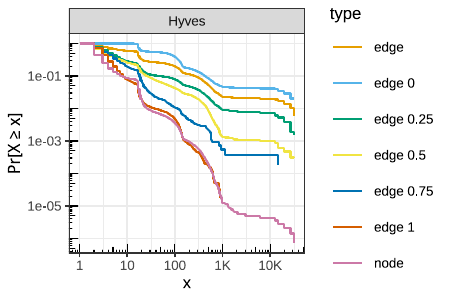}%
  \hfill%
  \includegraphics[scale=0.6]{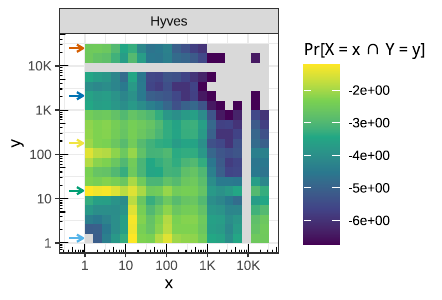}%
  \hfill%
  \includegraphics[scale=0.6]{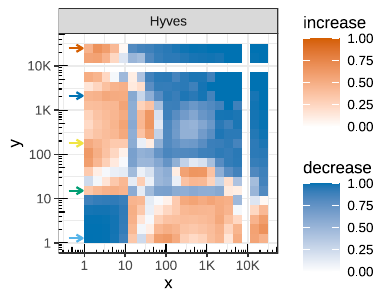}
    \caption{Conditional and Joint Degree Distributions of the \texttt{Hyves} network.}
    \label{fig:hyves}
\end{figure}

\begin{figure}[ht!]
    \centering
    \includegraphics[scale=0.6]{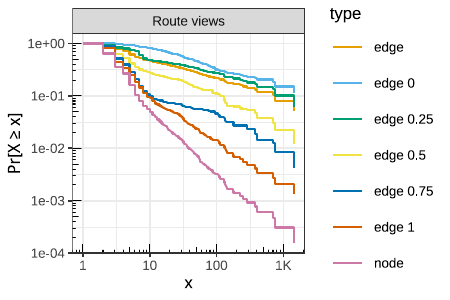}%
  \hfill%
  \includegraphics[scale=0.6]{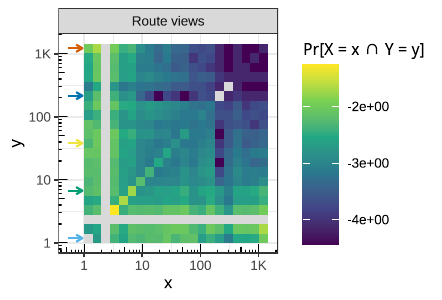}%
  \hfill%
  \includegraphics[scale=0.6]{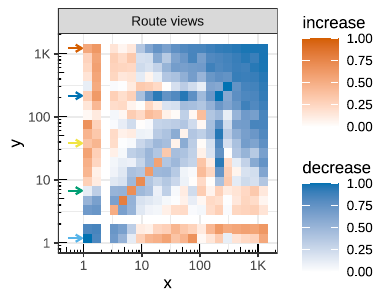}
    \caption{Conditional and Joint Degree Distributions of the \texttt{Route views} network.}
    \label{fig:route-views}
\end{figure}

\begin{figure}[ht!]
    \centering
    \includegraphics[scale=0.6]{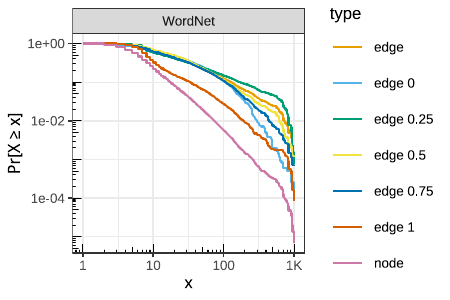}%
  \hfill%
  \includegraphics[scale=0.6]{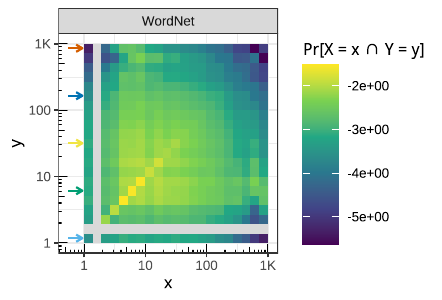}%
  \hfill%
  \includegraphics[scale=0.6]{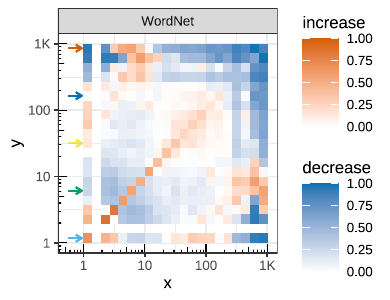}
    \caption{Conditional and Joint Degree Distributions of the \texttt{Wordnet} network.}
    \label{fig:wordnet}
\end{figure}

\begin{figure}[ht!]
    \centering
    \includegraphics[scale=0.6]{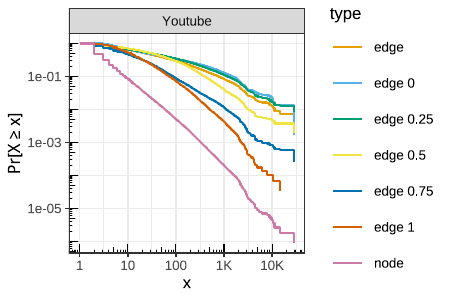}%
  \hfill%
  \includegraphics[scale=0.6]{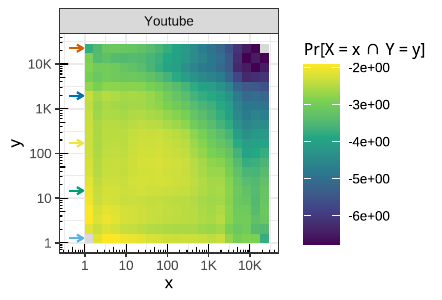}%
  \hfill%
  \includegraphics[scale=0.6]{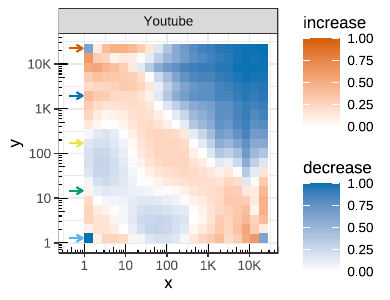}
    \caption{Conditional and Joint Degree Distributions of the \texttt{Youtube} network.}
    \label{fig:youtube}
\end{figure}

\subsection*{Generating TGIRGs}\label{sec:generating-tgirgs}
In this subsection, we describe how we modify the edge generation process in order to allow for the application of the GIRG generation algorithm (implemented by \cite{blasius2022efficiently}) to TGIRGs.
Our aim is to use this linear-time sampling algorithm for classical GIRGs as a black-box in order to generate TGIRGs in linear time. For this, we have to choose vertex weights such that the resulting GIRG is a supergraph of the desired TGIRG (i.e.\ the two graphs share the same vertex set, and the edge set of the TGIRG should be a subset of the edge set of the GIRG), and then perform an additional coin flip for each edge in order to correct the probability.

To achieve this, for all $v \in V$ denote by $w_v$  the weight of vertex $v$ in the TGIRG, and let $W := \sum_{v \in V} w_v$ be the sum of all vertex weights. Assuming $w_u \ge w_v$, in the TGIRG we have the connection probability
%
%
\begin{equation*}
  p_{uv} := \min \left\{ 1, \left( \frac{w_u
        \cdot w_v^\sigma}{W\cdot \rVert x_u - x_v\rVert^d}
    \right)^\alpha \right\}. 
\end{equation*}
For the GIRG we are going to generate, we use weights $w_v'$ with a total weight $W' := \sum_{v \in V} w_v'$ and we have connection probability
\begin{equation*}
  p_{uv}' := \min \left\{ 1, \left( \frac{w_u' \cdot w_v'}{W'\cdot
        \rVert x_u - x_v\rVert^d} \right)^\alpha \right\}.
\end{equation*}
Our goal is to choose weights $w_v'$ such that $p_{uv}' \ge p_{uv}$ for every pair of vertices $u, v \in V$. 
Thus, to show $p_{uv}' \ge p_{uv}$, it suffices to show that $p_{uv}' = 1$ or that
\begin{equation*}
  \frac{w_u' \cdot w_v'}{W'} \ge \frac{w_u \cdot w_v^\sigma}{W}.
\end{equation*}
Without loss of generality, we assume that $w_v \ge 1$ for all $v \in V$. If this is not the case, we can simply scale them, which yields a supergraph and then subsample with the correct weights.


\subsubsection*{Case $\sigma < 1$}

The case $\sigma < 1$ is easy. Let $w_{\min}$ be the minimum weight in the graph and set $w_v' := w_v \cdot w_{\min}^{\sigma - 1}$. As $\sigma < 1$, it holds that $w_{\min}^{\sigma - 1} \ge w_v^{\sigma - 1}$. Thus, we get
\begin{equation*}
  \frac{w_u' \cdot w_v'}{W'} = 
  \frac{w_u \cdot w_v \cdot w_{\min}^{\sigma - 1}}{W} \ge 
  \frac{w_u \cdot w_v \cdot w_{v}^{\sigma - 1}}{W} =
  \frac{w_u \cdot w_v^\sigma}{W}
\end{equation*}
as desired.

\subsubsection*{Case $\sigma > 1$}

The naive way to handle this case would be to scale all weights by the same factor, e.g.\ by $w_{\max}^{\sigma - 1}$. However, this would result in weights that are substantially larger than necessary, which yields a superlinear number of edges in the graph, for the two following reasons. First, if $w_u' \cdot w_v' \ge W'$ we have $p_{uv}' = 1$ and thus even larger weights do not have to be scaled. Secondly (and probably more importantly), the smaller weights are scaled more than necessary.

We deal with this as follows. We define auxiliary weights $w_v''$ as
\begin{equation*}
  w_v'' := w_v \cdot \min\{W^{1 / (\sigma + 1)}, w_v\}^{(\sigma - 1) / 2}.
\end{equation*}
With this, let $W'' := \sum_{v \in V} w_v''$. We claim that $w_v' := w_v'' \cdot \frac{W''}{W}$ yields weights such that $p_{uv}' \ge p_{uv}$.

To see this, first note that we get from $w''$ to $w'$ by scaling every weight with the same factor $\frac{W''}{W}$. Thus, $W' = W'' \cdot \frac{W''}{W}$. With this, we get
\begin{equation*}
  \frac{w_u' \cdot w_v'}{W'} =
  \frac{w_u'' \cdot \frac{W''}{W}  \cdot w_v'' \cdot
    \frac{W''}{W}}{W'' \cdot \frac{W''}{W}}
  = \frac{w_u'' \cdot  w_v''}{W}.
\end{equation*}
It remains to plug in the definition of $w_v''$, making a case distinction for the minimum. If $w_v > W^{1 / (\sigma + 1)}$, we get
\begin{equation*}
  \frac{w_u' \cdot w_v'}{W'} =
  \frac{w_u'' \cdot  w_v''}{W} \ge
  \frac{w_v''^2}{W} = 
  \frac{\left(w_v \cdot W^{\frac{\sigma - 1}{2(\sigma +
          1)}}\right)^2}{W} >
  \frac{\left(W^{\frac{1}{\sigma + 1}} \cdot W^{\frac{\sigma - 1}{2(\sigma +
          1)}}\right)^2}{W} = 1.
\end{equation*}
Thus, if $w_v > W^{1 / (\sigma + 1)}$ we get $p_{uv}' = 1 \ge p_{uv}$ and it remains to consider the case $w_v \le W^{1 / (\sigma + 1)}$. For this, we get $w_v'' = w_v \cdot w_v^{(\sigma - 1) / 2}$. Moreover, as $w_u \ge w_v$, we also have $w_u'' \ge w_u \cdot w_v^{(\sigma - 1) / 2}$ (note the $w_v$ instead of $w_u$ in the last factor). With this, we obtain 
\begin{equation*}
  \frac{w_u' \cdot w_v'}{W'} =
  \frac{w_u'' \cdot  w_v''}{W} \ge
  \frac{w_u \cdot w_v^{(\sigma - 1) / 2} \cdot w_v \cdot w_v^{(\sigma
      - 1) / 2}}{W} =
  \frac{w_u \cdot w_v \cdot w_v^{\sigma - 1}}{W} =
  \frac{w_u \cdot w_v^\sigma}{W},
\end{equation*}
which shows that $p_{uv}' \ge p_{uv}$ also holds in this case.

We still need to check that the number of edges in the supergraph (the GIRG with weights $(w_v')_{v\in V}$) is $\Theta(n)$, which is equivalent to checking that $\E[w_v']=\Theta(1)$. We start by computing 
\begin{align*}
    \E[w_v''] &= \E\left[w_v \cdot \min\{W^{\tfrac{1}{\sigma + 1}}, w_v\}^{\tfrac{\sigma-1}{2}}\right] \\
    &= \Theta\left( \int_1^{W^{\tfrac{1}{\sigma + 1}}} w_v \cdot w_v^{\tfrac{\sigma-1}{2}} \cdot w_v^{-\tau} dw_v + W^{\tfrac{\sigma-1}{2(\sigma+1)}} \int_{W^{\tfrac{1}{\sigma + 1}}}^{\infty} w_v^{1-\tau} dw_v \right) \\
    &= \Theta\left( \int_1^{W^{\tfrac{1}{\sigma + 1}}} w_v^{\tfrac{\sigma+1-2\tau}{2}} dw_v + W^{\tfrac{\sigma-1}{2(\sigma+1)}} \cdot W^{\tfrac{2-\tau}{\sigma+1}} \right) \\
    &= \Theta\left( \int_1^{W^{\tfrac{1}{\sigma + 1}}} w_v^{\tfrac{\sigma+1-2\tau}{2}} dw_v + W^{\tfrac{\sigma+3-2\tau}{2(\sigma+1)}} \right),
\end{align*}
where for the third equality we used that $\tau>2$. Now, note that since $\sigma<\tau-1$ (and $\tau>2$), we have $\sigma+3-2\tau<2-\tau<0$ (and equivalently $\tfrac{\sigma+1-2\tau}{2}<-1$), and hence
\[
\int_1^{W^{\tfrac{1}{\sigma + 1}}} w_v^{\tfrac{\sigma+1-2\tau}{2}} dw_v = \Theta(1)
\]
and 
\[
W^{\tfrac{\sigma+3-2\tau}{2(\sigma+1)}} = o(1).
\]
Therefore we get $\E[w_v'']=\Theta(1)$, which in particular implies that $\tfrac{W''}{W}=\Theta(1)$ and hence we also have $\E[w_v']=\Theta(1)$ as desired.

\subsubsection*{Average Degree}\label{sec:average-degree}

In order to approximate the desired average degree heuristically, for each parameter configuration, using a binary search approach, we iteratively generate four instances of TGIRGs and select the last generated instance.

\paragraph*{Acknowledgments}
M.K., J.L.\ and U.S.\ gratefully acknowledge support by the Swiss National Science Foundation [grant number
200021 192079].

T.B.\ gratefully acknowledges support by the Deutsche Forschungsgemeinschaft (DFG, German Research
Foundation) -- grant number 524989715.

\subsection*{Author contributions}
\textbf{Conceptualization:} Johannes Lengler.
\\
\textbf{Data Curation:} Thomas Bläsius.
\\
\textbf{Formal Analysis:} Marc Kaufmann, Ulysse Schaller, Thomas Bläsius and Johannes Lengler.
\\
\textbf{Funding Acquisition:} Johannes Lengler and Thomas Bläsius.
\\
\textbf{Investigation:} Marc Kaufmann, Ulysse Schaller, Thomas Bläsius and Johannes Lengler.
\\
\textbf{Methodology:} Marc Kaufmann, Ulysse Schaller, Thomas Bläsius and Johannes Lengler.
\\
\textbf{Project Administration:} Marc Kaufmann and Ulysse Schaller.
\\
\textbf{Software:} Thomas Bläsius.
\\
\textbf{Supervision:} Thomas Bläsius and Johannes Lengler.
\\
\textbf{Validation:} Thomas Bläsius.
\\
\textbf{Visualization:} Thomas Bläsius was the lead and implemented the visualizations, with all authors contributing to the visualization concepts.
\\
\textbf{Writing - Original Draft Preparation:} Marc Kaufmann and Ulysse Schaller.
\\
\textbf{Writing - Review \& Editing:} Marc Kaufmann, Ulysse Schaller, Thomas Bläsius and Johannes Lengler.


\bibliography{references}

\end{document}